\documentclass[12pt,a4paper]{article}
\usepackage[english]{babel}
\usepackage[whole]{bxcjkjatype}
\usepackage{parskip}

\usepackage[T1]{fontenc}
\usepackage[utf8]{inputenc}
\usepackage{lmodern}
\usepackage{parskip}

\usepackage{geometry}
\usepackage{verbatim}
\usepackage{float}
\usepackage{bm}
\usepackage{amsmath}
\usepackage{amssymb}
\usepackage{amsfonts}
\usepackage{amsthm}
\usepackage{mathtools}
\usepackage{graphicx}
\usepackage{dsfont}
\usepackage{array}
\usepackage{mathrsfs}
\usepackage{color}
\usepackage{enumitem}
\usepackage{authblk}
\usepackage{booktabs}

\numberwithin{equation}{section}
\mathtoolsset{showonlyrefs=true}

\usepackage{multirow}
\usepackage{makecell}
\usepackage{arydshln}
\usepackage[font=small]{caption}
\usepackage{subcaption}
\usepackage{tikz}
\usetikzlibrary{arrows.meta}

\usepackage{algorithm}
\usepackage{algpseudocode}

\RequirePackage[authoryear]{natbib}

\RequirePackage[colorlinks,citecolor=blue,urlcolor=blue]{hyperref}
\usepackage[dvipsnames]{xcolor}
\usepackage{cleveref}

\usepackage{pgfplots}
\pgfplotsset{compat=1.18}
\usepackage{appendix}

\allowdisplaybreaks

\theoremstyle{plain}
\newtheorem{theorem}{Theorem}[section]
\newtheorem{lemma}{Lemma}[section]
\newtheorem{proposition}{Proposition}[section]
\newtheorem{corollary}{Corollary}[section]
\newtheorem{assumption}{Assumption}[section]

\newtheorem{example}{Example}[section]

\theoremstyle{definition}
\newtheorem{remark}{Remark}[section]

\crefname{appendix}{appendix}{appendices}

\newcommand{\EE}{\mathbb{E}}

\newcommand{\PP}{\mathbb{P}}

\newcommand{\RR}{\mathbb{R}}

\newcommand{\cC}{\mathcal{C}}
\newcommand{\cD}{\mathcal{D}}

\newcommand{\cI}{\mathcal{I}}

\newcommand{\cK}{\mathcal{K}}

\newcommand{\cP}{\mathcal{P}}

\newcommand{\cR}{\mathcal{R}}
\newcommand{\cS}{\mathcal{S}}
\newcommand{\cT}{\mathcal{T}}

\newcommand{\dd}{\mathop{}\!\mathrm{d}}
\DeclareMathOperator{\Var}{\mathrm{Var}}

\DeclareMathOperator{\KL}{\mathrm{KL}}

\DeclareMathOperator{\Bern}{\mathrm{Bern}}

\DeclareMathOperator{\cov}{cov}

\newcommand{\1}{\mathbf{1}}

\newcommand{\expit}{\operatorname{expit}}
\newcommand{\logit}{\operatorname{logit}}
\newcommand{\NR}{\mathrm{nr}}

\newcommand{\obs}{\mathrm{obs}}
\newcommand{\Ext}{\mathrm{E}}

\newcommand{\distH}{d_{\mathrm{H}}}
\newcommand{\op}{o_{\PP}}
\newcommand{\Op}{O_{\PP}}

\title{
First-Crossing Reduction and Inference\\
for No-Rescue Effects under Deterministic Rescue
}
\author[1,2,$\ast$]{Shu Tamano}
\affil[1]{\small
Department of Multidisciplinary Sciences, Graduate School of Arts and Sciences, The University of Tokyo, 3-8-1 Komaba, Meguro-ku, Tokyo 153-8902, Japan
}
\affil[2]{
\small
Department of Epidemiology, National Institute of Infectious Diseases, Japan Institute for Health Security, 1-23-1 Toyama, Shinjuku-ku, Tokyo 162-0052, Japan
}
\affil[$\ast$]{\small
Email:
\href{mailto:tamano-shu212@g.ecc.u-tokyo.ac.jp}
{\texttt{tamano-shu212@g.ecc.u-tokyo.ac.jp}}
}
\date{}

\begin{document}
\maketitle

\begin{abstract}
    Clinical protocols may require rescue medication when a patient's condition crosses a prespecified threshold.
    The outcome under continued non-rescue is then unobserved after first crossing, preventing point identification without extrapolation assumptions.
    In this paper, we develop sharp partial identification and inference for this hypothetical estimand while retaining the original no-rescue intervention.
    We show that bounding this mean reduces exactly to restricting conditional mean effects of withholding rescue at first-crossing histories.
    Building on this representation, restrictions on prespecified crossing-stratum averages yield three results.
    (i) A finite convex program gives sharp bounds without modeling post-crossing trajectories, attained by full-data laws that preserve the observed distribution.
    (ii) Confidence intervals incorporate sampling uncertainty and numerical error to cover the entire sharp set, including strata with no observed crossings.
    (iii) Under local feasibility conditions, sampling and numerical errors control endpoint accuracy even with zero-probability strata, whereas valid coverage requires no such conditions.
    Our experiments show that the stratum-average restrictions can yield narrower confidence intervals in larger samples, and that accounting for uncertainty in empty strata maintains valid coverage when rescue is rare.
\end{abstract}


\section{Introduction}
\label{sec:introduction}

When a clinical protocol mandates rescue medication at a safety or severity threshold, continued non-rescue has conditional probability zero at the triggering history.
The mean outcome under a strategy that withholds rescue throughout follow-up can be a scientific target.
In the ICH E9(R1) framework, this target corresponds to a hypothetical strategy \citep{ich2019addendum,olarteparra2023hypothetical}.
It differs from a treatment-policy estimand that allows rescue.
However, no matter how flexible the outcome regression, it cannot restore the missing comparison, even in a randomized trial.
Therefore, an analysis retaining the no-rescue target must state what it assumes about the unobserved continuation \citep{robins1986a,petersen2012diagnosing}.

One approach to inadequate support is to change the intervention.
Stochastic interventions, longitudinal modified treatment policies, incremental propensity score interventions, and longitudinal flip interventions provide alternatives with different support requirements \citep{diazmunoz2012population,kennedy2019incremental,diaz2023nonparametric,mcclean2026longitudinal}.
For rescue medication specifically, \citet{michiels2021novel} propose an estimand that links rescue use under active treatment to rescue use under control.
These methods identify useful supported quantities, but do not identify the unsupported no-rescue intervention without further assumptions.

A second approach retains the target and specifies the unobserved outcomes through a sensitivity model.
Pattern-mixture models and controlled multiple imputation, including delta adjustments and reference-based assumptions, provide established frameworks for doing so \citep{little1993pattern,carpenter2013analysis,cro2020sensitivity}.
In the rescue setting, the sensitivity parameter is the average difference between the outcome under continued non-rescue and the observed outcome after rescue, given the history at which rescue was triggered.
A fixed value of this parameter plays the role of a delta adjustment and selects one extrapolation.
A class of values yields an identification region, and tipping-point analysis examines which values would change the scientific conclusion.

Partial identification retains uncertainty about the missing outcomes rather than selecting one extrapolation.
The outcome range provides basic bounds, which additional scientific restrictions can narrow \citep{manski1990nonparametric,horowitz2000nonparametric}.
The distinction between an identification region and a confidence region covering that entire region is also established in the literature on ignorance and uncertainty regions \citep{vansteelandt2006ignorance} and set inference \citep{chernozhukov2007estimation,romano2010inference}.
Related work derives sharp off-policy bounds under smoothness \citep{khan2024offpolicy}, longitudinal welfare orderings and bounds under instrumental-variable restrictions \citep{chen2023estimating,han2024optimal}, and sensitivity analyses for sequential unmeasured confounding and marginal structural models \citep{bonvini2022sensitivity,tan2025sensitivity}.
Sharp dynamic-regime bounds have also been studied under contemporaneous confounding and restrictions on state transitions \citep{alvarez2026sharp}.
These assumptions differ from restrictions on the effect of withholding a deterministic rescue course.
In particular, a sequential density-ratio sensitivity model cannot supply a comparison conditional on the complete trigger when the non-rescue probability is zero.
Coarsening the history can restore overlap, but changes which selection restrictions are required.

Longitudinal bounds must also preserve compatibility across visits.
Conditional means of the same potential outcome satisfy iterated expectation, and sharp extrema must be attainable under a single joint law of observed and potential variables.
Restrictions linking visits or patient groups preclude independent worst-case choices.
General discrete causal-bounding methods formulate polynomial programs and incorporate simultaneous confidence regions for observed probabilities \citep{duarte2024automated}.
However, explicit enumeration of longitudinal completions can be demanding, and estimated inputs introduce a further difficulty when important histories are rare or absent and optimized endpoints are sensitive to perturbations.

In this paper, we develop sharp partial identification and inference for no-rescue means and treatment contrasts under a deterministic rescue rule that, once triggered, remains in force.
We establish that, along histories up to first rescue, the conditional means of the final no-rescue outcome that are consistent with the observed data are characterized by the average effects of withholding rescue at the histories where rescue is first triggered.
These are precisely the sensitivity parameters introduced above.
Observed conditional expectations propagate them to earlier visits, so no model for the post-crossing course is required.
To obtain a finite and clinically interpretable sensitivity model, we group the histories at first crossing into prespecified strata, for example by crossing visit and baseline risk.
The model bounds the stratum-average effects and the probability-weighted sum of their squared deviations from reference effects, without requiring effects to be constant within a stratum.

Our contributions are threefold.
First, we construct a joint law realizing every specification of these parameters that respects the outcome range, in both treatment arms, while preserving the entire observed distribution.
This proves causal sharpness rather than only feasibility of conditional-mean equations.
Second, the stratum-average model reduces exactly to a finite convex program whose inputs are observed outcome means and stratum-level crossing probabilities and outcome contributions.
The program retains the joint restrictions and admits provable inner and outer numerical bounds without dividing by estimated stratum probabilities.
Third, projecting a simultaneous confidence region for these observed quantities through the program yields confidence intervals for the entire sharp set, including strata with no observed crossings.
Additional local feasibility conditions give endpoint-error bounds even with zero-probability strata.
These conditions are not needed for coverage.

The rest of the paper is organized as follows.
Section~\ref{sec:setting} defines the causal problem.
Sections~\ref{sec:geometry} and~\ref{sec:algorithm} establish the structural and computational reductions, and Section~\ref{sec:statistics} develops inference.
Sections~\ref{sec:numerical} and~\ref{sec:discussion} present the numerical study and discuss scope and limitations.
Appendix~\ref{app:proofs} contains proofs, and Appendix~\ref{app:detailed-numerical} gives detailed experimental methods and additional results.

\section{Problem setup}
\label{sec:setting}

This section formalizes the observed data and potential trajectories, and isolates the unobserved component of the no-rescue mean that requires sensitivity restrictions.

\subsection{Observed data and baseline identification}

We first specify the data structure and baseline identification conditions.
The sample contains $n$ independent and identically distributed patient records with law $P$ and generic record
\begin{equation*}
    O
    =
    (W,A,L_1,D_1,\ldots,L_K,D_K,Y)
    .
\end{equation*}
Here $W$ contains baseline covariates, $A\in\{0,1\}$ is baseline treatment, and $L_k$ is the state measured before the rescue decision at visit $k\in\{1,\ldots,K\}$.
The indicator $D_k\in\{0,1\}$ records whether rescue has been initiated by visit $k$, not an individual dose.
The final outcome $Y$ is measured after visit $K$ and satisfies $y_-\le Y\le y_+$ almost surely for known finite constants $y_-<y_+$.
The horizon $K\ge1$ is fixed and finite, all state spaces are nonempty standard Borel spaces, and within-patient dependence is unrestricted.

The rescue rule is formulated using the histories available before each decision.
Write $\bar{L}_k=(L_1,\ldots,L_k)$ and $\bar{D}_k=(D_1,\ldots,D_k)$, with empty vectors at $k=0$, and define the decision history as $X_k=(W,A,\bar{L}_{k-1},\bar{D}_{k-1},L_k)$.
Let $z_k$ be a known measurable function taking value $0$ at histories that trigger rescue and $1$ otherwise.
The protocol satisfies
\begin{equation}
    \label{eq:rule}
    T
    =
    \min\{k\in\{1,\ldots,K\}:z_k(X_k)=0\}
    ,
    \quad
    D_k
    =
    \1(T\le k)
    ,
\end{equation}
with $T=K+1$ if no trigger occurs and $\1$ denoting an indicator.
Therefore, rescue is absorbing, and trigger values after its initiation do not affect $T$.

Potential trajectories are introduced next to state the baseline identification conditions.
For $a\in\{0,1\}$, let $O^a=(W,a,L_1^a,D_1^a,\ldots,L_K^a,D_K^a,Y^a)$ denote the trajectory under baseline treatment $a$ and protocol rescue.
Let $\bar{L}_K^{a,\NR}=(L_1^{a,\NR},\ldots,L_K^{a,\NR})$ and $Y^{a,\NR}$ denote the state trajectory and final outcome under the same treatment with rescue withheld throughout follow-up.
The no-rescue outcome has the same known range.
Write $g_a(w)=P(A=a\mid W=w)$ for the baseline propensity.
\begin{assumption}[Baseline identification]
    \label{ass:baseline}
    Baseline consistency holds, $O=O^A$ almost surely.
    The joint collection of potential processes $(O^0,O^1,\bar{L}_K^{0,\NR},\bar{L}_K^{1,\NR},Y^{0,\NR},Y^{1,\NR})$ is independent of $A$ given $W$.
    For each $a$, $g_a(W)>0$ almost surely under the marginal law of $W$.
\end{assumption}
These conditions concern baseline treatment assignment.
Randomization supplies the baseline independence and known assignment probabilities, but does not supply non-rescue observations at triggering histories.

Baseline standardization identifies the treatment-specific protocol law.
Let $P_a$ be the law of $O^a$, retaining the marginal distribution of $W$ under $P$, and let $\EE_a$ denote expectation under $P_a$.
Writing $\EE$ for expectation under $P$, Assumption~\ref{ass:baseline} gives
\begin{equation}
    \label{eq:arm}
    \EE_a[f]
    =
    \EE\bigl[
        \omega_a f(O)
    \bigr]
    ,
    \quad
    \omega_a
    =
    \frac{\1(A=a)}{g_a(W)}
    ,
\end{equation}
for every measurable $f$ integrable under $P_a$.
In arm-specific expressions, unadorned trajectory symbols refer to $O^a$.
Expectations involving no-rescue variables refer to a candidate full-data extension of this standardized law.

\subsection{Temporal consistency}

We next state how the protocol and no-rescue trajectories agree before their rescue decisions diverge.
Applying the history and crossing-time definitions to $O^a$ gives $X_k^a$ and $T^a$.
\begin{assumption}[Temporal consistency]
    \label{ass:temporal}
    For each $a\in\{0,1\}$, almost surely,
    \begin{equation*}
        \bar{L}_k^{a,\NR}
        =
        \bar{L}_k^a
        \text{ on }\{T^a\ge k\}
        \quad
        (k=1,\ldots,K)
        ,
        \quad
        Y^{a,\NR}
        =
        Y^a
        \text{ on }\{T^a>K\}
        .
    \end{equation*}
\end{assumption}
The first condition requires the states to agree through the measurement that first triggers rescue, including that measurement itself;
the second requires the outcomes to agree when the protocol never initiates rescue.
Induction over visits shows that the no-rescue trajectory has the same first threshold crossing and the same crossing history as the protocol trajectory, although rescue is withheld there.
Their post-crossing states and outcomes need not agree.
Consequently, $P_a$ identifies the distribution of first-crossing histories without any model for the unobserved continuation.
Figure~\ref{fig:crossing} summarizes this distinction.

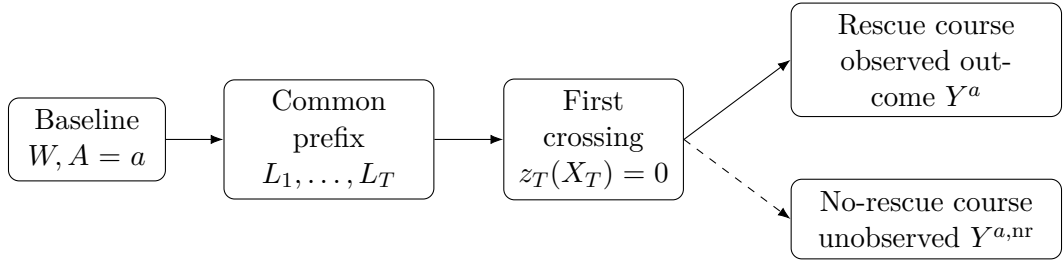
\begin{figure}[tb]
    \centering
    \begin{tikzpicture}[>=Latex, every node/.style={font=\small},
        box/.style={draw,rounded corners,align=center,minimum height=1.05cm}]
        \node[box,text width=1.8cm] (base) at (0,0) {Baseline\\$W,A=a$};
        \node[box,text width=2.5cm] (prefix) at (3.2,0) {Common prefix\\$L_1,\ldots,L_T$};
        \node[box,text width=2.1cm] (cross) at (6.7,0) {First crossing\\$z_T(X_T)=0$};
        \node[box,text width=3.3cm] (rescue) at (11.1,1.05) {Rescue course\\observed outcome $Y^a$};
        \node[box,text width=3.3cm] (nr) at (11.1,-1.05) {No-rescue course\\unobserved $Y^{a,\NR}$};
        \draw[->] (base)--(prefix);
        \draw[->] (prefix)--(cross);
        \draw[->] (cross.east)--(rescue.west);
        \draw[->,dashed] (cross.east)--(nr.west);
    \end{tikzpicture}
    \caption{
    The two strategies on a path with $T\le K$.
    The crossing time and measured prefix coincide by temporal consistency.
    The no-rescue continuation is missing after the decisions diverge, although the rescued-course outcome remains observed.
    When $T>K$, the final outcomes coincide as well.
    }
    \label{fig:crossing}
\end{figure}

The causal targets retain the intervention that withholds rescue throughout follow-up.
Define the standardized no-rescue mean and the treatment contrast by
\begin{equation}
    \label{eq:target}
    \theta_a
    =
    \EE[Y^{a,\NR}]
    ,
    \quad
    \Delta
    =
    \theta_1-\theta_0
    ,
\end{equation}
where the expectation is taken under a full-data law whose marginal distribution of $W$ is the observed one.
Let $\cP(P)$ denote the class of joint laws with observed marginal $P$ satisfying the setup and Assumptions~\ref{ass:baseline} and~\ref{ass:temporal}.
No sequential exchangeability assumption is imposed on rescue, and post-crossing no-rescue trajectories are otherwise unrestricted apart from the outcome range.

Without further restrictions, the outcome range gives the sharp interval for $\theta_a$
\begin{equation}
    \label{eq:range}
    \Bigl[
        \EE_a[Y\1(T>K)] + y_-P_a(T\le K)
        ,
        \quad
        \EE_a[Y\1(T>K)] + y_+P_a(T\le K)
    \Bigr]
    ,
\end{equation}
of width $(y_+-y_-)P_a(T\le K)$.
Consistency identifies the non-crossing contribution, and every value of the remainder is attainable, jointly across arms, by the construction in Lemma~\ref{lem:full-data-completion}.

\subsection{First-crossing effects}

We now parameterize the missing comparison by effects at first crossing.
For a fixed arm, define the observed rescued-course mean and the mean effect of withholding rescue at a first-crossing history as
\begin{equation}
    \label{eq:bridge}
    q_k(x)
    =
    \EE_a[Y\mid X_k=x,T=k]
    ,
    \quad
    \tau_k(x)
    =
    \EE_a[Y^{a,\NR}-Y\mid X_k=x,T=k]
    .
\end{equation}
We call $\tau_k$ a first-crossing effect, or bridge.
It compares continued non-rescue with the observed rescued course, rather than individual doses, and for an adverse outcome a positive value represents an average benefit from rescue at that history.
It is a delta-type conditional-mean sensitivity parameter whose reference is the observed rescued course \citep{little1993pattern,carpenter2013analysis,cro2020sensitivity};
the adjustment acts on the outcome scale, so for a binary outcome it is a risk difference rather than a log-odds shift.
Conditional means are defined almost everywhere under the corresponding crossing-history distribution, and a zero-probability crossing requires no specification.

These effects give an exact decomposition of the no-rescue mean.
Temporal consistency and total expectation imply
\begin{equation}
    \label{eq:causal-decomposition}
    \theta_a
    =
    \EE_a[Y]
    +
    \sum_{k=1}^K\EE_a\bigl[
        \tau_k(X_k)\1(T=k)
    \bigr]
    .
\end{equation}
The first term is identified from the observed law, whereas the sum isolates the unobserved no-rescue contribution.

\section{Structural reduction}
\label{sec:geometry}

This section shows that first-crossing effects determine the continuation means, are causally attainable, and yield sharp bounds that retain the declared restrictions.

\subsection{Compatible continuation means}

We first characterize how conditional means of the same final no-rescue outcome must agree across visits.
Fix an arm $a$ and define its continuation means and observed reference under a candidate full-data law as
\begin{equation*}
    m_k(x)
    =
    \EE_a[Y^{a,\NR}\mid X_k=x,T\ge k]
    ,
    \quad
    m_k^R(x)
    =
    \EE_a[Y\mid X_k=x,T\ge k]
    .
\end{equation*}
The event $\{T\ge k\}$ on which $m_k$ is defined is the visit-$k$ risk set:
patients not yet rescued before decision $k$.
Temporal consistency makes $X_k$ the history of either strategy before their decisions diverge, so $\EE_a[m_1(X_1)]=\theta_a$.

Observed transitions impose the non-crossing compatibility equations.
For $k<K$, define the observed averaging operator by
\begin{equation*}
    (P_kf)(x)
    =
    \EE_a[f(X_{k+1})\mid X_k=x,T>k]
    .
\end{equation*}
Iterated expectation and terminal outcome consistency give
\begin{equation}
    \label{eq:continuation}
    \begin{aligned}
        m_k(x)
        &=
        (P_km_{k+1})(x)
        &&\text{on }\{T>k\}
        ,
        \quad
        k<K
        ,\\
        m_K(x)
        &=
        \EE_a[Y\mid X_K=x,T>K]
        &&\text{on }\{T>K\}
        .
    \end{aligned}
\end{equation}
The observed reference $m^R$ satisfies the same equations and has crossing values $q_k$.
These equations average the same final outcome across visits.
They impose no comparison between rescue and non-rescue at triggering histories.

We represent deviations from the observed reference through their first-crossing values.
For crossing effects $\tau=(\tau_1,\ldots,\tau_K)$, define the backward extension $u=\Ext\tau$ on the risk sets by
\begin{equation}
    \label{eq:extension}
    \begin{split}
        u_K(x)
        &=
        \begin{cases}
            \tau_K(x),&z_K(x)=0,\\
            0,&z_K(x)=1,
        \end{cases}\\
        u_k(x)
        &=
        \begin{cases}
            \tau_k(x),&z_k(x)=0,\\
            (P_ku_{k+1})(x),&z_k(x)=1,
        \end{cases}
        \quad
        k<K
        .
    \end{split}
\end{equation}
A crossing effect changes the conditional no-rescue mean there, and observed transitions propagate that change backwards.
The zero terminal value leaves outcomes unchanged for patients who never receive rescue.

The known outcome range restricts the no-rescue means implied by a bridge.
A bridge is range compatible when
\begin{equation}
    \label{eq:compat}
    y_-
    \le
    q_k(x)+\tau_k(x)
    \le
    y_+
    \text{ almost everywhere on }\{T=k\}
    ,
    \quad
    k=1,\ldots,K
    .
\end{equation}
We call these range restrictions on the conditional no-rescue means the outcome caps.
The next result shows that they also suffice for a joint full-data law with the specified continuation means and unchanged observed distribution.

\begin{theorem}[First-crossing representation and causal attainability]
    \label{thm:trace}
    Under Assumptions~\ref{ass:baseline} and~\ref{ass:temporal}:
    \begin{enumerate}[label=(\roman*)]
        \item
        Every law in $\cP(P)$ has a range-compatible bridge and continuation $m=m^R+\Ext\tau$.
        For specified crossing values $q_k+\tau_k$, this is the unique sequence satisfying \eqref{eq:continuation}, up to equality almost everywhere on each risk set.

        \item
        Every range-compatible bridge is realized by a law in $\cP(P)$ whose continuation is $m^R+\Ext\tau$ and whose target is \eqref{eq:causal-decomposition}.
        Prescribed bridges in the two arms are jointly realizable while preserving the entire observed distribution and the common pre-crossing prefixes.
    \end{enumerate}
\end{theorem}
See Appendix~\ref{pf:thm-trace} for the proof.
Part (i) shows that first-crossing effects uniquely determine the compatible conditional-mean sequence through backward iterated expectation under the observed law.
Part (ii) establishes causal attainability while preserving the entire observed distribution, not merely the means appearing in the recursion.
Uniqueness concerns conditional means, not unobserved trajectories.
Different full-data laws can realize the same bridge.

\subsection{Sharp sensitivity bounds}

We next translate scientific restrictions on rescue effects into sharp no-rescue bounds.
Let $\cD$ be a nonempty declared class of range-compatible bridges, and restrict $\cP(P)$ to laws whose bridge belongs to $\cD$.
The class may restrict signs, magnitudes, variation, or specified effect averages.
These are scientific assumptions about unobserved no-rescue outcomes that the trial alone cannot validate.

\begin{proposition}[Sharp bounds]
    \label{prop:sharp}
    Under Assumptions~\ref{ass:baseline} and~\ref{ass:temporal}, the sharp identified set for the declared bridge class is
    \begin{equation}
        \label{eq:sharp-image}
        \Bigl\{
            \EE_a[Y]+\sum_{k=1}^K\EE_a\bigl[
                \tau_k(X_k)\1(T=k)
            \bigr]
            :\tau\in\cD
        \Bigr\}
        .
    \end{equation}
    If $\cD$ is convex, this set is an interval whose closure has endpoints given by the infimum and supremum of the displayed expression over $\cD$, and an endpoint belongs to the identified set exactly when its extremum is attained.
\end{proposition}
See Appendix~\ref{pf:prop-sharp} for the proof.
Because $\Ext\tau$ determines the entire continuation, restrictions on risk-set continuation means can be imposed on bridges by substituting $m^R+\Ext\tau$.
This substitution preserves the scientific class.
Replacing the induced restriction by a different norm on $\tau$ need not preserve it.
Outcome caps must likewise enter the optimization, because clipping an uncapped interval afterwards need not recover the capped sharp interval.

Sharp treatment-contrast bounds require joint attainability of the arm-specific choices.
For uncoupled arm classes with attained intervals $[L_a,U_a]$, the joint construction in Theorem~\ref{thm:trace} realizes every pair of arm means, so the sharp interval for $\Delta$ is $[L_1-U_0,\;U_1-L_0]$.
Shared parameters or cross-arm calibration must instead remain in a joint optimization, because independently optimized arm endpoints need not be jointly attainable.

\section{Computational reduction}
\label{sec:algorithm}

Building on the first-crossing representation in Section~\ref{sec:geometry}, this section derives an exact finite convex program for restrictions on prespecified crossing-stratum averages.

\subsection{Cell-average model}

We first place both outcomes on a common scale.
Rescale $Y$ and $Y^{a,\NR}$ to $[0,1]$, retaining their symbols.
An arm mean or endpoint $x$ returns to the original scale as $y_-+(y_+-y_-)x$, and effects and contrasts are multiplied by $y_+-y_-$.

Prespecified crossing cells summarize histories while retaining their contributions to the no-rescue mean.
Assign each first-crossing history $(k,x)$ to one of $M\ge1$ prespecified measurable strata, called cells, allowing empty cells.
Let $J\in\{1,\ldots,M\}$ identify the cell at first crossing, and set $J=0$ for no crossing.
For a fixed arm, define the observed moments and the unknown effect contributions as
\begin{equation}
    \label{eq:moments}
    \mu
    =
    \EE_a[Y]
    ,
    \quad
    p_j
    =
    P_a(J=j)
    ,
    \quad
    b_j
    =
    \EE_a[Y\1(J=j)]
    ,
    \quad
    t_j
    =
    \EE_a[(Y^{a,\NR}-Y)\1(J=j)]
    .
\end{equation}
Write $p,b,t\in\RR^M$ for the corresponding vectors.
For $p_j>0$, $t_j/p_j$ is the cell-average effect of withholding rescue.
When $p_j=0$, $t_j=0$ and no conditional mean is needed.
The target becomes
\begin{equation}
    \label{eq:mass-target}
    \theta_a
    =
    \mu+\sum_{j=1}^M t_j
    .
\end{equation}

Scientific restrictions bound the average effects and their joint departures from reference values.
Let $c\in[-1,1]^M$ contain reference effects and let $l,u\in[-1,1]^M$ give average-effect limits, with $l_j\le u_j$.
Let $H\in\RR^{r\times M}$ and $d\in\RR^r$ specify optional calibration equations, with $r=0$ indicating none.
Write $G=\Gamma^2$ for a sensitivity radius $\Gamma\ge0$.
Vector inequalities are componentwise.
\begin{assumption}[Prespecified cell-average model]
    \label{ass:cell-model}
    The partition and $c,l,u,H,G$ are fixed independently of the analysis data.
    The declared class consists of laws in $\cP(P)$ satisfying
    \begin{equation}
        \label{eq:cell-class}
        l_jp_j
        \le
        t_j
        \le
        u_jp_j
        ,
        \quad
        Ht
        =
        d
        ,
        \quad
        \sum_{j:p_j>0}\frac{(t_j-c_jp_j)^2}{p_j}
        \le
        G
        .
    \end{equation}
    The calibration target $d$ is fixed, or is a population quantity linked to an external source by a declared transport relation.
\end{assumption}
We refer to the three constraints in \eqref{eq:cell-class} as the effect limits, the calibration, and the budget.
The budget is shared across cells and bounds the probability-weighted sum of squared deviations of the cell-average effects $t_j/p_j$ from their reference effects $c_j$ by $G=\Gamma^2$.
These restrictions concern averages and do not require effects to be constant within cells.
For an adverse outcome, $l_j=0$ excludes an average harm from rescue in cell $j$.
It imposes neither individual-level monotonicity nor a nonnegative effect at every history within that cell.
External calibration must refer to the same continuation strategies and the same trial-population weights.

The budget has a root-mean-square (RMS) interpretation.
If $\pi_R=\sum_jp_j>0$, the final constraint in \eqref{eq:cell-class} is equivalent to
\begin{equation}
    \label{eq:rms}
    \Biggl\{
        \sum_{j:p_j>0}\frac{p_j}{\pi_R}\Biggl(
            \frac{t_j}{p_j}-c_j
        \Biggr)^2
    \Biggr\}^{1/2}
    \le
    \frac{\Gamma}{\sqrt{\pi_R}}
    ,
\end{equation}
whose left side is the conditional RMS deviation of cell-average effects among crossers.
Thus $\Gamma$ is neither a per-cell limit nor a bound on individual effects.
With $\pi_R=0.5$, a conditional radius of $0.20$, that is $20$ percentage points for a binary outcome, corresponds to $\Gamma=0.20\sqrt{0.5}$.
A conditional radius $\rho\ge0$ can instead be declared directly by replacing the final constraint with
\begin{equation}
    \label{eq:conditional-budget}
    \sum_{j:p_j>0}\frac{(t_j-c_jp_j)^2}{p_j}
    \le
    \rho^2\sum_jp_j
    .
\end{equation}
This constraint is jointly convex in $(p,t)$ and has both sides equal to zero when there are no crossers.
In confidence projection, the probability on its right side varies with the candidate population moments rather than being replaced by a fixed empirical estimate.
The results below are stated for fixed $G$ and extend to \eqref{eq:conditional-budget} as noted in Appendix~\ref{pf:prop-stability}.

A practical specification starts with the outcome scale, the rescue course being compared, and a small number of scientifically meaningful crossing groups, for example crossing visit and baseline risk.
Threshold exceedance can be added when it changes plausible rescue effects.
The limits $l_j,u_j$ express cell-average effect ranges, $c_j$ states the reference comparison, and $\Gamma$ or $\rho$ bounds their joint deviation.
These choices should be elicited before examining outcome-dependent sensitivity results, with cell counts reported afterwards.
A grid of radii and reference effects then supports tipping-point reporting while keeping scientific uncertainty distinct from confidence enlargement \citep{vansteelandt2006ignorance,cro2020sensitivity}.

\subsection{Finite-dimensional optimization}

We next combine the scientific restrictions with outcome caps in a finite convex program.
Let $\eta=(\mu,p,b,d)$ collect the primitive inputs.
The outcome range and the effect limits give the contribution bounds
\begin{equation}
    \label{eq:cell-box}
    L_j(p,b)
    =
    \max\{l_jp_j,-b_j\}
    ,
    \quad
    U_j(p,b)
    =
    \min\{u_jp_j,p_j-b_j\}
    .
\end{equation}
The cellwise outcome caps require the no-rescue outcome contribution $b_j+t_j$ to lie between zero and $p_j$.
The quadratic perspective represents the budget in contribution coordinates.
On $p\ge0$ and $|t|\le p$, define
\begin{equation}
    \label{eq:perspective}
    F_c(p,t)
    =
    \sum_{j=1}^M f_{c_j}(p_j,t_j)
    ,
    \quad
    f_c(p,t)
    =
    \begin{cases}
        (t-cp)^2/p,&p>0,\\
        0,&p=t=0.
    \end{cases}
\end{equation}
Other points are infeasible.
The contribution set is
\begin{equation}
    \label{eq:T}
    \cT(\eta,G)
    =
    \bigl\{
        t:
        L_j(p,b)\le t_j\le U_j(p,b),
        \ Ht=d,
        \ F_c(p,t)\le G
    \bigr\}
    ,
\end{equation}
and the target endpoints minimize and maximize \eqref{eq:mass-target} over it.

The perspective constraint has an equivalent second-order-cone representation.
With auxiliary variables $s_j\ge0$, it is equivalent to
\begin{equation}
    \label{eq:cone}
    (t_j-c_jp_j)^2
    \le
    p_js_j
    ,
    \quad
    \sum_j s_j
    \le
    G
    .
\end{equation}
This representation remains defined at zero probabilities (Lemma~\ref{lem:closed-perspective}).
Candidate observed moments must additionally respect probability and outcome coherence,
\begin{equation}
    \label{eq:coherence}
    p
    \ge
    0
    ,
    \quad
    \sum_j p_j
    \le
    1
    ,
    \quad
    0
    \le
    b_j
    \le
    p_j
    ,
    \quad
    0
    \le
    \mu-\sum_j b_j
    \le
    1-\sum_j p_j
    ,
\end{equation}
where the last inequality bounds the non-crossing outcome contribution.
Let $\cC_{\mathrm{coh}}$ denote these conditions together with known identities for redundant calibration coordinates.

Input regions let the same program give plug-in bounds or confidence intervals.
For a supplied region $\cC$, compute
\begin{equation}
    \label{eq:outer-program}
    \inf/\sup
    \Bigl\{
        \mu+\sum_j t_j:
        \eta\in\cC\cap\cC_{\mathrm{coh}},
        \ t\in\cT(\eta,G)
    \Bigr\}
    .
\end{equation}
A singleton gives plug-in endpoints.
A simultaneous confidence region gives an outer interval accounting for input uncertainty.
Box and second-order-cone input regions yield second-order-cone programs, and the likelihood regions of Section~\ref{sec:statistics} add exponential-cone constraints.
No conditional crossing-outcome regression and no inverse product of non-rescue probabilities is estimated.

\begin{example}[Two crossing cells]
    \label{ex:two-cells}
    Consider a binary adverse outcome with $\mu=0.40$, $p=(0.20,0.30)$, and $b=(0.04,0.27)$.
    The two observed crossing means are $0.20$ and $0.90$, and the non-crossing outcome contribution is $0.09$.
    Use $c=0$, $l=0$, $u=1$, no calibration, and $G=0.02$.
    Thus, the conditional root-mean-square cell effect is at most $0.20$.
    The range-only and sign-only intervals are $[0.09,0.59]$ and $[0.40,0.59]$.
    Under the budget, maximizing without cellwise caps would assign $t=(0.04,0.06)$ and give upper endpoint $0.50$.
    The second cap instead requires $t_2\le0.03$.
    The exact maximizing vector is
    \begin{equation*}
        t_2=0.03
        ,
        \quad
        t_1=\sqrt{0.20\{0.02-0.03^2/0.30\}}
        =
        \sqrt{0.0034}
        ,
    \end{equation*}
    giving sharp interval $[0.4000,0.4883]$ to four decimals.
    The first cell-average effect is about $0.292$, illustrating why a root-mean-square radius is not a per-cell cap.
    This is a numerical illustration of the restrictions, not a clinical calibration.
\end{example}

\subsection{Numerical approximation}

We construct inner and outer linear approximations to separate numerical error from sensitivity uncertainty.
The conic program can be solved directly or bracketed by affine tangents.
For $h>0$, let $\cR_h$ be an ordered grid spanning $[-1,1]$, containing both endpoints, with adjacent spacings at most $h$.
Define the tangent approximation and its error allowance as
\begin{equation}
    \label{eq:tangent}
    \begin{split}
        f_{c,h}(p,t)
        &=
        \max\Bigl[
            0,
            \ \max_{v\in\cR_h}
            \{2(v-c)t+(c^2-v^2)p\}
        \Bigr]
        ,\\
        F_{c,h}(p,t)
        &=
        \sum_j f_{c_j,h}(p_j,t_j)
        ,
        \quad
        e_h
        =
        h^2/4
        .
    \end{split}
\end{equation}
Replacing $F_c\le G$ by $F_{c,h}\le G$ defines the outer set $\cT_h^+(\eta,G)$, and using $F_{c,h}\le G-e_h$ defines the inner set $\cT_h^-(\eta,G)$.
Both retain all contribution bounds and calibration equations.
For $G>0$, a negative inner budget gives an empty inner set.
At $G=0$, use $t_j=c_jp_j$ exactly for both sets, retaining the other constraints.
Auxiliary variables represent the maxima by linear inequalities, so a polyhedral $\cC$ yields linear programs (LPs).

\begin{proposition}[Lossless reduction and numerical brackets]
    \label{prop:algorithm}
    Under Assumptions~\ref{ass:baseline}, \ref{ass:temporal}, and~\ref{ass:cell-model}, suppose the population class is nonempty.
    \begin{enumerate}[label=(\roman*)]
        \item
        At the true moments, $\cT(\eta,G)$ is exactly the attainable set of crossing contributions.
        Its target extrema are attained and are sharp causal endpoints.

        \item
        On $p\ge0$, $\sum_j p_j\le1$, and $|t_j|\le p_j$,
        \begin{equation}
            \label{eq:mesh-error}
            0
            \le
            F_c(p,t)-F_{c,h}(p,t)
            \le
            e_h\sum_j p_j
            \le
            e_h
            ,
        \end{equation}
        and
        \begin{equation}
            \label{eq:brackets}
            \cT_h^-(\eta,G)
            \subseteq
            \cT(\eta,G)
            \subseteq
            \cT_h^+(\eta,G)
            \subseteq
            \cT(\eta,G+e_h)
            .
        \end{equation}

        \item
        For a fixed closed $\cC$ with a nonempty exact feasible intersection, nested grids with $h\downarrow0$ give monotone convergence of the outer endpoints in \eqref{eq:outer-program} to the exact endpoints.
    \end{enumerate}
\end{proposition}
See Appendix~\ref{pf:prop-algorithm} for the proof.
Part (i) connects the finite program to causal attainability:
every feasible contribution vector lifts to a full-data law preserving the observed distribution.
Part (ii) controls the numerical relaxation uniformly over coherent inputs without a $1/p_j$ factor.
Part (iii) gives outer convergence without positive cell masses or strict feasibility.
The numerical allowance controls the budget approximation.
Endpoint-error rates and inner feasibility additionally use the local feasibility conditions in Section~\ref{sec:statistics}.
Refining $h$ changes numerical accuracy, whereas changing the partition or radius changes the scientific model.

Algorithm~\ref{alg:boundary} summarizes the computation.
\begin{algorithm}[tb]
    \caption{Sharp-set estimation and confidence projection}
    \label{alg:boundary}
    \begin{algorithmic}[1]
        \Require
        Prespecified trigger, outcome range, partition, $c,l,u,H,G$, and any external transport relation.

        \State
        Normalize outcomes;
        estimate the moment inputs and construct a simultaneous region $\cC_n$.

        \State
        Form coherent plug-in and confidence programs using \eqref{eq:cell-box}--\eqref{eq:coherence}.

        \State
        Compute conic endpoints or tangent-grid brackets;
        use $t_j=c_jp_j$ for every cell at $G=0$.

        \State
        Obtain an outward lower bound on the minimizing optimum and an outward upper bound on the maximizing optimum;
        check inner witnesses against the original constraints.

        \State
        Refine numerical accuracy as needed;
        distinguish infeasibility from numerical failure, and do not increase the scientific budget to force feasibility.

        \State
        Report sharp-set estimates, confidence enlargement, numerical gaps, cell counts, and sensitivity specifications.

        \State
        For an empty confidence program or an unavailable outward certificate, report $[0,1]$ and record the reason.
    \end{algorithmic}
\end{algorithm}

\section{Statistical inference}
\label{sec:statistics}

This section constructs confidence intervals that cover the entire sharp interval and gives separate conditions under which their endpoints are accurate.

\subsection{Confidence projection}

We first transfer uncertainty about the observed moments through the reduced program.
For coherent $\eta$ with nonempty $\cT(\eta,G)$, let $\cI(\eta,G)=[L(\eta,G),U(\eta,G)]$ denote the interval obtained by minimizing and maximizing \eqref{eq:mass-target} over $\cT(\eta,G)$.
At the population moments, Proposition~\ref{prop:algorithm} identifies it as the sharp causal interval.
Let $\cC_n$ be a closed random region for $\eta$ whose feasibility events and optimization values are measurable.
The box and likelihood regions below have these properties, and the coherence constraints, outcome caps, and $d=Ht$ make every feasible intersection compact even when $\cC_n$ is unbounded.

The reported interval lets the inputs vary over their joint region rather than treating their estimates as known.
Define $\cI_n^{\mathrm{out}}$ as \eqref{eq:outer-program} over $\cC_n$, or an outward numerical enclosure of it, intersected with $[0,1]$.
For an empty program or an unavailable outward certificate, report the physical-range fallback $[0,1]$ and record the reason.
Let $\PP$ denote joint sampling probability across the data sources and $\alpha\in(0,1)$ the nominal error probability.
\begin{corollary}[Whole-set confidence projection]
    \label{cor:coverage}
    Under Assumptions~\ref{ass:baseline}, \ref{ass:temporal}, and~\ref{ass:cell-model}, suppose the population class is nonempty and $\cC_n$ satisfies the preceding conditions.
    Then
    \begin{equation}
        \label{eq:coverage}
        \PP\{\cI(\eta,G)\subseteq\cI_n^{\mathrm{out}}\}
        \ge
        \PP(\eta\in\cC_n)
        .
    \end{equation}
\end{corollary}
See Appendix~\ref{pf:cor-coverage} for the proof.
On $\{\eta\in\cC_n\}$, every population-feasible contribution vector is feasible in the confidence program, so the program's extrema enclose the whole sharp interval.
Thus $1-\alpha$ coverage of the input region transfers to whole-set coverage at the same level.
This applies the projection principle for a confidence region covering an identification region to the reduced program \citep{horowitz2000nonparametric,vansteelandt2006ignorance,chernozhukov2007estimation,duarte2024automated}.
The coverage target is the entire sharp interval rather than one unidentified value \citep{imbens2004confidence}.
Finite-sample, asymptotic, and uniform input guarantees transfer with their respective qualifications.
The same event gives simultaneous coverage over any prespecified finite or countable grid of centers and radii with nonempty population classes.
No endpoint differentiability or local feasibility conditions of the type in Assumptions~\ref{ass:anchor} and~\ref{ass:anchor-general} are required, including at zero cell masses, zero budgets, or binding constraints.

\subsection{Simultaneous input regions}

Two families of regions supply $\cC_n$:
coordinatewise bounded-moment regions, which are valid in finite samples feature by feature, and joint multinomial likelihood regions, which avoid allocating error across features.

Coordinatewise regions use bounded trial features and external scores.
For $q=1,\ldots,D$, let $V_{q,1},\ldots,V_{q,n_q}$ be observations of a primitive, aggregate, or external feature whose mean $\zeta_q$ is a coordinate of $\eta$ or a specified linear function of it.
\begin{assumption}[Bounded moment sampling]
    \label{ass:bounded-moments}
    Within each feature, observations are independent and identically distributed in a known interval of finite length $B_q$, with $n_q\ge2$.
    Dependence across features is allowed.
    Pooled trial features use known baseline probabilities bounded away from zero, or independent samples from $P_a$ use unweighted features.
    External scores have the declared calibration mean under the transport relation.
\end{assumption}
For a trial with known assignment probabilities, the features are $\omega_aY$, $\omega_a\1(J=j)$, and $\omega_aY\1(J=j)$ for $j=1,\ldots,M$.
We also use the aggregates $\omega_a\1(T\le K)$ and $\omega_aY\1(T>K)$.
These bound $\sum_jp_j$ and $\mu-\sum_jb_j$ directly and retain information that summing cellwise limits loses.
With independent samples from $P_a$, the weights are omitted and all trial features lie in $[0,1]$.
They are Bernoulli when $Y$ is binary.
Weighted sample moments need not satisfy \eqref{eq:coherence}, and Appendix~\ref{pf:prop-stability} describes the projection onto the coherence polytope used in that case.
The baseline probability bound controls the range of weighted features.
It does not concern the probability of non-rescue after triggering.
For sample mean $\bar{V}_q$ and unbiased sample variance $\hat{v}_q$, the empirical Bernstein (EB) half-widths are
\begin{equation}
    \label{eq:eb}
    w_q^{\mathrm{EB}}
    =
    \sqrt{\frac{2\hat{v}_q\log(4D/\alpha)}{n_q}}
    +
    \frac{7B_q\log(4D/\alpha)}{3(n_q-1)}
    .
\end{equation}
Applying \citet[Theorem~4]{maurer2009empirical} to both signs and taking a union bound over the $D$ features gives simultaneous coverage at least $1-\alpha$ for $|\zeta_q-\bar{V}_q|\le w_q^{\mathrm{EB}}$.
Hoeffding limits and, for Bernoulli features, Clopper--Pearson (CP) limits with coordinate error $\alpha/D$ give finite-sample alternatives \citep{clopper1934use,hoeffding1963probability}.
Appendix~\ref{pf:cor-coverage} states them with their boundary conventions.
Empty sample cells retain uncertainty through their probability limits: a zero count does not imply zero population mass, and no conditional-effect estimate is formed inside the cell.
For a population-zero cell, the Bernstein probability width has order $\log(D/\alpha)/n_q$ rather than the Hoeffding order $\sqrt{\log(D/\alpha)/n_q}$, and binomial inversion can be sharper for binary features.
The coordinatewise construction allocates error across the $D$ features, which can be conservative when $D$ is large.

Joint regions avoid this allocation for binary independent-arm data.
Let $\pi_{jy}=P_a(J=j,Y=y)$ for $j=0,\ldots,M$ and $y\in\{0,1\}$.
The vector $\pi$ lies in the $d_\pi=2(M+1)$-atom probability simplex and determines $(\mu,p,b)$ linearly.
Split the sample independently of its values into two nonempty parts with counts $N_{v,jy}$ and sizes $n_v$, $v=1,2$, and define the smoothed training probabilities $\tilde{\pi}_{-v,jy}=(N_{3-v,jy}+1/2)/(n_{3-v}+d_\pi/2)$.
The cross-fit split likelihood-ratio region is
\begin{equation}
    \label{eq:joint-split}
    \cR_n^{\mathrm{split}}
    =
    \Biggl\{
        \pi:
        \frac{1}{2}\sum_{v=1}^2
        \exp\Biggl[
            \sum_{j,y}N_{v,jy}
            \log\Biggl(
                \frac{\tilde{\pi}_{-v,jy}}{\pi_{jy}}
            \Biggr)
        \Biggr]
        \le\frac{1}{\alpha}
    \Biggr\}
    ,
\end{equation}
where zero-count terms equal zero and a positive count at a zero candidate probability gives an infinite term.
This is the universal-inference construction of \citet{wasserman2020universal} specialized to the multinomial law of $(J,Y)$.
The smoothing enters only the training distribution and assigns no population mass to absent atoms.
Lemma~\ref{lem:joint-region} shows that the region is closed and convex with an exponential-cone representation.
It has finite-sample coverage at least $1-\alpha$ for every $\pi$, including vectors with zero coordinates, and projects through the capped program as in Corollary~\ref{cor:coverage}.
The usual multinomial likelihood-ratio (LR) region uses the Kullback--Leibler (KL) divergence,
\begin{equation}
    \label{eq:joint-lr}
    2n\KL(\hat{\pi}\|\pi)
    \le
    \chi^2_{d_\pi-1,1-\alpha}
    ,
    \quad
    \KL(v\|w)
    =
    \sum_iv_i\log(v_i/w_i)
    ,
\end{equation}
where $\hat{\pi}$ contains the empirical atom frequencies and $\chi^2_{\nu,1-\alpha}$ is the $(1-\alpha)$ quantile of the chi-square distribution with $\nu$ degrees of freedom.
We use $0\log(0/w_i)=0$, including when $w_i=0$, and a positive numerator with zero denominator gives an infinite term.
The LR region is a regular-law comparator:
its chi-square calibration is asymptotically valid at fixed laws with all declared atoms positive \citep[Chapter~16]{vandervaart1998asymptotic}.
Only known structural zeros may be removed in advance.
Projecting either joint region is equivalent to inverting a test obtained by profiling its defining criterion over $(\pi,t)$ (Appendix~\ref{pf:cor-coverage}).
This connects the computation to the set-inference literature \citep{chernozhukov2007estimation,romano2010inference}.
Section~\ref{sec:numerical} compares the resulting interval widths.

When baseline propensities are estimated from a pooled sample of size $N$, fitted-weight scores are generally dependent and need not have the target moments as their means.
Augmented inverse-probability scores that combine a propensity fit with a regression of each observed feature on $W$ give root-$N$ moment inference under cross-fitting and product-rate conditions.
A correctly specified logistic fit admits a sandwich variance.
Appendices~\ref{app:estimated-ps} and~\ref{app:augmented-moments} give the logistic sandwich calculation and augmented expansion, respectively \citep{bang2005doubly,chernozhukov2018double}.
The augmentation concerns observed moments only and does not extrapolate $Y^{a,\NR}$.

\subsection{Endpoint stability}

Coverage follows from inclusion alone.
This subsection gives local feasibility conditions under which improved moment precision also yields accurate endpoints, even with zero-probability cells.
These conditions provide positive budget slack and, with calibration, interior margins in selected contribution bounds.
Input distances use $\ell^1$, and matrix norms are the induced $1\to1$ norms.

Local stability requires a feasible effect configuration with room to accommodate input perturbations.
\begin{assumption}[Budget slack]
    \label{ass:anchor}
    Fix $M,c,l,u$ and coherent $\eta$, with $-1\le l_j\le0\le u_j\le1$ and no calibration.
    There exist $t^\circ\in\cT(\eta,G)$ and $\sigma>0$ with $F_c(p,t^\circ)\le G-\sigma$.
\end{assumption}
The anchor $t^\circ$ is a feasible effect configuration, not necessarily the true one, whose unused budget of at least $\sigma$ permits interpolation.
Zero lies in every coherent contribution interval under the stated effect limits, so with $c=0$ and $G>0$ the choice $t^\circ=0$ satisfies the assumption for every coherent $(p,b)$, with no lower bound on any cell mass.
With calibration equations $Ht=d$, corrections must also be allocated to cells with interior margins.
Appendix~\ref{pf:prop-stability} states this condition as Assumption~\ref{ass:anchor-general} and proves Proposition~\ref{prop:stability} in that generality.

Endpoint accuracy is measured by the larger of the two endpoint errors, the Hausdorff distance $\distH([L,U],[L',U'])=\max\{|L-L'|,|U-U'|\}$.
Let $\hat{\cI}_{n,h}$ be the exactly optimized outer-grid interval at a coherent plug-in $\tilde{\eta}_n$, and $\cI_{n,h}^{\mathrm{out}}$ its projection over $\cC_n$.
Both use $[0,1]$ for empty programs.

\begin{proposition}[Endpoint stability and error separation]
    \label{prop:stability}
    Under Assumptions~\ref{ass:baseline}, \ref{ass:temporal}, \ref{ass:cell-model}, and~\ref{ass:anchor}:
    \begin{enumerate}[label=(\roman*)]
        \item
        Every coherent $(\eta',G')$ sufficiently close to $(\eta,G)$ has a nonempty contribution set, and for a finite constant $C_{\ast}$,
        \begin{equation}
            \label{eq:lipschitz}
            \distH\{\cI(\eta',G'),\cI(\eta,G)\}
            \le
            C_{\ast}\{\|\eta'-\eta\|_1+|G'-G|\}
            .
        \end{equation}

        \item
        Suppose $\cC_n$ satisfies the conditions of Corollary~\ref{cor:coverage}, contains $\tilde{\eta}_n$ almost surely, and, for deterministic $r_n\to0$,
        \begin{equation}
            \label{eq:shrinking-region}
            \|\tilde{\eta}_n-\eta\|_1
            +
            \sup_{\eta'\in\cC_n\cap\cC_{\mathrm{coh}}}
            \|\eta'-\eta\|_1
            =
            \Op(r_n)
            .
        \end{equation}
        Then, for deterministic $h_n\to0$,
        \begin{equation}
            \label{eq:rate}
            \distH\{\hat{\cI}_{n,h_n},\cI(\eta,G)\}
            +
            \distH\{\cI_{n,h_n}^{\mathrm{out}},\cI(\eta,G)\}
            =
            \Op(r_n+h_n^2)
            .
        \end{equation}
    \end{enumerate}
\end{proposition}
See Appendix~\ref{pf:prop-stability} for the proof.
Part (i) controls how perturbations of the moment inputs and the budget move the endpoints.
Part (ii) separates the resulting statistical error from the numerical error of the tangent grid.
Under its conditions, the probability of an empty program tends to zero.
Implemented endpoints additionally incur outward-certificate gaps.
The same rate holds when those gaps are $\Op(r_n+h_n^2)$ and numerical-failure probabilities vanish.

The outcome caps explain why small cell probabilities need not destabilize the budget.
On $|t_j|\le p_j$ and $|t_j'|\le p_j'$,
\begin{equation}
    \label{eq:perspective-lip}
    |F_c(p,t)-F_c(p',t')|
    \le
    4\|t-t'\|_1+2\|p-p'\|_1
    .
\end{equation}
The caps force each contribution to vanish with its probability, so the bound contains no inverse-probability factor.
Budget slack permits feasible interpolation, and Assumption~\ref{ass:anchor-general} additionally permits calibration correction.
Without budget slack, endpoint changes can have square-root order even with positive cell masses (Remark~\ref{rem:feasibility-boundary}).
Corollary~\ref{cor:coverage} still applies there.

Bounded moment sampling supplies the rate in \eqref{eq:shrinking-region}.
Under Assumption~\ref{ass:bounded-moments}, fixed dimension, fixed $\alpha$, and diverging source sizes, Bernstein regions enlarged to contain the coherent plug-in satisfy \eqref{eq:shrinking-region}, with $r_n=O(n^{-1/2})$ for comparable source sizes.
Lemma~\ref{lem:primitive-rates} gives the source-specific rate, including zero-probability cells, and the fixed-law CP counterpart.
Endpoint convergence does not remove the scientific uncertainty represented by the limiting width $U-L$.
Persistent external uncertainty or a shrinking anchor margin can prevent contraction without invalidating coverage.

Treatment contrasts use the same projection.
A joint input region supports direct optimization of $\theta_1-\theta_0$ with any cross-arm restrictions retained.
For uncoupled classes, arm intervals with error allocations summing to $\alpha$ can be combined, with normalized fallback $[-1,1]$.
Their endpoint-error rates add when both arms satisfy Proposition~\ref{prop:stability}.

\section{Numerical experiments}
\label{sec:numerical}

This section evaluates the finite reduction, confidence-interval construction, and sensitivity to the declared effect restrictions.
We distinguish changes in the population sharp interval from differences between inferential procedures for the same class.
Rare-cell and active-cap experiments examine loss of empirical support and retention of cellwise restrictions, respectively.
Appendix~\ref{app:detailed-numerical} specifies the generating laws, comparators, and additional designs.

\subsection{Design}
\label{sec:num-design}

We first specify the reference law, scientific classes, and performance measures.
The reference law has $K=5$ visits, $s=4$ non-triggering states, no baseline covariates, and binary adverse outcomes.
Reaching state $s$ triggers absorbing rescue.
The transition intercept is chosen for an equally arm-averaged rescue probability of $0.5$, or $0.75$ where stated.
At crossing visit $k$, the rescued-course mean is $0.20+0.08k/K-0.04a$, and the no-rescue mean exceeds it by $\tau_k=0.08+0.12(K-k)/K$.
Non-crossers have the same outcome under both strategies.
Independent samples of $n$ patients are drawn from each standardized arm law.
Appendix~\ref{app:finite-protocol} gives the complete transition and outcome distributions.

The reference sensitivity class uses one cell per crossing visit, $c=0$, $l=0$, $u=1$, and $\Gamma=0.2$.
We call this the signed-budget class.
An oracle-calibrated version additionally fixes $t_1+t_3$ and $t_5$ at their population values.
Scientific comparisons include range-only and sign-only restrictions, a signed $\delta$-box with $\delta=0.2$, and a binary time-only density-ratio specialization of \citet{tan2025sensitivity}.
Each class is assessed against its own sharp interval.
The parameters $\delta$, $\Gamma$, and the ratio bound $\lambda$ constrain different quantities and are not treated as interchangeable sensitivity scales.

Inferential comparisons hold the signed-budget class fixed.
The default implementation, Cell CP, projects coordinatewise Clopper--Pearson limits for the $2M+3$ features defined in Appendix~\ref{app:inference-comparators}.
Other coordinatewise constructions use Bernoulli--KL, empirical Bernstein, Hoeffding, or Wald limits for the same features.
We also compare joint split-likelihood and likelihood-ratio regions, Aggregate CP, and an endpoint bootstrap.
Aggregate CP retains a population relaxation of the sharp interval.
The remaining procedures optimize the declared cellwise class.
The validity conditions and numerical implementations are specified in Appendix~\ref{app:inference-comparators}.

The primary measures are whole-set coverage, reported interval width, and plug-in Hausdorff error.
Whole-set coverage requires containment of both population endpoints.
Enlargement is reported width minus the sharp width of the declared class and can be negative when containment fails.
For sharp-model procedures it reflects sampling and numerical effects.
For Aggregate CP it also includes the population relaxation gap.
Target coverage records containment of the generating no-rescue mean and is reported separately from set coverage.
Except for the tail conventions of the one-probability diagnostic, the nominal whole-set level is $0.95$.
Contrasts allocate error $0.025$ to each arm.
Finite-state references use exact recursion and analytic or water-filling calculations, whereas continuous-history references use independent simulation.
Procedures compared within an experiment share datasets, and all physical-range fallbacks remain in the summaries.
We report Monte Carlo standard errors (MCSEs) for mean summaries and exact binomial Monte Carlo (MC) intervals for coverage and failure frequencies.
Replication counts and calculation details are specified in Appendix~\ref{app:reporting}.

\subsection{Exactness and numerical control}
\label{sec:num-exact}

We check the elimination of full histories separately from the tangent-grid approximation.
The full-history completion programs have $42$, $682$, and $3110$ mass variables for $(K,s)=(3,4),(5,4),(5,6)$, compared with $3$, $5$, and $5$ cell contributions.
Under identical scientific restrictions and tangent grids, the maximum endpoint and outward-value differences between the independently assembled programs are $5.55\times10^{-17}$ and $7.66\times10^{-15}$, respectively (Table~\ref{tab:completion}).
These calculations check the reduction rather than compare competing scientific models.

Mesh refinement holds the population inputs fixed.
Reducing $h$ from $1/16$ to $1/128$ decreases the inner-primal/outer-value gap from $1.95\times10^{-3}$ to $2.71\times10^{-5}$ at this law (Table~\ref{tab:mesh}).
Numerical diagnostics for the sampling experiments are recorded separately.
Allocating the entire budget independently to every cell increases the population width from $0.144$ to $0.321$, illustrating the effect of relaxing the shared restriction.

\subsection{Scientific information and confidence enlargement}
\label{sec:num-information}

We first compare the information supplied by different scientific classes.
Table~\ref{tab:scientific} uses the reference arm-zero law, whose no-rescue mean is $0.405$, and $n=1000$.
Every displayed class contains the generating law.
The signed budget reduces the sharp width from the range-only value $0.516$ to $0.144$, and oracle calibration reduces it further to $0.0853$.
The $\delta$-box gives a narrower interval than the budget because its restrictions are stronger at these parameter values.
The time-only ratio class imposes a different comparison across crossing strata, as described in Appendix~\ref{app:scientific-comparators}.

\begin{table}[tb]
    \centering
    \caption{Scientific classes at the reference arm-zero law, $n=1000$, and $500$ paired replications.
    Each row uses its own sharp interval and the inference construction specified in Appendix~\ref{app:scientific-comparators}.
    Exact calibration supplies oracle information.}
    \label{tab:scientific}
    \small
    \begin{tabular}{@{}llrrr@{}}
        \toprule
        Scientific class & Inference & Sharp width & Mean width & Set coverage \\
        \midrule
        Range only & Endpoint CP & $0.516$ & $0.568$ & $0.962$ \\
        Signed effects & Endpoint CP & $0.388$ & $0.445$ & $0.960$ \\
        $\delta$-box ($\delta=0.2$) & Cell CP & $0.103$ & $0.200$ & $0.998$ \\
        Time-only ($\lambda=1.25$) & Conditional CP & $0.127$ & $0.260$ & $1.00$ \\
        Time-only ($\lambda=1.5$) & Conditional CP & $0.205$ & $0.343$ & $1.00$ \\
        Time-only ($\lambda=2$) & Conditional CP & $0.295$ & $0.436$ & $1.00$ \\
        Signed budget ($\Gamma=0.2$) & Cell CP & $0.144$ & $0.237$ & $0.998$ \\
        Calibrated budget & Cell CP & $0.0853$ & $0.184$ & $0.998$ \\
        \bottomrule
    \end{tabular}
\end{table}

We next compare confidence constructions for the fixed signed-budget class.
In Table~\ref{tab:inference}, Cell CP has mean width $0.2372$ (MCSE $9.62\times10^{-5}$), compared with $0.2129$ (MCSE $9.41\times10^{-5}$) for Aggregate CP and $0.2018$ (MCSE $1.20\times10^{-4}$) for the endpoint bootstrap.
The aggregate relaxation is sharp at this population law because its maximizing allocation satisfies every outcome cap.
The bootstrap has empirical set coverage $0.960$, with Monte Carlo interval $[0.939,0.975]$.
Among the finite-sample-valid coordinatewise constructions evaluated here, Cell CP has the smallest mean width.
Table~\ref{tab:atomic-kl} additionally shows the effect of protecting different collections of observed features.

\begin{table}[tb]
    \centering
    \caption{Inference for the fixed signed-budget class, $n=1000$, $M=5$, $\Gamma=0.2$, and $500$ paired replications.
    Coordinatewise methods share the same features and optimization.
    Aggregate CP is an outer relaxation; Wald and bootstrap have regular-law justifications.
    MC intervals quantify replication uncertainty.}
    \label{tab:inference}
    \small
    \begin{tabular}{@{}lrrrr@{}}
        \toprule
        Procedure & Mean width & MCSE & Set coverage & $95\%$ MC interval \\
        \midrule
        Cell CP & $0.237$ & $9.62\times 10^{-5}$ & $0.998$ & $[0.989, 1.00]$ \\
        Cell KL & $0.257$ & $9.91\times 10^{-5}$ & $1.00$ & $[0.993, 1.00]$ \\
        Empirical Bernstein & $0.298$ & $9.98\times 10^{-5}$ & $1.00$ & $[0.993, 1.00]$ \\
        Hoeffding & $0.263$ & $9.04\times 10^{-5}$ & $1.00$ & $[0.993, 1.00]$ \\
        Wald & $0.236$ & $9.68\times 10^{-5}$ & $0.998$ & $[0.989, 1.00]$ \\
        Aggregate CP & $0.213$ & $9.41\times 10^{-5}$ & $0.982$ & $[0.966, 0.992]$ \\
        Endpoint bootstrap & $0.202$ & $1.20\times 10^{-4}$ & $0.960$ & $[0.939, 0.975]$ \\
        \bottomrule
    \end{tabular}
\end{table}

Joint regions are evaluated on a separate set of paired datasets.
At the same reference law, their mean widths are $0.2746$ for joint LR and $0.3301$ for joint split, compared with $0.2373$ for Cell CP on those datasets (Table~\ref{tab:joint-regions}).
Thus avoiding coordinatewise error allocation does not reduce target-interval width in this setting.
The paired results under active caps are considered below.

Sample size and sensitivity radius affect different sources of uncertainty.
For the fixed oracle-calibrated class, increasing $n$ from $50$ to $10000$ decreases mean plug-in error from $0.0673$ (MCSE $0.00309$) to $0.00395$ (MCSE $1.23\times10^{-4}$) and Cell CP width from $0.5133$ (MCSE $0.00088$) to $0.1165$ (MCSE $0.00004$), while the sharp width remains $0.0853$ (Table~\ref{tab:accuracy}).
Table~\ref{tab:contrast} reports contrasts formed from separately protected arm intervals.
At $n=10000$ and $\Gamma=0.05$, imposed-set coverage is $0.994$ (95\% MC interval $[0.983,0.999]$) but generating-target coverage is zero (95\% MC interval $[0.000,0.007]$; Table~\ref{tab:radius-grid}).
The population sensitivity curves in Figure~\ref{fig:sensitivity} distinguish these scientific restrictions from sampling uncertainty.

\subsection{Rare cells and active outcome caps}
\label{sec:num-boundary}

We separate uncertainty about an absent empirical cell from the effect of a binding outcome restriction.
The rare-cell diagnostic estimates one Bernoulli probability $p$, with observed outcomes known to be zero and range-only sharp interval $[0,p]$.
At $n=1000$ and $np=0.5$, the zero-count frequency is $0.605$ (95\% MC interval $[0.598,0.612]$).
Wald and bootstrap upper limits vanish at zero counts, giving empirical set coverage $0.395$ for each procedure (95\% MC interval $[0.388,0.402]$).
The four bounded-moment constructions retain positive upper limits and cover in all $20000$ replications at this setting (95\% MC interval $[0.9998,1.0000]$ for each procedure; Figure~\ref{fig:boundary}, left).
This diagnostic isolates uncertainty at zero counts; its tail conventions, stated in Appendix~\ref{app:external-rare}, do not support an equal-level efficiency comparison.

The active-cap design assesses retention of the cellwise scientific model.
Changing the odd-visit crossing means to $q^R_{0,k}=0.96$ and $\tau_k=0.02$ reduces the sharp width to $0.103$, whereas the aggregate population relaxation retains width $0.144$.
At $n=2500$, Aggregate CP is shorter than Cell CP, with mean widths $0.1892$ (MCSE $0.00006$) and $0.2032$ (MCSE $0.00007$).
At $n=10000$, the ordering reverses, with widths $0.1664$ (MCSE $0.00003$) and $0.1569$ (MCSE $0.00005$; Figure~\ref{fig:boundary}, right).
The aggregate relaxation gap persists at the population level even as sampling uncertainty decreases.

The separate joint-region comparison also benefits from retaining active caps.
Joint LR is shorter than Cell CP by $0.00332$ at $n=1000$ and $0.01057$ at $n=10000$, while joint split remains wider (Table~\ref{tab:joint-regions}).
The paired MCSEs of these width differences are $1.63\times10^{-4}$ and $1.46\times10^{-5}$, respectively.
These are comparisons at the specified laws, rather than a general width ordering among input regions.
Binding outcome caps alone do not imply a nonregular sampling law.

\begin{figure}[tb]
    \centering
    \includegraphics[width=0.49\linewidth]{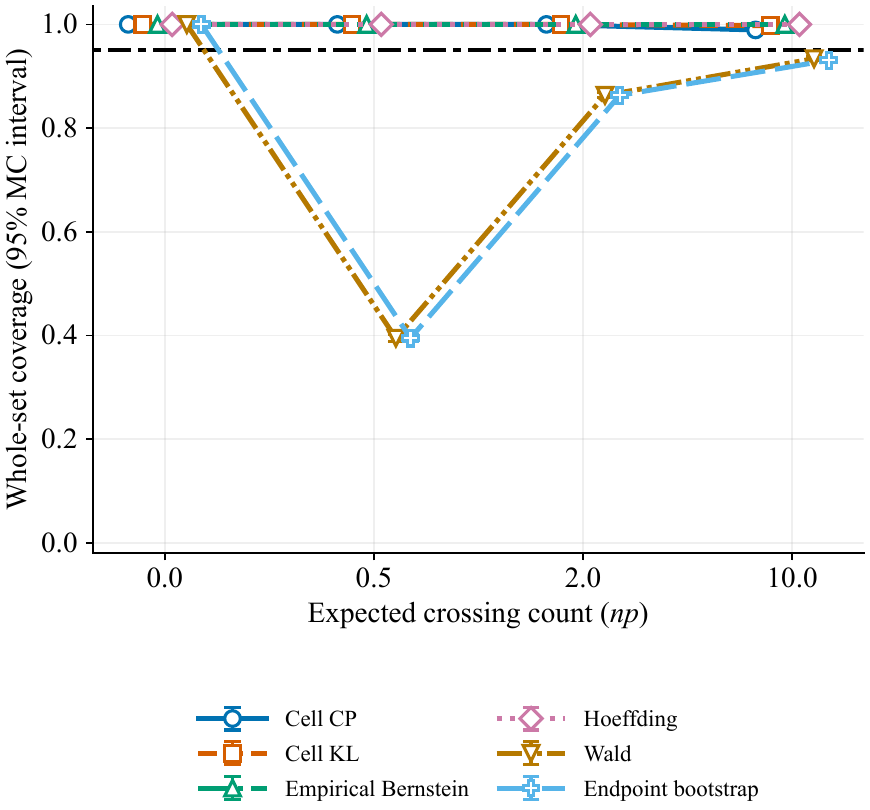}\hfill
    \includegraphics[width=0.49\linewidth]{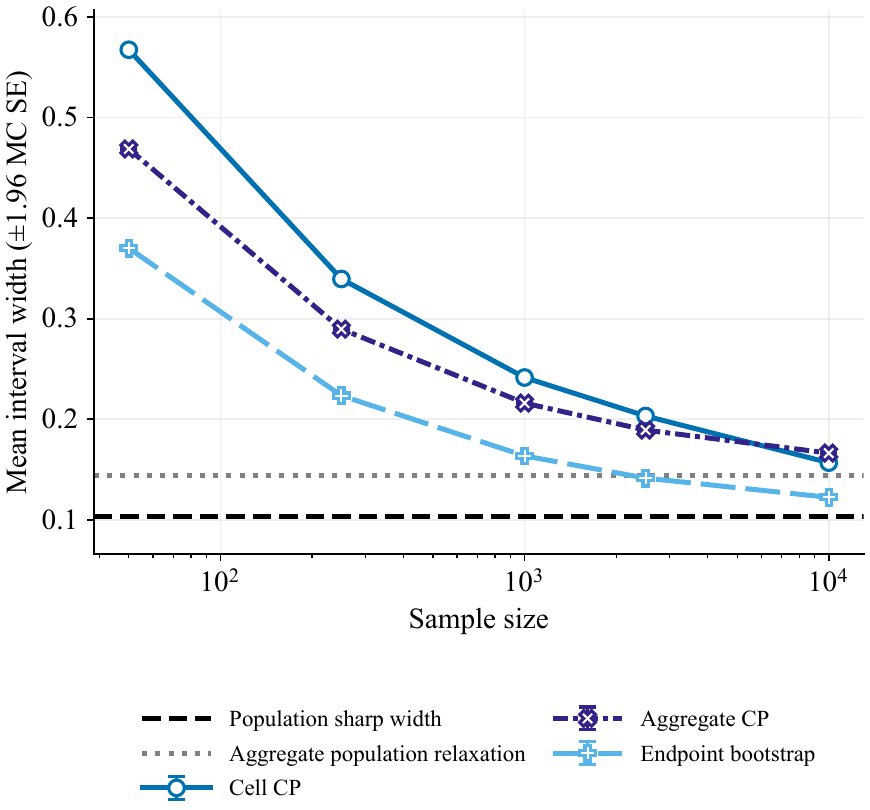}
    \caption{Rare empirical cells and active outcome caps.
    Left: coverage of $[0,p]$ at $n=1000$, with $20000$ replications per setting and $95\%$ binomial Monte Carlo intervals.
    The bounded-moment upper limits use the upper endpoints of two-sided $95\%$ intervals, while the bootstrap uses a one-sided $95\%$ construction.
    Right: mean widths under active outcome caps, with $500$ replications per setting and $\pm1.96$ MCSE.
    Horizontal lines on the right mark the sharp width and aggregate population relaxation.}
    \label{fig:boundary}
\end{figure}

\subsection{Partition refinement and additional designs}
\label{sec:num-observational}

We examine partition refinement and then assess changes in the input-estimation problem and outcome distribution.
The partition experiment uses continuous histories and nested, prespecified cells with $M\in\{1,5,10,20,40,80\}$.
With common sign limits, zero center, and fixed $\Gamma=0.2$, the fine-cell class is contained in the coarse-cell class.
Under heterogeneous caps, the sharp width decreases from $0.1405$ at $M=5$ to $0.0949$ at $M=10$.
Under the original law it is $0.1414$ for both partitions at the displayed precision.
All calculations impose the outcome caps.

Confidence width need not decrease under scientific refinement.
Under heterogeneous caps, refining from $M=5$ to $M=10$ increases mean Cell CP width by $0.00583$ at $n=1000$ but decreases it by $0.01121$ at $n=10000$ (paired MCSE $0.00005$ for the latter difference; Figure~\ref{fig:partition}).
Under the original law, the corresponding changes are increases of $0.00563$ and $0.00180$.
Table~\ref{tab:partition} gives the corresponding enlargement, plug-in errors, and empty-cell fractions.
The comparison supports reporting scientifically prespecified partitions together with their sampling cost, rather than choosing a partition by its observed interval width.
Additional continuous-history configurations are reported in Table~\ref{tab:continuous}.

\begin{figure}[tb]
    \centering
    \includegraphics[width=0.49\linewidth]{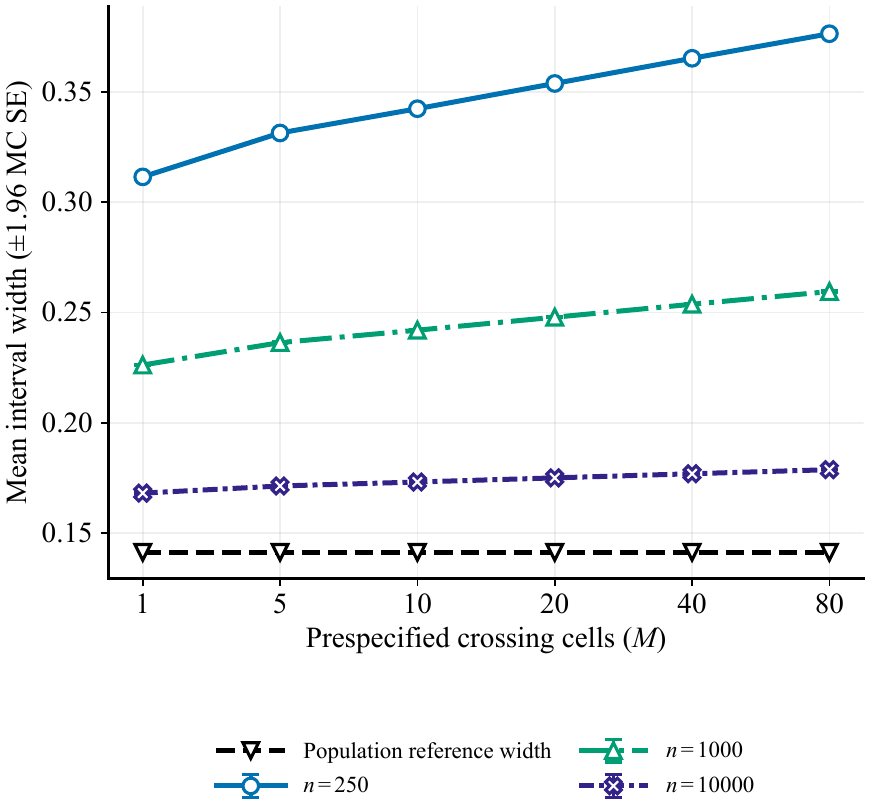}\hfill
    \includegraphics[width=0.49\linewidth]{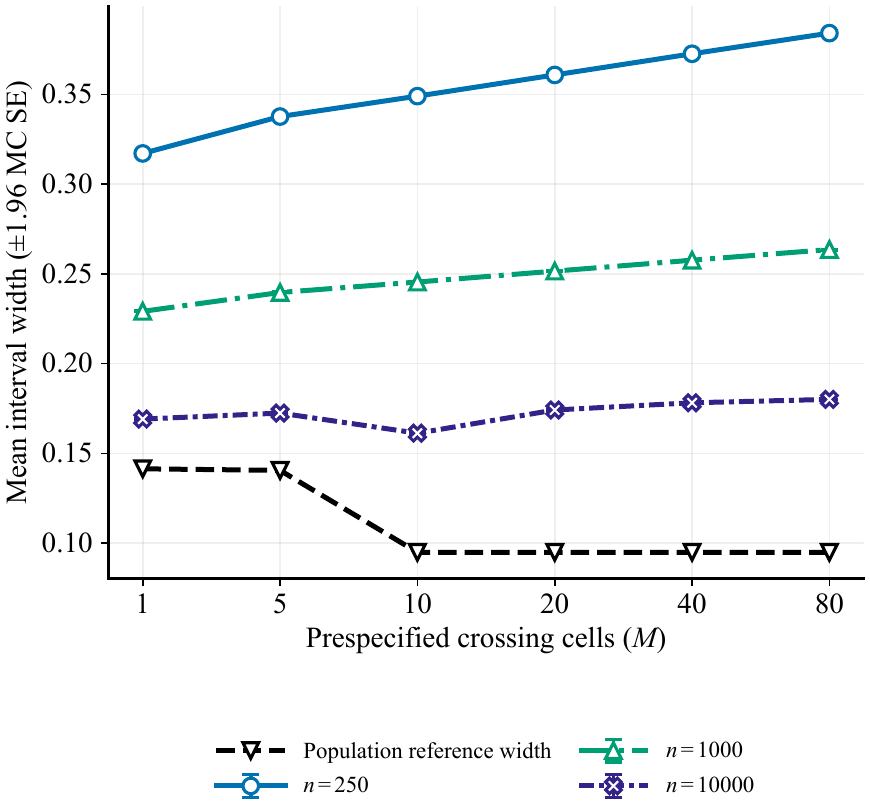}
    \caption{Partition refinement under the original continuous-history law (left) and heterogeneous outcome caps (right).
    Both panels retain all outcome restrictions.
    Solid curves give mean Cell CP widths at fixed $n$, with $\pm1.96$ MCSE; dashed curves give each partition's population sharp width.}
    \label{fig:partition}
\end{figure}

The observational experiment changes baseline assignment and moment estimation while retaining the standardized target.
Under nonlinear assignment at $N=10000$, a correctly specified logistic fit gives mean width $0.1909$ (MCSE $7.66\times10^{-5}$) with sandwich variance and $0.2084$ (MCSE $9.41\times10^{-5}$) with naive variance.
Omitting assignment interactions gives mean observed-mean error $-0.0423$ (MCSE $0.00046$) and set coverage $0.008$ (95\% MC interval $[0.001,0.029]$) despite the sandwich correction (Table~\ref{tab:ps-nonlinear}).
Tables~\ref{tab:ps-linear}--\ref{tab:ps-nonlinear} and Figure~\ref{fig:propensity-coverage} report the full comparison, including small-sample and unaugmented forest-weight diagnostics.

The bounded continuous-outcome experiment preserves the reference conditional means while replacing Bernoulli outcomes by beta outcomes.
It therefore holds the population sharp interval fixed and tests confidence construction beyond binary outcome features (Appendix~\ref{app:continuous-outcome}).
The hybrid procedures combine CP limits for probability features with EB or Hoeffding limits for outcome features.
At $n=1000$, mean widths are $0.2263$ (MCSE $9.70\times10^{-5}$) for Hybrid EB--CP, $0.2617$ (MCSE $9.57\times10^{-5}$) for Hybrid Hoeffding--CP, $0.1839$ (MCSE $9.57\times10^{-5}$) for Cell Wald, and $0.1654$ (MCSE $1.05\times10^{-4}$) for the endpoint bootstrap.
Both hybrid procedures cover in all $500$ replications (95\% MC interval $[0.993,1.000]$), compared with coverage $0.998$ for Cell Wald (95\% MC interval $[0.989,1.000]$) and $0.950$ for bootstrap (95\% MC interval $[0.927,0.967]$).
Table~\ref{tab:continuous-outcome} gives the full sample-size comparison and numerical diagnostics.

\section{Discussion}
\label{sec:discussion}

Under a rescue protocol that remains in force once triggered, the no-rescue mean is determined by the observed law and the average effects of withholding rescue at first-crossing histories.
Every specification of these effects within the outcome range is realized by a full-data law that preserves the observed distribution.
For restrictions on prespecified stratum averages, the sharp interval is exactly the image of a finite convex program.
Confidence projection inherits the validity of a simultaneous input region without requiring endpoint differentiability or positive crossing-cell probabilities.
The local feasibility conditions additionally yield endpoint-error bounds.
The experiments show how retaining cellwise outcome caps can improve interval width in larger samples while preserving the declared scientific class.
In practice, an analysis should state the partition, effect limits, reference effects, calibration, and budget scale.
Displaying the sharp interval separately from its confidence enlargement over a grid of radii distinguishes scientific uncertainty from sampling uncertainty.

The first-crossing reduction uses the protocol's deterministic, absorbing rescue rule.
If rescue is administered before the trigger or withheld after it, the protocol crossing time need not be the first point of divergence between the actual course and sustained non-rescue.
Discarding noncompliant records can introduce selection bias.
A decomposition at the first actual rescue is possible under consistency assumptions, but changes the strata and the rescued-course reference.
A hybrid analysis incorporating non-rescued crossers must address selection into non-rescue and support for subsequent decisions.
Such an extension requires its own attainability argument.

Three extensions merit further study.
First, the finite-sample joint likelihood region is developed for binary independent-arm data.
Under the stated bounded-score conditions, coordinatewise regions also accommodate known-propensity weighting and bounded continuous outcomes with finite-sample guarantees.
For estimated-propensity and augmented moments, validity instead follows under the stated asymptotic conditions.
Tests tailored to the target while retaining joint error control may reduce the conservativeness of coordinatewise error allocation.
Second, endpoint stability is established for a fixed partition under local feasibility conditions.
Without budget slack, endpoint changes can have square-root order even when all cells have positive mass.
Data-adaptive partition selection is a separate extension that requires accounting for selection or simultaneous inference over the candidate partitions.
Third, additional restrictions can be incorporated directly when they are jointly convex in the optimization variables.
For positive-mass cells, ordering cell-average effects is linear when the cell probabilities are fixed, but is generally nonconvex when those probabilities vary over a confidence region.
Confidence projection for such restrictions requires a formulation that retains them, for example a globally solved nonconvex program.

The framework retains the original no-rescue target and makes assumptions about unobserved continuations explicit.
Its structural reduction, finite optimization, and confidence projection separate scientific restrictions, sampling uncertainty, and numerical error within one analysis.

\section*{Acknowledgements}
Shu Tamano was supported by JSPS KAKENHI Grant Number 25K24203.

\section*{Code Availability}
\label{sec:code-availability}
The Python scripts reproducing the numerical experiments are available at \url{https://github.com/shutech2001/first-crossing-no-rescue-bounds}.

\clearpage
\appendix
\section{Omitted proofs}
\label{app:proofs}

This appendix proves the results of Sections~\ref{sec:geometry}--\ref{sec:statistics}.
Appendices~\ref{pf:thm-trace} and~\ref{pf:prop-sharp} use the original outcome range;
the remaining arguments use the normalized range $[0,1]$.
Conditional statements hold almost everywhere under the relevant risk-set or crossing distributions, zero-mass components need no conditional-mean specification, and regular conditional kernels exist on the standard Borel spaces \citep[Section~4.1.3]{durrett2019probability}.

\subsection{Proof of Theorem~\ref{thm:trace}}
\label{pf:thm-trace}

We first construct joint potential trajectories that realize prescribed crossing means, and then identify the compatible continuation means.
Let $P_W$ be the observed marginal law of $W$, and let $Q_a(w,\dd v)$ be a conditional law of $V=(L_1,D_1,\ldots,L_K,D_K,Y)$ given $(W,A)=(w,a)$.
Baseline positivity determines this kernel up to $P_W$-null sets, since $\int_Bg_a\dd P_W=0$ and $g_a>0$ almost everywhere force $P_W(B)=0$.
Thus $P_W(\dd w)Q_a(w,\dd v)$, with $a$ inserted, is the standardized protocol law $P_a$.

\begin{lemma}[Joint full-data completion]
    \label{lem:full-data-completion}
    For each arm and visit, prescribe a measurable crossing mean $\tilde{q}_{a,k}\in[y_-,y_+]$ almost everywhere on that arm's crossing set.
    Under the observed setup and baseline positivity, these means are jointly realizable by a law in $\cP(P)$.
    Under every law in $\cP(P)$, the protocol and no-rescue trajectories have the same first crossing and crossing history within each arm.
\end{lemma}

\begin{proof}[Proof of Lemma~\ref{lem:full-data-completion}]
    Agreement through first crossing follows by finite induction.
    The visit-one states agree by temporal consistency, so their trigger values agree.
    If no crossing has occurred before visit $k$, temporal consistency equates the states through $k$, and hence the current trigger.
    The histories therefore agree through first crossing, or throughout follow-up if no crossing occurs.
    On the latter event, outcome consistency also gives agreement of outcomes.

    To construct a law, draw $W\sim P_W$ and, given $W$, draw $V^0$ and $V^1$ independently with laws $Q_0(W,\cdot)$ and $Q_1(W,\cdot)$.
    Insert $(W,a)$ to obtain $O^a$, and draw independent uniforms $U_0,U_1$, independently of the protocol trajectories.
    Choose a fixed point $l_k^\dagger$ in each nonempty state space and set
    \begin{equation*}
        L_k^{a,\NR}
        =
        \begin{cases}
            L_k^a,&k\le T^a
            ,\\
            l_k^\dagger,&k>T^a
            .
        \end{cases}
    \end{equation*}
    Set all no-rescue decisions to zero and define
    \begin{equation}
        \label{eq:completion}
        Y^{a,\NR}
        =
        \begin{cases}
            Y^a,
            &T^a>K
            ,\\
            y_-+(y_+-y_-)\1\Biggl\{
                U_a\le\dfrac{\tilde{q}_{a,k}(X_k^a)-y_-}{y_+-y_-}
            \Biggr\},
            &T^a=k
            ,
        \end{cases}
    \end{equation}
    after modifying the prescribed means on crossing-null sets to lie in the outcome range everywhere.
    The construction is measurable and satisfies temporal consistency, and independence of $U_a$ gives the prescribed crossing means.

    Finally draw $A$ from $g_a(W)$, conditionally independently of all potential variables given $W$, and observe $O=O^A$.
    For bounded measurable $f$,
    \begin{equation*}
        \EE[f(O)]
        =
        \sum_a\int g_a(w)
        \Biggl\{
            \int f(w,a,v)Q_a(w,\dd v)
        \Biggr\}P_W(\dd w)
        =
        \int f(o)P(\dd o)
        ,
    \end{equation*}
    so the entire observed law is preserved, and baseline consistency and conditional independence hold by construction.
    Both arm completions are defined before assignment, which proves joint attainability.
    The conditional independence used to couple the arms is a device of the construction and imposes no restriction on the declared class.
\end{proof}

The same construction proves the range-only interval in \eqref{eq:range}.
Choosing all crossing means equal to $y_-+v(y_+-y_-)$, $v\in[0,1]$, realizes every point of the interval, and the reverse inclusion follows by splitting the target into crossing and non-crossing contributions and applying the outcome range.

\begin{proof}[Proof of Theorem~\ref{thm:trace}]
    \emph{Part (i).}
    For bounded measurable $v$ and $k<K$, iterated expectation gives
    \begin{equation*}
        \EE_a[\1(T>k)v(X_k)m_{k+1}(X_{k+1})]
        =
        \EE_a[\1(T>k)v(X_k)Y^{a,\NR}]
        =
        \EE_a[\1(T>k)v(X_k)m_k(X_k)]
        ,
    \end{equation*}
    because $X_{k+1}$ contains $X_k$ and the visit-$k$ decision, and the known trigger determines $\{T>k\}$ within the risk set.
    This proves the nonterminal equations in \eqref{eq:continuation}.
    The terminal equation follows from outcome consistency, and the same calculation with $Y$ shows that $m^R$ satisfies the same equations.
    At crossing, $m_k-m_k^R=\tau_k$ and the potential-outcome range gives \eqref{eq:compat}.
    The difference $m-m^R$ is zero at terminal non-crossing histories and follows the homogeneous averaging recursion elsewhere, so backward induction gives $m-m^R=\Ext\tau$, uniquely on each risk set.
    Splitting the target by $T$ yields \eqref{eq:causal-decomposition}.

    \emph{Part (ii).}
    Apply Lemma~\ref{lem:full-data-completion} with crossing means $q_k+\tau_k$.
    The resulting law has the prescribed bridge, and part (i) identifies its continuation as $m^R+\Ext\tau$.
    Prescribing both arm bridges in the same construction gives joint attainment.
\end{proof}

\subsection{Proof of Proposition~\ref{prop:sharp}}
\label{pf:prop-sharp}

\begin{proof}[Proof of Proposition~\ref{prop:sharp}]
    Every admissible law has a bridge in $\cD$ and target \eqref{eq:causal-decomposition}, and Theorem~\ref{thm:trace} realizes every bridge in $\cD$, so the affine image \eqref{eq:sharp-image} is exactly the identified set.
    The image of a convex class under a real-valued affine functional is an interval, bounded here by the outcome range, and an endpoint belongs to the image precisely when the respective extremum is attained.
\end{proof}

\subsection{Proof of Proposition~\ref{prop:algorithm}}
\label{pf:prop-algorithm}

We prove finite attainability and control the tangent approximation on the capped perspective domain.

\begin{lemma}[Closed capped perspective]
    \label{lem:closed-perspective}
    On $p\ge0$, $|t|\le p$, the function $f_c$ is continuous and convex, with $0\le f_c(p,t)\le4p$ for $|c|\le1$.
    For $s\ge0$, the conditions $f_c(p,t)\le s$, $(t-cp)^2\le ps$, and $\|(2(t-cp),p-s)\|_2\le p+s$ are equivalent.
\end{lemma}

\begin{proof}[Proof of Lemma~\ref{lem:closed-perspective}]
    For $p>0$, $f_c$ is the perspective of the convex quadratic $(x-c)^2$ \citep[Section~3.2.6]{boyd2004convex}.
    The cap implies $|t-cp|\le2p$, which gives the bound and continuity at $(0,0)$, and convexity extends to that point by continuity.
    Multiplication by $p$ proves the first equivalence at positive mass, while the cap forces $t=0$ at zero mass.
    Squaring the norm inequality, whose right side is nonnegative, proves the second equivalence.
\end{proof}

\begin{proof}[Proof of Proposition~\ref{prop:algorithm}]
    \emph{Part (i).}
    Every admissible law satisfies $0\le b_j+t_j\le p_j$ and \eqref{eq:cell-class}, so its contribution vector lies in $\cT(\eta,G)$.
    Conversely, for any vector in that set, assign the no-rescue mean $(b_j+t_j)/p_j$ to every crossing history in cell $j$ when $p_j>0$, and arbitrary values in $[0,1]$ on zero-mass cells.
    Lemma~\ref{lem:full-data-completion} realizes these means without altering the observed law, and their contribution vector is exactly $t$.
    The realizing bridge equals the assigned mean minus $q_k(x)$ and can vary within a cell, so the construction neither requires constant bridge effects nor strengthens the average-effect restrictions.
    The contribution box is compact, the capped perspective is continuous, and the constraints are convex, so the nonempty feasible set has attained extrema and an interval image;
    Lemma~\ref{lem:closed-perspective} gives the conic representation.

    \emph{Part (ii).}
    For $p>0$, write $x=t/p\in[-1,1]$.
    For a tangent location $v$,
    \begin{equation}
        \label{eq:tangent-identity}
        f_c(p,t)-\{2(v-c)t+(c^2-v^2)p\}
        =
        p(x-v)^2
        ,
    \end{equation}
    so zero and every affine term minorize $f_c$, and a grid point within $h/2$ of $x$ gives error at most $ph^2/4$.
    At zero mass both functions vanish.
    Summation proves \eqref{eq:mesh-error}, and retaining all other constraints proves \eqref{eq:brackets}.
    At $G=0$, the nonnegative terms sum to zero exactly when $t_j=c_jp_j$ for every cell, which justifies the exact convention.

    \emph{Part (iii).}
    Coherence bounds $\mu,p,b$, the outcome caps bound $t$, and $d=Ht$ bounds the calibration coordinates, so all closed joint feasible sets lie in a common compact box.
    They contain the exact nonempty set, and nested grids increase the tangent envelope, so outer minima increase and outer maxima decrease.
    A subsequential limit of outer optimizers satisfies the closed linear constraints and the closed input region, and its budget is at most $G$ because the approximate budgets are at most $G+e_h$ and the perspective is continuous.
    The limit is therefore exactly feasible, and continuity of the objective identifies the limiting extrema with the exact extrema.
\end{proof}

The reported endpoints are outward bounds on the optimal values of the linear programs in Section~\ref{sec:algorithm}, obtained from the following weak-duality certificate;
Algorithm~\ref{alg:boundary} and the proof of Corollary~\ref{cor:coverage} use it.
For a nonempty bounded linear program $\inf\{q^\top x:Bx\le b,\ a\le x\le z\}$ and any $\lambda\ge0$, define
\begin{equation}
    \label{eq:lp-cert}
    v_-(\lambda)
    =
    -\lambda^\top b
    +\sum_i\min\{(q+B^\top\lambda)_ia_i,(q+B^\top\lambda)_iz_i\}
    .
\end{equation}

\begin{lemma}[Outward objective certificate]
    \label{lem:outward-certificate}
    The value $v_-(\lambda)$ is a lower bound on the minimizing optimum.
    Negating the corresponding certificate for the negative objective gives an upper bound on a maximizing optimum.
\end{lemma}

\begin{proof}[Proof of Lemma~\ref{lem:outward-certificate}]
    For every feasible $x$, $q^\top x+\lambda^\top(Bx-b)\le q^\top x$, and minimizing the left side over the coordinate box gives \eqref{eq:lp-cert} by weak Lagrangian duality \citep[Section~5.1]{boyd2004convex}.
    Equalities may be represented by paired inequalities or unrestricted multipliers, and no dual stationarity condition is required.
\end{proof}

\subsection{Proof of Corollary~\ref{cor:coverage}}
\label{pf:cor-coverage}

We first record the coordinatewise input regions of Section~\ref{sec:statistics} with their boundary conventions and verify the closedness and measurability that Corollary~\ref{cor:coverage} requires;
the joint region is treated in Lemma~\ref{lem:joint-region} below.
For a Bernoulli count $k$ from $n\ge1$ trials and coordinate error $\delta$, let $B^{-1}_{u,v}$ denote the beta quantile function.
The Clopper--Pearson interval is
\begin{equation}
    \label{eq:cp}
    \left[
        \begin{cases}
            0,&k=0,\\
            B^{-1}_{k,n-k+1}(\delta/2),&k>0,
        \end{cases}
        \quad
        \begin{cases}
            1,&k=n,\\
            B^{-1}_{k+1,n-k}(1-\delta/2),&k<n,
        \end{cases}
    \right]
\end{equation}
with $[0,1]$ when there are no trials and the singleton value for a known constant coordinate.
Inverting the two binomial tails bounds each noncoverage probability by $\delta/2$ \citep{clopper1934use};
at $k=0$ the upper limit is $1-(\delta/2)^{1/n}$, and at $k=n$ the lower limit is $(\delta/2)^{1/n}$.
The empirical Bernstein half-width \eqref{eq:eb} is \citet[Theorem~4]{maurer2009empirical} scaled to an interval of length $B_q$ and applied to both signs with $\gamma=\alpha/(2D)$;
the pairwise variance $\{n(n-1)\}^{-1}\sum_{i<j}(Z_i-Z_j)^2$ in that theorem equals the unbiased sample variance $\hat{v}$.
Hoeffding's inequality $\PP\{|\bar{Z}-\EE[Z]|>x\}\le2\exp(-2nx^2/B^2)$ gives the half-width $B_q\sqrt{\log(2D/\alpha)/(2n_q)}$ at coordinate error $\alpha/D$ \citep{hoeffding1963probability}.
A union bound over the $D$ coordinates gives simultaneous coverage at least $1-\alpha$, and coordinate types may be mixed.
The resulting region is an intersection of inverse images of closed intervals with the known linear identities and is therefore closed.
Its feasibility events and optimization values are measurable.
For fixed scientific inputs and grid, all coherent calibrated feasible variables lie in a common compact box, so feasibility and objective sublevel sets are projections of closed constraint sets along compact variables and are closed as functions of the interval endpoints, which are themselves measurable.

\begin{proof}[Proof of Corollary~\ref{cor:coverage}]
    On $E_n=\{\eta\in\cC_n\}$, every $t\in\cT(\eta,G)$ remains feasible for the confidence program, so its infimum is at most $L(\eta,G)$ and its supremum at least $U(\eta,G)$.
    By \eqref{eq:brackets} and Lemma~\ref{lem:outward-certificate}, tangent relaxation and outward certificates only enlarge these limits.
    The physical-range fallback contains the population interval, and intersecting with $[0,1]$ removes no population-feasible value.
    Hence $E_n\subseteq\{\cI(\eta,G)\subseteq\cI_n^{\mathrm{out}}\}$, which proves \eqref{eq:coverage}.
    The same event implies every containment on a prespecified finite or countable sensitivity grid that uses the region and has nonempty population classes, and taking probabilities preserves the finite-sample, asymptotic, or uniform qualification of the input guarantee.
\end{proof}

The joint region \eqref{eq:joint-split} requires a separate argument at empty empirical and population categories.

\begin{lemma}[Joint multinomial likelihood region]
    \label{lem:joint-region}
    For binary independent-arm observations and a fixed split into two nonempty parts, \eqref{eq:joint-split} is a closed convex region on the probability simplex, with
    \begin{equation*}
        \PP_\pi\{\pi\in\cR_n^{\mathrm{split}}\}
        \ge
        1-\alpha
    \end{equation*}
    for every population probability vector, including vectors with zero coordinates.
    Its linear image in $(\mu,p,b)$ supplies an input region for Corollary~\ref{cor:coverage}.
    Joint optimization with the scientific constraints is convex and has an exponential-cone representation in addition to the perspective cones.
\end{lemma}

\begin{proof}[Proof of Lemma~\ref{lem:joint-region}]
    Write $S_v(\pi)$ for the exponentiated likelihood ratio contributed by split $v$.
    Conditional on the other split, its training distribution $\tilde{\pi}_{-v}$ is fixed and normalized.
    For true support $\cS=\{i:\pi_i>0\}$, the multinomial theorem gives
    \begin{equation*}
        \EE_\pi[S_v(\pi)\mid N_{3-v}]
        =
        \Biggl(
            \sum_{i\in\cS}\tilde{\pi}_{-v,i}
        \Biggr)^{n_v}
        \le
        1
        ,
    \end{equation*}
    and a positive count at a population-zero atom has probability zero.
    Consequently $\EE_\pi[(S_1+S_2)/2]\le1$ although the two ratios are dependent, and Markov's inequality proves the confidence statement, including the boundary conventions.

    Each log ratio is a constant minus a nonnegative weighted sum of $\log\pi_i$, hence convex and lower semicontinuous with the specified infinite values.
    The log-sum-exponential function is convex and nondecreasing in each argument, so the logarithm of $(S_1+S_2)/2$ has the same properties and its sublevel set is closed and convex on the compact simplex.
    Logarithmic and exponential epigraphs have exponential-cone representations, and the budget has the perspective-cone representation of Lemma~\ref{lem:closed-perspective}.
    The linear image is compact, possibly empty, and on the coverage event contains the true observed moments, so the projection argument of Corollary~\ref{cor:coverage} applies with the physical-range fallback for an empty feasible intersection.
    Optimization values are measurable because, for each sample size, there are finitely many count vectors and the split is fixed independently of the values.
\end{proof}

Profiling gives an equivalent test-inversion form of either joint program.
For a candidate mean $x$, minimize the region's criterion over $(\pi,t)$ subject to the scientific constraints, coherence, and $\mu+\sum_jt_j=x$; the accepted values are exactly the projection of that region under its stated calibration.
This connects the computation to the set-inference literature \citep{chernozhukov2007estimation,romano2010inference}, which also develops other criteria and calibrations.

\subsection{Proof of Proposition~\ref{prop:stability}}
\label{pf:prop-stability}

We separate deterministic perturbation of the feasible set from the sampling rates of the input regions.

\begin{assumption}[Calibration-anchor stability]
    \label{ass:anchor-general}
    Fix $M,c,l,u,H$ and coherent $\eta$, with $-1\le l_j\le0\le u_j\le1$.
    Use $r$ independent calibration rows after a prespecified algebraic reduction, retaining known identities for omitted coordinates.
    There exist $t^\circ\in\cT(\eta,G)$, $\sigma>0$, a set of cells $I\subseteq\{1,\ldots,M\}$, and $R\in\RR^{M\times r}$ such that
    \begin{equation}
        \label{eq:anchor-general}
        \begin{split}
            &HR
            =
            I_r
            ,
            \quad
            R_{j\cdot}=0
            \quad
            (j\notin I)
            ,\\
            &L_j(p,b)+\sigma
            \le
            t_j^\circ
            \le
            U_j(p,b)-\sigma
            \quad
            (j\in I)
            ,
            \quad
            F_c(p,t^\circ)
            \le
            G-\sigma
            .
        \end{split}
    \end{equation}
    Here $I_r$ is the identity matrix and $R_{j\cdot}$ is row $j$ of $R$.
    Without calibration, take $I=\varnothing$ and omit $R$ and the calibration equations.
\end{assumption}
$R$ assigns calibration corrections to the cells in $I$, which therefore need interior margins;
Assumption~\ref{ass:anchor} is the case $r=0$, $I=\varnothing$.

Fix the scientific inputs in Assumption~\ref{ass:anchor-general}.
All neighborhoods below are relative to coherent primitives satisfying the retained calibration identities and to nonnegative budgets.
Write $\cK(p,b)=\prod_j[L_j(p,b),U_j(p,b)]$ for the contribution box and let $Q't$ denote coordinatewise projection onto $\cK(p',b')$.

\begin{lemma}[Capped perturbation bounds]
    \label{lem:capped-perturbation}
    The perspective satisfies \eqref{eq:perspective-lip}.
    For coherent moment pairs,
    \begin{equation}
        \label{eq:box-lip}
        |L_j(p,b)-L_j(p',b')|
        \vee
        |U_j(p,b)-U_j(p',b')|
        \le
        |p_j-p_j'|+|b_j-b_j'|
        ,
    \end{equation}
    and, for $t\in\cK(p,b)$,
    \begin{equation}
        \label{eq:proj-box}
        \|Q't-t\|_1
        \le
        \|p'-p\|_1+\|b'-b\|_1
        .
    \end{equation}
\end{lemma}

\begin{proof}[Proof of Lemma~\ref{lem:capped-perturbation}]
    For $p>0$ and $x=t/p$,
    \begin{equation*}
        \partial_t f_c=2(x-c)
        ,
        \quad
        \partial_p f_c=c^2-x^2
        ,
    \end{equation*}
    which on $|x|,|c|\le1$ are bounded in absolute value by four and two.
    Integrating along the segment between two positive-mass points gives the Lipschitz bound; if an endpoint has zero mass, truncate the segment and use continuity from Lemma~\ref{lem:closed-perspective}.
    Summing over cells proves \eqref{eq:perspective-lip} without differentiating at zero.

    Maxima and minima of two numbers are one-Lipschitz in the maximum norm, so \eqref{eq:cell-box} with $|l_j|,|u_j|\le1$ gives \eqref{eq:box-lip}.
    Coherence and $l_j\le0\le u_j$ put zero in every contribution interval, so both boxes are nonempty.
    A coordinate projection moves a feasible coordinate by at most the corresponding endpoint change, which gives \eqref{eq:proj-box} after summation.
\end{proof}

We next transfer the anchor to a nearby coherent input vector.
Let $t^\circ,\sigma,I,R$ be the anchor quantities and set
\begin{equation*}
    \varepsilon=\|\eta'-\eta\|_1+|G'-G|
    ,
    \quad
    \kappa=1+\|H\|_{1\to1}
    ,
    \quad
    \beta=\|R\|_{1\to1}\kappa
    ;
\end{equation*}
without calibration, omit all correction terms and take $\beta=0$.

\begin{lemma}[Transfer of the anchor]
    \label{lem:anchor-transfer}
    For sufficiently small coherent perturbations,
    \begin{equation*}
        v^\circ=Q't^\circ
        ,
        \quad
        t^{\circ\prime}
        =
        v^\circ+R(d'-Hv^\circ)
    \end{equation*}
    satisfies the new calibration, lies in the new box with margins at least $\sigma/2$ on $I$, and has budget at most $G'-\sigma/2$.
    Moreover, $\|t^{\circ\prime}-t^\circ\|_1\le(1+\beta)\varepsilon$.
\end{lemma}

\begin{proof}[Proof of Lemma~\ref{lem:anchor-transfer}]
    Lemma~\ref{lem:capped-perturbation} gives displacement at most $\varepsilon$ and $\|d'-Hv^\circ\|_1\le\kappa\varepsilon$.
    The correction has norm at most $\beta\varepsilon$, is supported on $I$, and restores calibration because $HR=I_r$.
    Projection, moving endpoints, and correction use at most $(2+\beta)\varepsilon$ of the original margins, so for sufficiently small perturbations the new vector lies in the box with margins at least $\sigma/2$ on $I$.
    Applying \eqref{eq:perspective-lip} after box feasibility has been established,
    \begin{equation*}
        F_c(p',t^{\circ\prime})
        \le
        F_c(p,t^\circ)+\{4(1+\beta)+2\}\varepsilon
        \le
        G'-\sigma+(7+4\beta)\varepsilon
        ,
    \end{equation*}
    and further shrinking the neighborhood retains slack $\sigma/2$.
    The displacement bound follows by the triangle inequality.
\end{proof}

\begin{proof}[Proof of Proposition~\ref{prop:stability}]
    \emph{Part (i): forward perturbation.}
    Lemma~\ref{lem:anchor-transfer} gives local nonemptiness.
    For any $t\in\cT(\eta,G)$, let $v=Q't$.
    The preceding bounds give
    \begin{equation*}
        \|v-t\|_1\le\varepsilon
        ,
        \quad
        \|d'-Hv\|_1\le\kappa\varepsilon
        ,
        \quad
        F_c(p',v)\le G'+7\varepsilon
        .
    \end{equation*}
    Define
    \begin{equation*}
        \lambda=\frac{2(7+4\beta)\varepsilon}{\sigma}
        ,
        \quad
        w=(1-\lambda)v+\lambda t^{\circ\prime}
        ,
        \quad
        t'=w+R(d'-Hw)
        ,
    \end{equation*}
    taking the neighborhood small enough that $\lambda\le1$.
    The mixture has box margins at least $\lambda\sigma/2$ on $I$.
    Its calibration residual is $(1-\lambda)(d'-Hv)$, so the final correction has norm at most $\beta\varepsilon$.
    It preserves the box because $\lambda\sigma/2=(7+4\beta)\varepsilon\ge\beta\varepsilon$, and restores calibration.
    Convexity and then \eqref{eq:perspective-lip} give
    \begin{equation*}
        F_c(p',t')
        \le
        G'+7\varepsilon-\lambda\sigma/2+4\beta\varepsilon
        \le
        G'
        ,
    \end{equation*}
    so $t'$ is feasible.
    Both $v$ and $t^{\circ\prime}$ have $\ell^1$ norm at most one, so
    \begin{equation}
        \label{eq:feasible-transfer}
        \|t'-t\|_1
        \le
        (1+\beta)\varepsilon+2\lambda
        \le
        C_0\varepsilon
        ,
        \quad
        C_0=1+\beta+\frac{4(7+4\beta)}{\sigma}
        .
    \end{equation}

    \emph{Part (i): reverse perturbation and endpoints.}
    We apply the same construction in the reverse direction.
    Project any new feasible vector onto the original box, mix with the original anchor using the same $\lambda$, and restore the original calibration.
    The original anchor has at least the margins and slack $\sigma/2$ used above, so the same bounds give an original feasible vector within $C_0\varepsilon$.
    Applying both transfers to the two endpoint optimizers and adding $|\mu'-\mu|$ proves \eqref{eq:lipschitz} with $C_\ast=1+C_0$; no cell probability outside $I$ is divided by.
    All nearby primitives have anchors with common margin $\sigma/2$ and the same right inverse, so the argument also gives a uniform pairwise Lipschitz bound on a smaller neighborhood.

    \emph{Part (ii): sampling and numerical relaxation.}
    The shrinking region puts the plug-in and confidence programs in this common neighborhood.
    Let $D_n$ denote the left side of \eqref{eq:shrinking-region} and set $e_n=h_n^2/4$.
    Then $D_n=\Op(r_n)$.
    With probability tending to one, every coherent input in the region and the budgets $G\pm e_n$ belong to the common stability neighborhood.
    The lower budget is eventually positive because the anchor implies $G\ge\sigma>0$, and local anchors give nonemptiness.
    For the plug-in outer endpoints, the tangent inclusions imply
    \begin{equation*}
        \begin{split}
            L(\tilde{\eta}_n,G+e_n)
            &\le\hat{L}_{n,h_n}\le L(\tilde{\eta}_n,G)
            ,\\
            U(\tilde{\eta}_n,G)
            &\le\hat{U}_{n,h_n}\le U(\tilde{\eta}_n,G+e_n)
            ,
        \end{split}
    \end{equation*}
    and uniform local stability bounds their deviations from the population endpoints by a constant times $D_n+e_n$.

    The same uniform bound controls the confidence projection.
    Every confidence-outer feasible pair obeys $F_c(p',t)\le G+e_n$, so its objective lies between $L(\eta',G+e_n)$ and $U(\eta',G+e_n)$, which differ from the population endpoints by at most $C(D_n+e_n)$ for a common finite $C$.
    In the opposite directions, the confidence region contains the coherent plug-in and hence its exact feasible interval.
    Consequently,
    \begin{equation*}
        \begin{split}
            &L(\eta,G)-C(D_n+e_n)
            \le
            L_{n,h_n}^{\mathrm{out}}
            \le
            L(\tilde{\eta}_n,G)
            \le
            L(\eta,G)+CD_n
            ,\\
            &U(\eta,G)-CD_n
            \le
            U(\tilde{\eta}_n,G)
            \le
            U_{n,h_n}^{\mathrm{out}}
            \le
            U(\eta,G)+C(D_n+e_n)
            .
        \end{split}
    \end{equation*}
    This uses containment of the plug-in rather than containment of the truth on every realization.
    It proves \eqref{eq:rate} and shows that empty-program probabilities tend to zero.
    By \eqref{eq:brackets} at $\eta'$ and uniform pairwise stability, the inner--outer endpoint gap is $O(e_n)$; implemented endpoints add outward-certificate gaps, and the rate persists when these are $\Op(r_n+h_n^2)$ and numerical failures have vanishing probability.
\end{proof}

We next verify the input-region rate for bounded moment sampling.
Project the raw observed-coordinate means onto the fixed coherence polytope in Euclidean distance, retain source means for external coordinates, and reconstruct redundant calibration coordinates by their known identities.
Enlarge every primitive and aggregate interval just enough to include its value at this coherent plug-in.

\begin{lemma}[Rates for primitive regions]
    \label{lem:primitive-rates}
    Under Assumption~\ref{ass:bounded-moments}, fixed dimension, fixed $\alpha$, and diverging relevant sample sizes, this Bernstein construction satisfies \eqref{eq:shrinking-region} with
    \begin{equation}
        \label{eq:rn}
        r_n
        =
        n^{-1/2}
        +\sum_{j=1}^M\sqrt{p_j/n}
        +M/n
        +n_{\mathrm{ext}}^{-1/2}
        +n_{\mathrm{ext}}^{-1}
        ,
    \end{equation}
    where $n$ is the trial sample size and $n_{\mathrm{ext}}$ the external sample size.
    Omit the external terms for absent or known calibration and add source-specific terms for multiple sources.
    For binary independent-arm data, CP trial limits give the same fixed-law conclusion.
\end{lemma}

\begin{proof}[Proof of Lemma~\ref{lem:primitive-rates}]
    Bounded second moments control both sampling error and Bernstein widths.
    Let $\epsilon_A>0$ be a lower bound on the known baseline assignment probabilities.
    For $f=\1(J=j)$ or $Y\1(J=j)$,
    \begin{equation*}
        \EE[(\omega_af)^2]
        =\EE_a[f^2/g_a(W)]
        \le\epsilon_A^{-1}p_j
        ,
    \end{equation*}
    with constant one for independent arm samples.
    At positive $p_j$, the sample-mean error is $\Op(\sqrt{p_j/n})$ and the expected unbiased sample variance is at most a constant times $p_j$, so Markov's inequality gives $\hat{v}_j=\Op(p_j)$, including along positive-mass sequences without $np_j\to\infty$.
    At $p_j=0$, the feature and its empirical variance vanish almost surely.
    Thus crossing-feature Bernstein widths are $\Op(\sqrt{p_j/n}+n^{-1})$; overall outcomes and aggregates contribute $\Op(n^{-1/2}+n^{-1})$, and bounded external coordinates contribute their source-specific analogues.
    Summing gives \eqref{eq:rn} for raw errors and widths.

    Coherent projection and interval enlargement retain this rate.
    Write $\hat{\eta}_n^{\obs}$ for the raw observed moments and $\tilde{\eta}_n^{\obs}$ for their projection.
    The true observed vector lies in the nonempty closed convex coherence polytope, and Euclidean projection is nonexpansive, so
    \begin{equation*}
        \|\tilde{\eta}_n^{\obs}-\eta^{\obs}\|_2
        \le
        \|\hat{\eta}_n^{\obs}-\eta^{\obs}\|_2
        .
    \end{equation*}
    Fixed-dimensional norm equivalence gives the same $\ell^1$ rate, the displacement from the raw means is $\Op(r_n)$, and fixed linear aggregates and reconstructed coordinates retain it.
    Every point of an enlarged interval is within the raw estimation error, the original interval width, and the plug-in displacement of the true coordinate, so summation proves \eqref{eq:shrinking-region};
    enlargement also preserves the original confidence event.

    For CP trial limits at a fixed probability in $(0,1)$, the empirical frequency has root-$n$ error.
    With probability tending to one, the beta shapes in \eqref{eq:cp} are proportional to $n$, their means differ from the empirical frequency by $O(n^{-1})$, and their variances are $O(n^{-1})$, so Chebyshev's inequality bounds the quantile deviations by $O(n^{-1/2})$ at fixed tail probabilities.
    At probability zero or one, the explicit boundary formulas give $O(n^{-1})$ widths.
    The same projection and enlargement argument then proves the CP result without a normal approximation at a vanishing probability.
\end{proof}

For fixed dimension and comparable source sizes, \eqref{eq:rn} is $O(n^{-1/2})$ because $\sum_j\sqrt{p_j}\le\sqrt{M}$.
For the conditional-RMS convention \eqref{eq:conditional-budget}, set $G(p)=\rho^2\sum_jp_j$ in the perturbation argument; then $|G(p')-G(p)|\le\rho^2\|p'-p\|_1$, and the same conclusions hold under the corresponding strict budget slack.

\begin{remark}[Failure without feasibility slack]
    \label{rem:feasibility-boundary}
    Positive cell masses alone do not ensure Lipschitz endpoint stability.
    Set $p_1=p_2=1/2$, $b_1=b_2=1/4$, $\mu=1/2$, $c=0$, $l_j=-1$, and $u_j=1$.
    Calibrate $t_1=d$ for $0<d<1/4$.
    The contribution boxes are $[-1/4,1/4]$ and the budget is $2d^2+2t_2^2\le G$.
    At $G_0=2d^2$, only $t_2=0$ is feasible, whereas at $G=G_0+\epsilon$ for $0<\epsilon<1/8$, the sharp interval is
    \begin{equation*}
        [\mu+d-\sqrt{\epsilon/2}
        ,
        \quad
        \mu+d+\sqrt{\epsilon/2}]
        .
    \end{equation*}
    Its endpoint displacement has square-root rather than Lipschitz order, although both cell masses are positive.
    The confidence inclusion in Corollary~\ref{cor:coverage} remains valid at $G_0$.
\end{remark}

\subsection{Augmented observed-moment inference}
\label{app:augmented-moments}

We give the observed-moment expansion for the augmented estimator mentioned in Section~\ref{sec:statistics}.

For arm $a$ and a bounded observed feature $f_q$, define $m_{a,q}(w)=\EE[f_q(O)\mid W=w,A=a]$ and $\zeta_{a,q}=\EE_a[f_q]$.
Partition the $N$ records into a fixed number of folds independently of their values, fit $\hat{g}_a^{(-v(i))}$ and $\hat{m}_{a,q}^{(-v(i))}$ on the folds other than $v(i)$, and set
\begin{equation}
    \label{eq:aipw-score}
    \hat{Z}_{a,q,i}
    =
    \hat{m}_{a,q}^{(-v(i))}(W_i)
    +\frac{\1(A_i=a)}{\hat{g}_a^{(-v(i))}(W_i)}
    \{f_q(O_i)-\hat{m}_{a,q}^{(-v(i))}(W_i)\}
    ,
    \quad
    \hat{\zeta}_{a,q}^{\mathrm{aug}}=N^{-1}\sum_i\hat{Z}_{a,q,i}
    .
\end{equation}
The regression concerns a fully observed feature, not the unobserved no-rescue outcome.

\begin{lemma}[Augmented observed-moment expansion]
    \label{lem:aipw}
    Suppose $g_a$ and all fitted $\hat{g}_a^{(-v)}$ are bounded below by a common positive constant, the features and fitted regressions are uniformly bounded, and, in every fold,
    \begin{equation*}
        \begin{split}
            \|\hat{g}_a^{(-v)}-g_a\|_{L^2(P_W)}
            +\|\hat{m}_{a,q}^{(-v)}-m_{a,q}\|_{L^2(P_W)}
            &=
            \op(1)
            ,\\
            \|\hat{g}_a^{(-v)}-g_a\|_{L^2(P_W)}
            \|\hat{m}_{a,q}^{(-v)}-m_{a,q}\|_{L^2(P_W)}
            &=
            \op(N^{-1/2})
            .
        \end{split}
    \end{equation*}
    Then, for fixed feature dimension,
    \begin{equation}
        \label{eq:aipw-expansion}
        \begin{split}
            \hat{\zeta}_{a,q}^{\mathrm{aug}}-\zeta_{a,q}
            &=
            N^{-1}\sum_i\varphi_{a,q}(O_i)+\op(N^{-1/2})
            ,
            \\
            \varphi_{a,q}(O)
            &=
            m_{a,q}(W)-\zeta_{a,q}
            +\frac{\1(A=a)}{g_a(W)}\{f_q(O)-m_{a,q}(W)\}
            ,
        \end{split}
    \end{equation}
    and the sample covariance of the fitted scores consistently estimates the covariance of these influence functions.
\end{lemma}

\begin{proof}[Proof of Lemma~\ref{lem:aipw}]
    Conditional on a fold's training data, with hats denoting that fold's fits and $Z(\hat{g},\hat{m})$ its score,
    \begin{equation*}
        \EE[Z(\hat{g},\hat{m})\mid\text{training}]-\zeta_{a,q}
        =
        \EE\Biggl[
            \frac{\hat{g}_a-g_a}{\hat{g}_a}
            (\hat{m}_{a,q}-m_{a,q})
            \Biggm|\text{training}
        \Biggr]
        ,
    \end{equation*}
    which Cauchy--Schwarz and the propensity lower bound make $\op(N^{-1/2})$.
    Boundedness and nuisance consistency give the difference between the fitted and true scores conditional $L^2(P)$ norm $\op(1)$.
    Conditional variance bounds therefore make its centered contribution to the overall average $\op(N^{-1/2})$.
    Summing over a fixed number of folds proves \eqref{eq:aipw-expansion} without requiring independence between fitted folds.
    The same conditional $L^2$ argument, boundedness, and the law of large numbers give consistency of the fitted-score second moments.
\end{proof}

The expansion supplies a regular-law simultaneous input region as in \eqref{eq:ps-sandwich-box} when the limiting variances are positive.
Under the stated boundedness and propensity-positivity conditions, consistency of either nuisance suffices for consistency of the moment estimator.
Root-$N$ inference uses the stronger product-rate and nuisance-consistency conditions of Lemma~\ref{lem:aipw}.
Augmented scores are neither Bernoulli nor independent after fitting, so Clopper--Pearson and unconditional empirical-Bernstein limits do not apply directly.

\section{Detailed numerical experiments}
\label{app:detailed-numerical}

This appendix specifies the generating laws, comparator calculations, and numerical settings used in Section~\ref{sec:numerical}.
The additional results distinguish scientific restrictions, input-region construction, and implementation diagnostics.
Scientific comparisons use each class's own sharp interval;
inferential comparisons hold the scientific class fixed.

\subsection{Designs and comparators}
\label{app:numerical-setup}

We give the reference and additional designs, then specify how each confidence procedure is calculated and evaluated.

\subsubsection{Finite-state reference law}
\label{app:finite-protocol}

We specify the finite-state law and its sample-size, calibration, and active-cap variants.
Write $\expit(x)=(1+e^{-x})^{-1}$.
Set $S_0=(s-1)/2$ and $L_k=S_k$.
At an at-risk visit, state $s$ is reached from previous state $v$ with probability
\begin{equation*}
    h_{a,k}(v)
    =
    \expit\{\alpha_0+0.6(v/(s-1)-1/2)+0.15(k-1)-0.10a\}
    ,
\end{equation*}
and reaching it triggers rescue.
Conditional on not triggering, the probability of $l\in\{0,\ldots,s-1\}$ is proportional to
\begin{equation*}
    \exp\Biggl[
        -\frac{\{l-0.6v-0.25(s-1)\}^2}{2\{0.35(s-1)+0.3\}^2}
    \Biggr]
    .
\end{equation*}
Observed post-crossing states equal $s$.
For non-crossers, both outcomes are the same Bernoulli variable with mean $0.25+0.35S_K/(s-1)-0.05a$.
At crossing visit $k$, the observed and no-rescue means are
\begin{equation*}
    q^R_{a,k}=0.20+0.08k/K-0.04a
    ,
    \quad
    q^{\NR}_{a,k}=q^R_{a,k}+\tau_k
    ,
    \quad
    \tau_k=0.08+0.12(K-k)/K
    ,
\end{equation*}
and a common independent uniform generates the two outcomes.
Individual rescue benefits are therefore nonnegative under this generating law, but this is not an assumption of the cell-average model.
Exact state recursion and root finding set $\alpha_0$ so that the equally arm-averaged rescue probability is $0.5$ or $0.75$.
Exact recursion also gives the joint distribution of $(J,Y,Y^{a,\NR})$, and sampling from it is equivalent to sampling trajectories for all implemented moment procedures.

The primary design uses both arms, rescue probabilities $0.5$ and $0.75$, and $n=1000$ per arm.
The contraction design uses arm zero, rescue probability $0.5$, exact calibration, and $n\in\{50,100,250,500,1000,2500,10000\}$.
The active-cap design changes only the odd-visit crossing means to $q^R_{0,k}=0.96$ and $\tau_k=0.02$, with $n\in\{50,250,1000,2500,10000\}$.
The oracle-calibrated class uses
\begin{equation*}
    H_{1j}=\1(j\in\{1,3\})
    ,
    \quad
    H_{2j}=\1(j=5)
    ,
    \quad
    d=H(p_j\tau_j)_j
    ,
\end{equation*}
which are population contributions rather than effects conditional on crossing.
The uncalibrated model is a separate scientific class.
Population sensitivity curves use $\Gamma=0,0.01,\ldots,0.35$ and $\lambda=1,1.05,\ldots,2.5$, and the radius experiment pairs $\Gamma\in\{0,0.05,0.10,0.20,0.35\}$ within datasets at $n\in\{50,1000,10000\}$.
Membership of the generating law in a class is checked from the generating law itself, not from the analysis data.

\subsubsection{Scientific comparators}
\label{app:scientific-comparators}

We define the alternative scientific classes and calculate the sharp interval and confidence construction for each.

\paragraph{Time-only restrictions.}
The time-only comparator restricts selection across crossing strata after coarsening the history.
Fix an arm with $P_a(T>K)>0$ and define
\begin{equation*}
    e_k=P_a(T>k\mid T\ge k)
    ,
    \quad
    v_k=P_a(Y^{a,\NR}=1\mid T\ge k)
    ,
    \quad
    q_k^{\NR}=P_a(Y^{a,\NR}=1\mid T=k)
    ,
\end{equation*}
where the last quantity is needed only for positive crossing mass, and set $v_{K+1}=P_a(Y=1\mid T>K)$, which temporal consistency identifies.
The binary time-only specialization of the outcome-density-ratio restrictions in \citet{tan2025sensitivity} imposes
\begin{equation}
    \label{eq:tan-ratios}
    \lambda^{-1}
    \le
    \frac{P_a(Y^{a,\NR}=y\mid T=k)}{P_a(Y^{a,\NR}=y\mid T>k)}
    \le
    \lambda
    ,
    \quad y\in\{0,1\}
    ,
\end{equation}
for $\lambda\ge1$.
Conditioning retains visit and prior non-rescue but omits the current trigger and the measured state history, so these are selection restrictions after coarsening rather than exchangeability at the full triggering history.
At zero denominators use the equivalent inequalities below.

Monotone envelopes give a sharp backward recursion.
Define
\begin{equation*}
    a_\lambda(v)=\max\{v/\lambda,1-\lambda(1-v)\}
    ,
    \quad
    b_\lambda(v)=\min\{\lambda v,1-(1-v)/\lambda\}
    .
\end{equation*}
Then \eqref{eq:tan-ratios} is equivalent to $a_\lambda(v_{k+1})\le q_k^{\NR}\le b_\lambda(v_{k+1})$, and the mixture identity $v_k=e_kv_{k+1}+(1-e_k)q_k^{\NR}$ gives
\begin{equation}
    \label{eq:tan-recursion}
    \begin{split}
        v_k^-&=e_kv_{k+1}^-+(1-e_k)a_\lambda(v_{k+1}^-)
        ,\\
        v_k^+&=e_kv_{k+1}^++(1-e_k)b_\lambda(v_{k+1}^+)
        ,
    \end{split}
\end{equation}
initialized at $v_{K+1}^-=v_{K+1}^+=v_{K+1}$; at zero crossing mass, $e_k=1$ and the recursion copies the later mean.
Backward endpoint choices assign means to disjoint crossing strata, Lemma~\ref{lem:full-data-completion} realizes them jointly, and the mixture identities preserve the earlier restrictions.
At fixed observed probabilities the mass-coordinate constraints are linear, so convexity fills $[v_1^-,v_1^+]$; an independently assembled linear program checks the recursion.

Conditional binomial limits provide inference for this interval.
Let $N_k$ count $T\ge k$, $C_k$ count $T>k$, and $N_{K+1},S_{K+1}$ be the terminal non-crosser count and its outcome successes.
Conditionally on their denominators, $C_k$ and $S_{K+1}$ are binomial with probabilities $e_k$ and $v_{K+1}$.
Use $K+1$ Clopper--Pearson intervals with error $\alpha/(K+1)$ and $[0,1]$ for a zero denominator; conditional coverage and a union bound give simultaneous coverage.
The lower recursion increases and the upper recursion decreases with $e_k$, and both increase with the terminal mean, so lower limits for every $e_k$ enter both outward recursions together with the respective terminal limit.
This conditional-CP procedure is the comparator labelled Conditional CP.

\paragraph{Range, sign, and box restrictions.}
Write $b_0=\mu-\sum_jb_j$.
The range-only, cell-average sign-only, and signed $\delta$-box intervals are, respectively,
\begin{equation*}
    [b_0,b_0+\textstyle\sum_jp_j]
    ,
    \quad
    [\mu,b_0+\textstyle\sum_jp_j]
    ,
    \quad
    \Bigl[
        \mu,\mu+\sum_j\min\{\delta p_j,p_j-b_j\}
    \Bigr]
    .
\end{equation*}
When $\delta=\Gamma$ the box implies the signed zero-center budget, since $\sum_{p_j>0}t_j^2/p_j\le\delta^2\sum_jp_j\le\Gamma^2$; a narrower box interval therefore reflects a stronger restriction rather than more efficient estimation of the same set.
For binary outcomes the range endpoints are means of $Y\1(J=0)$ and $Y\1(J=0)+\1(J>0)$, and the sign-only lower endpoint is $\mu$; a lower and an upper one-sided Clopper--Pearson limit, each with error $\alpha/2$, give the whole-set interval labelled Endpoint CP.
The $\delta$-box uses Cell CP with its effect limits and without a budget.

\subsubsection{Inference comparators}
\label{app:inference-comparators}

We specify the inferential comparisons for the fixed signed-budget class.

\paragraph{Coordinatewise regions.}
The coordinatewise constructions protect the same observed features under different marginal confidence limits.
For binary independent-arm data, the feature vector is
\begin{equation*}
    \Bigl(
        Y,
        (\1(J=j))_{j=1}^M,
        (Y\1(J=j))_{j=1}^M,
        \1(J>0),
        Y\1(J=0)
    \Bigr)^\top
    ,
\end{equation*}
with $D=2M+3$ coordinates.
Cell CP, Cell KL, empirical Bernstein, Hoeffding, and Wald allocate error $\alpha/D$ to each coordinate.
Appendix~\ref{pf:cor-coverage} gives the finite-sample justifications for CP, empirical Bernstein, and Hoeffding limits.
The KL guarantee and regular-law calibrations are specified below.
For Bernoulli probabilities define $\KL(v\|p)=v\log(v/p)+(1-v)\log\{(1-v)/(1-p)\}$ with the usual boundary conventions.
The coordinatewise Bernoulli--KL region imposes
\begin{equation}
    \label{eq:kl-region}
    \KL(\hat{\zeta}_q\|\zeta_q)
    \le\log(2D/\alpha)/n
    ,
    \quad q=1,\ldots,D
    .
\end{equation}
Binomial Chernoff tails and a union bound give simultaneous coverage at level $1-\alpha$.
At empirical zero, the upper limit is $1-\exp\{-\log(2D/\alpha)/n\}$.
These regions use simultaneous limits for observed probabilities, as in \citet{duarte2024automated}, constructed here for the reduced model.
Atomic variants protect the $2(M+1)$ atoms of $(J,Y)$ directly, either alone or together with the aggregates $Y$, $\1(J>0)$, and $Y\1(J=0)$, using coordinatewise KL limits.
They are coordinatewise regions, distinct from the joint multinomial likelihood region \eqref{eq:joint-lr}.

At a fixed law with all declared atoms positive, the multinomial likelihood-ratio statistic in \eqref{eq:joint-lr} has the chi-square limit used there \citep[Chapter~16]{vandervaart1998asymptotic}.
Let $\Phi$ denote the standard normal distribution function.
For a sample-mean feature with positive limiting variance, the central limit theorem and a consistent variance give Wald limits of half-width $\Phi^{-1}(1-\alpha/(2D))\sqrt{\hat{v}/n}$.
A union bound gives simultaneous asymptotic coverage.
Neither calibration is uniform along rare-cell sequences, because a positive population variance can coexist with zero empirical variance.

\paragraph{Joint likelihood regions.}
The joint-region comparison evaluates Cell CP, the split-likelihood region \eqref{eq:joint-split}, the likelihood-ratio region \eqref{eq:joint-lr} with $d_\pi=12$ and no empirical deletion of atoms, and Aggregate CP.
There are $500$ paired datasets in each of three settings: the reference law at $n=1000$ and the active-cap law at $n=1000$ and $n=10000$.
Each dataset consists of two independent multinomial count vectors of sizes $\lfloor n/2\rfloor$ and $n-\lfloor n/2\rfloor$.
Their sum is used by Cell CP, joint LR, and Aggregate CP, and the split region uses the two halves with the training probabilities of Section~\ref{sec:statistics}.
All four procedures use $\alpha=0.05$, $\Gamma=0.2$, and the same population reference.
These datasets are distinct from those of Tables~\ref{tab:scientific}--\ref{tab:inference}.
Therefore, paired differences are computed within Table~\ref{tab:joint-regions}.
Comparisons across the two experiments are descriptive rather than paired.

\paragraph{Aggregate relaxation.}
The aggregate comparator bounds the total effect contribution without retaining the cellwise outcome caps.
For the uncalibrated signed zero-center budget, Cauchy--Schwarz gives
\begin{equation}
    \label{eq:aggregate-bound}
    0
    \le
    \sum_j t_j
    \le
    \Gamma\sqrt{\sum_jp_j}
    ,
    \quad
    \cI(\eta,\Gamma^2)
    \subseteq
    \Biggl[
        \mu,
        \min\Biggl\{
            1,\mu+\Gamma\sqrt{\sum_jp_j}
        \Biggr\}
    \Biggr]
    .
\end{equation}
Aggregate CP uses lower and upper one-sided CP limits for $\mu$ and an upper limit for $\sum_jp_j$, each with error $\alpha/3$.
For positive crossing mass, a sufficient condition for the displayed population relaxation to be sharp is that
\begin{equation*}
    t_j
    =
    \frac{\Gamma p_j}{\sqrt{\sum_jp_j}}
\end{equation*}
satisfies every cellwise outcome cap.
With no crossings, the population interval is $\{\mu\}$.
When the relaxation is wider than the sharp interval, its excess width persists even with known population moments.

\paragraph{Endpoint bootstrap.}
The endpoint bootstrap approximates the joint sampling error of the two optimized bounds.
Resample patient records and recompute the empirical moments and exact capped endpoints by water filling.
For binary outcomes, multinomial resampling of the $(J,Y)$ counts is equivalent to resampling records.
For plug-in endpoints $(\hat{L},\hat{U})$ and bootstrap endpoints $(\hat{L}^\ast,\hat{U}^\ast)$, let $c_n$ be the nonnegative empirical $(1-\alpha)$ quantile of
\begin{equation*}
    \max\{\hat{L}^\ast-\hat{L},\hat{U}-\hat{U}^\ast\}
    .
\end{equation*}
The reported interval is $[\hat{L}-c_n,\hat{U}+c_n]\cap[0,1]$.
Use $999$ resamples and the upper order-statistic quantile convention.
The regular-law justification requires differentiable endpoints and a locally stable active set \citep[Chapter~23]{vandervaart1998asymptotic}.
The rare-cell diagnostic examines a setting outside these conditions.
Two-arm comparisons use error $0.025$ per arm before combining endpoints.

\subsubsection{Computational details}
\label{app:numerical-audits}

We describe population references, independent reduction checks, and numerical controls.

\paragraph{Population references.}
For the signed zero-center budget the lower endpoint is $\mu$, and the maximizing contributions have the water-filling form
\begin{equation}
    \label{eq:waterfill}
    t_j=\min\{p_j-b_j,\kappa p_j\}
    ,
    \quad
    \kappa\ge0
    ,
\end{equation}
where the full vector of caps is used if it is within budget and otherwise $\kappa$ solves $\sum_{p_j>0}t_j^2/p_j=G$, with zero contributions on zero-mass cells.
For $\kappa>0$, each positive-mass coordinate maximizes $t_j-t_j^2/(2\kappa p_j)$ on its box.
Summing these inequalities with the active budget proves optimality.
At $G=0$, $t=0$.
Successively fixing saturated cells gives a finite calculation, which the endpoint bootstrap also uses.
For disjoint calibration groups with masses $P_l$, fixed contributions $d_l$, and uncalibrated mass $P_U$, the endpoints are
\begin{equation*}
    \mu+\sum_ld_l
    ,
    \quad
    \mu+\sum_ld_l+
    \sqrt{\Biggl(
        \Gamma^2-\sum_{l:P_l>0}d_l^2/P_l
    \Biggr)P_U}
    ,
\end{equation*}
provided the residual budget is nonnegative and the group-constant and uncalibrated maximizing effects satisfy every cap.
A zero-mass group requires $d_l=0$.

\paragraph{Independent representations.}
The full-history check enumerates, for $(K,s)\in\{(3,4),(5,4),(5,6)\}$, the crossing histories split by observed outcome.
For observed atom mass $a_h$, let $u_h\in[0,a_h]$ be its joint mass with $Y^{a,\NR}=1$.
Substitute $t_j=\sum_{h\in j}u_h-b_j$ into the identical sign, shared-budget, and tangent constraints, with objective $b_0+\sum_hu_h$.
The resulting program has the same extrema as the reduced program because every reduced vector lifts by allocating its cell contribution across observed atoms.
The two programs are assembled separately, and the reported dimensions exclude epigraph variables.
The independent-budget relaxation, with width $\sum_j\min\{p_j-b_j,\Gamma\sqrt{p_j}\}$, and a mass-coordinate linear program for the time-only recursion serve as further checks.

\paragraph{Tangent-grid accuracy.}
The tangent grid starts at $h=1/128$ and is refined dyadically at most twice in the finite-state and continuous-history comparisons when an available inner-primal/outer-certificate gap exceeds $10^{-4}$.
The observational, partition, and radius experiments retain $h=1/128$.
The deterministic mesh check uses $h\in\{1/8,1/16,1/32,1/64,1/128\}$, and zero budget uses $t=0$ exactly.
Linear programs are solved with primal and dual feasibility tolerances $10^{-9}$.
These tolerances are not substituted for outward certificates.

\paragraph{Joint-likelihood optimization.}
Write the log criterion of a joint region as
\begin{equation*}
    \ell(\pi)
    =
    \log\Biggl\{
        B^{-1}\sum_{v=1}^B
        \exp\Biggl(
            C_v-\sum_iN_{vi}\log\pi_i
        \Biggr)
    \Biggr\}
    .
\end{equation*}
For the split region, use $B=2$, $C_v=\sum_iN_{vi}\log\tilde{\pi}_{-v,i}$, and threshold $\log(1/\alpha)$.
For the LR region, use $B=1$, $N_{1i}=N_i$, $C_1=\sum_iN_i\log\hat{\pi}_i$, and threshold half the chi-square quantile in \eqref{eq:joint-lr}.
Zero-count terms are zero.
For a strictly positive vector $v$ and probability weights $w_b$, the log-sum-exp variational inequality and the tangent inequality for $-\log$ give the affine minorant
\begin{equation}
    \label{eq:likelihood-cut}
    \ell(\pi)
    \ge
    \sum_bw_bC_b
    +\sum_i\bar{N}_i(1-\log v_i-\pi_i/v_i)
    -\sum_bw_b\log w_b-\log B
    ,
    \quad
    \bar{N}_i=\sum_bw_bN_{bi}
    .
\end{equation}
Requiring the minorant to be at most the threshold is an outer relaxation for any $w$.
Any finite collection of such cuts preserves outer inclusion.
The implementation uses dyadic weights, evaluates logarithms in interval arithmetic with $45$ decimal digits, and encloses coefficient-rounding error on the probability box.
The joint comparison fixes $h=1/128$, permits at most $120$ cuts per endpoint, targets an optimizer-based gap of $10^{-4}$, and uses LP feasibility tolerances $10^{-8}$.
A smooth constrained solve supplies candidate cut locations without serving as a certificate.
The certified outer value is retained even if the cut budget is exhausted, with a stopping flag recorded.

\paragraph{Numerical reporting.}
The reported optimization values use outward objective bounds rather than feasible primal values alone.
The LP certificate in Lemma~\ref{lem:outward-certificate} is evaluated in rational arithmetic for the supplied floating-point coefficients, with correctly signed nonnegative multipliers and outward final rounding.
When intended coefficients differ from the supplied coefficients by $(\Delta B_{j\cdot},\Delta b_j)$, a uniform bound on the absolute change in row $j$'s residual is
\begin{equation*}
    \epsilon_j
    =
    \sum_i|\Delta B_{ji}|\max\{|a_i|,|z_i|\}
    +|\Delta b_j|
    .
\end{equation*}
Rows are relaxed outward by these allowances before certification.
Equivalently, for fixed nonnegative multipliers $\lambda_j$ and the same variable box, subtracting $\sum_j\lambda_j\epsilon_j$ from the supplied-program lower certificate protects against these row perturbations.
These coefficient allowances are separate from the perspective approximation error $e_h$.
Inner residual checks and optimizer-based endpoint gaps, when computed, are diagnostic quantities distinct from outward certificates.
CP, KL, and chi-square quantile calculations use a $10^{-12}$ outward guard rather than verified special-function enclosures.
The optimization certificates are therefore conditional on the supplied input limits.
An empty confidence program or unavailable outward certificate gives the physical-range fallback, with the cause recorded.
Reaching a refinement or cut limit with an available outward certificate retains that outward value and records a stopping flag.
All such runs remain in the summaries, and the scientific budget is never increased to obtain feasibility.

\subsubsection{Continuous histories}
\label{app:continuous-design}

We generate continuous histories to check the reduced analysis with correlated innovations and longer follow-up.

The continuous-history configurations are
\begin{equation*}
    (K,d_W,\rho)
    \in
    \{(5,5,0),(5,5,0.65),(5,50,0.35),(8,5,0.35)\}
    ,
\end{equation*}
with arm zero and $n=1000$.
Draw $W\sim N_{d_W}(0,\Sigma)$ with $\Sigma_{ij}=0.35^{|i-j|}$, initialize $S_0^{\NR}=0.4W_1-0.2W_2$ and $B_0^{\NR}=0.3W_1+0.2W_3$, and let $(\epsilon_{S,k},\epsilon_{B,k})$ be independent visit-specific bivariate standard-normal innovations with correlation $\rho$.
The no-rescue states obey
\begin{equation}
    \label{eq:dgp}
    \begin{split}
        S_k^{\NR}
        &=0.55S_{k-1}^{\NR}-0.25a+0.15W_1
        +0.10\sin S_{k-1}^{\NR}+0.60\epsilon_{S,k}
        ,\\
        B_k^{\NR}
        &=0.45\tanh B_{k-1}^{\NR}+0.25\tanh S_{k-1}^{\NR}-0.15a
        +0.10\tanh W_2+0.60\epsilon_{B,k}
        ,
    \end{split}
\end{equation}
and first crossing occurs at $B_k^{\NR}\ge c_0+0.05k/K$.
Observed states copy this path through first crossing; later observed states do not enter the moment procedures.
With $\bar S^{\NR}=K^{-1}\sum_kS_k^{\NR}$, set
\begin{equation*}
    \lambda_a
    =-0.20+0.50S_K^{\NR}+0.20\bar S^{\NR}-0.20a+0.15W_3
    +0.10\sin W_1+0.10W_4W_5
    ,
    \quad
    q^{\NR}=\expit(\lambda_a)
    ,
\end{equation*}
and, with $e_T=B_T^{\NR}-(c_0+0.05T/K)$ on crossing paths,
\begin{equation*}
    q^R
    =
    \begin{cases}
        q^{\NR},&T>K,\\
        \expit\{\lambda_a-0.7(1+0.50(K-T)/K+0.25e_T)\},&T\le K.
    \end{cases}
\end{equation*}
A common independent uniform generates both outcomes.
Only the first five baseline coordinates enter these equations; $d_W=50$ adds correlated, conditionally irrelevant covariates.
The trigger $c_0$ is chosen from $100000$ independent calibration paths, with common random numbers, to give approximately $0.5$ rescue probability, and is then held fixed.
The comparisons use $M=4K$ cells combining visit, $\1(W_1\ge0)$, and $\1(e_T>0.5)$, and report Cell CP under the signed $\Gamma=0.2$ budget.
Population references average conditional Bernoulli probabilities over $2\times10^6$ independent paths per law in $20$ batches of $100000$; batch standard errors are reported separately from replication error, and moving the reference endpoints by $\pm2.58$ batch standard errors serves as a sensitivity diagnostic for reference uncertainty.

\subsubsection{Crossing partitions}
\label{app:partition-protocol}

We vary prespecified crossing partitions while holding the generating observations and sensitivity scale fixed.
The partition design fixes $(K,d_W,\rho)=(5,5,0.35)$ and $n\in\{250,1000,10000\}$.
For $M=1$ all crossings are pooled; for $M=5B$, $B\in\{1,2,4,8,16\}$, visit is crossed with $W_1$ categories having fixed cutpoints $\Phi^{-1}(b/B)$, $b=1,\ldots,B-1$, where $\Phi$ is the standard normal distribution function.
The heterogeneous-cap law sets $q^R=0.98$ and $q^{\NR}=0.99$ on $\{T\le K,W_1\ge0\}$ and leaves the other outcome probabilities unchanged.
The original law retains the outcome probabilities in Appendix~\ref{app:continuous-design};
both analyses impose the cellwise outcome caps.
Both laws share state paths and outcome uniforms, all partitions share observations, and the zero center, sign limits, and $\Gamma=0.2$ remain fixed, with membership of the generating law checked for each partition.
These classes are nested under refinement: for fine cells $j$ within coarse cell $g$, Cauchy--Schwarz gives
\begin{equation}
    \label{eq:partition-jensen}
    \frac{(\sum_{j\in g}t_j)^2}{\sum_{j\in g}p_j}
    \le
    \sum_{j\in g:p_j>0}t_j^2/p_j
    ,
\end{equation}
with both sides zero for a zero-mass coarse cell, and fine-cell sign and outcome restrictions imply the coarse ones.
Every partition is assessed against its own sharp reference, so its gap from $M=80$ measures scientific coarsening rather than regression bias.

\subsubsection{Observational design}
\label{app:obs-dgp}

We vary baseline treatment assignment and propensity estimation while retaining the same standardized causal reference.
Let $W_1,\ldots,W_5$ be independent $\Bern(1/2)$ variables and put $x_j=W_j-1/2$.
Write $\logit(x)=\log\{x/(1-x)\}$ for $0<x<1$.
The assignment laws are
\begin{equation}
    \label{eq:obs-propensity}
    \begin{split}
        \logit e_{\mathrm{lin}}(W)
        &=-0.1+0.8x_1-0.6x_2+0.4x_3
        ,\\
        \logit e_{\mathrm{nonlin}}(W)
        &=\logit e_{\mathrm{lin}}(W)+3x_1x_2+2x_4x_5
        ,
    \end{split}
\end{equation}
with $A\mid W\sim\Bern\{e(W)\}$ drawn independently of the potential-process innovations.
Both probabilities are bounded away from zero and one on the $32$-point support, and pooled sizes are $N\in\{50,250,1000,10000\}$.
Baseline risk enters the rescue law of the $K=5$, $s=4$ reference model by adding $0.65x_1+0.35x_2$ to the trigger logit.
A common intercept gives an equally arm-averaged, baseline-standardized rescue probability of $0.5$.
With $\beta_W=1.4x_1-0.8x_2+2.4x_1x_2+1.6x_4x_5$, the outcome means are
\begin{equation*}
    \begin{split}
        q^R_{a,k}(W)&=\expit(-0.8+0.08k+\beta_W-0.2a)
        ,\\
        \tau_k(W)&=\{0.08+0.12(K-k)/K\}(0.7+0.3W_1)
        ,\\
        r_a(v,W)&=\expit\{-0.2+0.6v/(s-1)+\beta_W-0.2a\}
        .
    \end{split}
\end{equation*}
At crossing $q^{\NR}_{a,k}=q^R_{a,k}+\tau_k$.
Otherwise both outcomes use $r_a(S_K,W)$, and a common uniform generates them.
The generating law satisfies the signed $\Gamma=0.2$ budget, and arm-zero analysis uses visit cells without calibration.
Summation over the $32$ baseline patterns gives exact standardized references, identical under the two assignment laws.

The correctly specified logistic fit uses $\tilde{W}=(1,x_1,\ldots,x_5,x_1x_2,x_4x_5)^\top$ in both scenarios, and the main-effects model is correct only under linear assignment.
An unpenalized fit is accepted when it is finite, its coefficient magnitudes are below $25$, its maximum normalized score component is below $\min(10^{-6},N^{-1})$, and its information condition number is below $10^{10}$.
Failed fits give the physical-range fallback.
Fitted probabilities are clipped to $[0.02,0.98]$, which contains the true probabilities in its interior, and clipping frequency is recorded.
The forest uses two-fold stratified cross-fitting and $200$ trees, with all five covariates at each split, no maximum-depth restriction, and minimum leaf size $\max\{5,\lceil\sqrt{N_{\mathrm{train}}}/2\rceil\}$, where $N_{\mathrm{train}}$ is the training-fold sample size.
Predictions come from the other fold and use the same clipping.
For fitted assignment probabilities $\hat e(W_i)$, define $\hat{\omega}_{0,i}=(1-A_i)/\{1-\hat e(W_i)\}$.
Propensity root-mean-square error, clipping frequency, and the effective sample size $(\sum_i\hat{\omega}_{0,i})^2/\sum_i\hat{\omega}_{0,i}^2$ are recorded as diagnostics.

\subsubsection{Estimated-propensity inference}
\label{app:estimated-ps}

We specify the observed-moment estimators and uncertainty calculations used in the observational experiment.

\paragraph{Inverse weighting and logistic sandwich.}
The experiments estimate each of the $D=13$ observed feature means $\zeta_q=\EE_0[f_q]$ by
\begin{equation*}
    \hat{\zeta}_q
    =N^{-1}\sum_{i=1}^N\frac{1-A_i}{1-\hat{e}(W_i)}f_q(O_i)
    ,
\end{equation*}
which are unnormalized standardized means.
Oracle EB uses the known probabilities and independent bounded scores of range at most $1/\min_W\{1-e(W)\}$.
Oracle Wald uses the regular-law normal approximation.
For a correctly specified $e_\beta(W)=\expit(\tilde{W}^\top\beta)$, assume a finite interior coefficient, nonsingular information, baseline positivity, fixed-dimensional regularity, normalized score residual $\op(N^{-1/2})$, and eventually inactive clipping.
Write $e_i=e_\beta(W_i)$ at the truth and set
\begin{equation}
    \label{eq:ps-if}
    \begin{split}
        s_i&=\tilde{W}_i(A_i-e_i)
        ,
        \quad
        I_\beta=\EE[e_i(1-e_i)\tilde{W}_i\tilde{W}_i^\top]
        ,\\
        a_q&=\EE[\omega_{0,i}e_i f_q(O_i)\tilde{W}_i]
        ,
        \quad
        \phi_{q,i}=\omega_{0,i}f_q(O_i)-\zeta_q+a_q^\top I_\beta^{-1}s_i
        .
    \end{split}
\end{equation}
The derivative of the weighted feature is $\omega_0e_\beta f_q\tilde{W}$, and the coefficient influence function is $I_\beta^{-1}s$.
The smooth estimating-equation expansion gives \eqref{eq:ps-if} \citep[Chapter~5]{vandervaart1998asymptotic}.
All components are estimated on the same observations to retain their covariance.
When clipping is active, the implemented weight derivative is zero outside the unclipped interval.
The logistic sandwich region is
\begin{equation}
    \label{eq:ps-sandwich-box}
    |\zeta_q-\hat{\zeta}_q|
    \le
    \Phi^{-1}(1-\alpha/(2D))\hat{\sigma}_q/\sqrt{N}
    ,
    \quad
    \hat{\sigma}_q^2
    =
    \frac{1}{N-1}\sum_i(\hat{\phi}_{q,i}-\bar{\hat{\phi}}_q)^2
    .
\end{equation}
It has simultaneous asymptotic input coverage at fixed laws with positive limiting variances, which Corollary~\ref{cor:coverage} transfers to whole-set coverage.
Under correct specification $\cov(\omega_0f_q,s)=-a_q$ and $\Var(s)=I_\beta$, so $\Var(\phi_q)=\Var(\omega_0f_q)-a_q^\top I_\beta^{-1}a_q$.
This explains why propensity estimation can reduce variance.
The naive logistic variance omits this term, and the sandwich does not remove misspecification bias.
Forest naive uses fitted-weight sample variances, and Forest EB inserts the clipped range bound $50$.
Neither controls fitted-weight dependence or nuisance bias by itself.

The observational tables evaluate the inverse-weighting procedures specified above.
Appendix~\ref{app:augmented-moments} gives an augmented alternative and its regular-law conditions.
All evaluated procedures project incompatible observed moments onto the coherence polytope and enlarge their input intervals to contain the coherent plug-in.
This preserves an already valid input event but does not repair invalid nuisance inference.

\subsubsection{External and rare-cell designs}
\label{app:external-rare}

We separately examine uncertainty in external calibration and uncertainty about an absent empirical crossing category.

\paragraph{External calibration.}
The external experiment is an idealized transport-valid benchmark for the propagation of calibration uncertainty.
It uses the arm-zero law, $n=1000$, external sizes $250$, $2000$, and $20000$, and no-calibration and exact-calibration benchmarks.
Source records follow the same standardized pre-crossing population, including non-crossers.
At crossing, the entire continuation is randomized equally between no rescue ($E^{\mathrm{ext}}=0$) and rescue ($E^{\mathrm{ext}}=1$).
With $H_{l,0}=0$, the score $V_l^{\mathrm{ext}}=H_{l,J}\{2\1(E^{\mathrm{ext}}=0)Y^{\mathrm{ext}}-2\1(E^{\mathrm{ext}}=1)Y^{\mathrm{ext}}\}$, interpreted as zero for non-crossers, has range $[-2,2]$ and mean $\sum_jH_{lj}t_j=d_l$ by randomization.
Including non-crossers retains the trial-population contribution weights.
Trial Clopper--Pearson and external Bernstein intervals are combined under one coordinatewise error allocation, and coverage is evaluated for the class with the true calibration target.

\paragraph{Rare cells.}
The rare-cell diagnostic isolates uncertainty about an absent empirical category.
Let $n=1000$, $np\in\{0,0.5,2,10\}$, and suppose the observed outcome is known to be identically zero.
Only the Bernoulli crossing probability is estimated, and the range-only sharp interval is $[0,p]$.
Use the upper endpoints of two-sided $95\%$ CP, KL, empirical Bernstein, Hoeffding, and Wald intervals, with the lower target endpoint fixed at zero.
The first four constructions have upper-tail error at most $0.025$.
Wald uses the corresponding normal critical value.
The basic bootstrap instead uses a one-sided $95\%$ upper limit with $999$ resamples.
These choices assess collapse at zero counts rather than relative efficiency under a common one-sided calibration.
Record whole-set coverage, mean upper endpoint, and zero-count frequency over $20000$ replications.

\subsubsection{Bounded continuous outcomes}
\label{app:continuous-outcome}

We replace binary outcomes by bounded continuous outcomes while retaining the reference conditional means.
Use the arm-zero finite-state reference law with $K=5$, $s=4$, equally arm-averaged rescue probability $0.5$, and $M=5$ visit cells.
Let $Q_\phi(u;q)$ be the $u$-quantile of the beta distribution with shape parameters $\phi q$ and $\phi(1-q)$.
Fix $\phi=10$ before generating the analysis datasets.
For an independent $U\sim\operatorname{Uniform}(0,1)$, define
\begin{equation*}
    (Y,Y^{a,\NR})
    =
    \begin{cases}
        \bigl(
            Q_\phi(U;r_a(S_K)),
            Q_\phi(U;r_a(S_K))
        \bigr),
        &T>K,
        \\
        \bigl(
            Q_\phi(U;q^R_{a,k}),
            Q_\phi(U;q^R_{a,k}+\tau_k)
        \bigr),
        &T=k,
    \end{cases}
\end{equation*}
where $r_a(v)=0.25+0.35v/(s-1)-0.05a$ and the crossing means are those of Appendix~\ref{app:finite-protocol}.
The common outcome on non-crossing paths enforces temporal consistency.
Sampling retains the terminal-state distribution among non-crossers, rather than replacing its mixture of beta distributions by a single beta distribution with the marginal mean.

The analysis uses the signed $\Gamma=0.2$ budget without a beta outcome model.
Because every conditional outcome mean and crossing probability is unchanged, the population inputs $(\mu,p,b)$ and the sharp interval coincide with those of the binary reference law.
Population endpoints therefore use the same exact reference calculation and are independent of $\phi$.
This comparison changes sampling variability while holding identification uncertainty fixed.

The two finite-sample constructions combine CP probability limits with bounded-variable outcome limits.
Hybrid EB--CP uses empirical Bernstein limits for $Y$, $Y\1(J=j)$, and $Y\1(J=0)$, and CP limits for $\1(J=j)$ and $\1(J>0)$.
Hybrid Hoeffding--CP substitutes Hoeffding limits for those outcome features.
Both allocate error $\alpha/(2M+3)$ to each of the same $2M+3$ coordinates, with $\alpha=0.05$.
Cell Wald uses sample-variance Wald limits for all coordinates with the same allocation.
The endpoint bootstrap resamples the observed patient-level $(J,Y)$ pairs and recomputes their empirical moments; binary outcome-count resampling is not used.
The discrete joint likelihood regions are not included because the experiment does not assume a finite outcome support.
Use $n\in\{250,1000,10000\}$ and $500$ paired replications per sample size.
The coordinatewise projections use the fixed grid $h=1/128$, and the bootstrap uses $999$ resamples.

\subsubsection{Reporting conventions}
\label{app:reporting}

We retain every replication and report uncertainty in the simulation summaries.
There are $500$ replications per finite-state, external-calibration, radius, joint-region, and bounded continuous-outcome setting, $250$ per observational, partition, and continuous-history setting, and $20000$ per rare-cell setting.
Procedures share datasets within each comparison, and sensitivity radii and partitions share datasets as specified above.
Every bootstrap comparison uses $999$ resamples.

For population interval $[L,U]$ and reported interval $[\hat{L},\hat{U}]$, set coverage is the mean of $\1(\hat{L}\le L,\hat{U}\ge U)$ across replications.
Target coverage is the mean of $\1(\hat{L}\le\theta_a\le\hat{U})$, with the generating contrast substituted in two-arm analyses.
Enlargement is $(\hat{U}-\hat{L})-(U-L)$ and is retained with its sign.
For Aggregate CP it includes any population relaxation gap, in addition to sampling and numerical effects.
Plug-in error is the Hausdorff distance between the plug-in and population sharp intervals, rather than the distance from the confidence interval.

Mean summaries use Monte Carlo standard errors (MCSEs), calculated as the replication-level sample standard deviation divided by the square root of the number of replications.
Coverage and failure frequencies use exact binomial Monte Carlo (MC) intervals.
Paired comparisons use within-dataset differences and their MCSEs.
An all-success result is accompanied by its binomial interval rather than interpreted as population coverage equal to one.

\subsection{Additional results}
\label{app:numerical-results}

This subsection reports the results referred to in Section~\ref{sec:numerical}, under the conventions of Appendix~\ref{app:reporting}.

\subsubsection{Reduction checks}
\label{app:results-reduction}

We report independent representation checks and fixed-input numerical refinement.
Table~\ref{tab:completion} compares the full-history and reduced formulations under the same tangent constraints.
Table~\ref{tab:mesh} separates the perspective budget allowance, the inner-primal/outer-value gap, and the supplied-LP certificate gap.
The finite-grid pattern is consistent with the local approximation result but does not establish an endpoint rate outside its assumptions.

\begin{table}[tbp]
    \centering
    \caption{Full-history completion program against the reduced program under identical restrictions and tangent grids.
    Dimensions exclude epigraph variables.
    Endpoint difference compares independently assembled primal calculations; outward difference compares their objective bounds.}
    \label{tab:completion}
    \small
    \begin{tabular}{@{}rrrrrr@{}}
        \toprule
        $K$ & $s$ & Completion & Reduced & Endpoint difference & Outward difference \\
        \midrule
        $3$ & $4$ & $42$ & $3$ & $5.55\times 10^{-17}$ & $3.77\times 10^{-15}$ \\
        $5$ & $4$ & $682$ & $5$ & $5.55\times 10^{-17}$ & $7.66\times 10^{-15}$ \\
        $5$ & $6$ & $3110$ & $5$ & $5.55\times 10^{-17}$ & $7.61\times 10^{-15}$ \\
        \bottomrule
    \end{tabular}
\end{table}

\begin{table}[tbp]
    \centering
    \caption{Mesh refinement at fixed population moments.
    Endpoint gap compares an inner-primal value with an outward objective bound; $h^2/4$ bounds the budget error.
    Inner violations are residual diagnostics.}
    \label{tab:mesh}
    \small
    \begin{tabular}{@{}rrrrr@{}}
        \toprule
        $h$ & $h^2/4$ & Endpoint gap & Inner violation & Certificate gap \\
        \midrule
        $1/8$ & $0.00391$ & $0.00781$ & $0.00$ & $5.55\times 10^{-17}$ \\
        $1/16$ & $9.77\times 10^{-4}$ & $0.00195$ & $0.00$ & $1.11\times 10^{-16}$ \\
        $1/32$ & $2.44\times 10^{-4}$ & $4.34\times 10^{-4}$ & $1.73\times 10^{-18}$ & $1.22\times 10^{-14}$ \\
        $1/64$ & $6.10\times 10^{-5}$ & $1.09\times 10^{-4}$ & $0.00$ & $2.44\times 10^{-14}$ \\
        $1/128$ & $1.53\times 10^{-5}$ & $2.71\times 10^{-5}$ & $0.00$ & $6.48\times 10^{-14}$ \\
        \bottomrule
    \end{tabular}
\end{table}

\subsubsection{Input-region comparisons}
\label{app:results-moments}

We compare protected feature collections and joint likelihood regions for the same signed-budget model.
Table~\ref{tab:atomic-kl} changes the coordinatewise input construction, while Table~\ref{tab:joint-regions} reports the separate paired joint-region experiment.
The latter includes paired width differences and numerical stopping diagnostics, so its comparisons do not rely on small differences between independently simulated tables.

\begin{table}[tbp]
    \centering
    \caption{Protected feature set under coordinatewise limits at the arm-zero law, $n=1000$, $M=5$, $\Gamma=0.2$, $500$ paired replications.
    All rows use the same reduced program; the KL rows use coordinatewise Bernoulli--KL limits.}
    \label{tab:atomic-kl}
    \small
    \begin{tabular}{@{}lrr@{}}
        \toprule
        Input-region construction & Mean width & Set coverage \\
        \midrule
        Atomic KL & $0.410$ & $1.00$ \\
        Atomic KL + aggregate & $0.258$ & $1.00$ \\
        Cell KL & $0.257$ & $1.00$ \\
        Cell CP & $0.237$ & $0.998$ \\
        \bottomrule
    \end{tabular}
\end{table}

\begin{table}[tbp]
    \centering
    \caption{Joint-region comparison with $500$ paired replications per setting; the four procedures share datasets within each block, which are distinct from the datasets of Tables~\ref{tab:scientific}--\ref{tab:inference}.
    Difference is width minus Cell CP width on the same dataset, with its paired MCSE.
    Sharp widths are $0.1437$ for the reference law and $0.1033$ under active caps.
    No physical-range fallbacks occurred for any of the four procedures, and the three LP-based procedures recorded no empty confidence programs, unavailable outward certificates, or numerical stopping flags.
    In each setting, zero of $500$ datasets had any procedure fail (95\% MC interval $[0.000,0.007]$).
    The maximum recorded diagnostic endpoint gaps were $9.82\times10^{-5}$ for the reference setting and $9.99\times10^{-5}$ for each active-cap setting.}
    \label{tab:joint-regions}
    \small
    \setlength{\tabcolsep}{4pt}
    \begin{tabular}{@{}lrrrrrr@{}}
        \toprule
        Procedure & \shortstack{Mean\\width} & MCSE & Difference & \shortstack{Paired\\MCSE} & \shortstack{Set\\coverage} & $95\%$ MC interval \\
        \midrule
        \multicolumn{7}{l}{Reference, $n=1000$} \\
        Cell CP & $0.237$ & $9.68\times 10^{-5}$ & $0.00$ & $0.00$ & $0.996$ & $[0.986, 1.00]$ \\
        Joint split & $0.330$ & $0.00125$ & $0.0929$ & $0.00125$ & $1.00$ & $[0.993, 1.00]$ \\
        Joint LR & $0.275$ & $1.08\times 10^{-4}$ & $0.0373$ & $2.47\times 10^{-5}$ & $1.00$ & $[0.993, 1.00]$ \\
        Aggregate CP & $0.213$ & $9.62\times 10^{-5}$ & $-0.0243$ & $1.08\times 10^{-5}$ & $0.972$ & $[0.953, 0.985]$ \\
        \midrule
        \multicolumn{7}{l}{Active caps, $n=1000$} \\
        Cell CP & $0.241$ & $9.54\times 10^{-5}$ & $0.00$ & $0.00$ & $1.00$ & $[0.993, 1.00]$ \\
        Joint split & $0.296$ & $0.00130$ & $0.0547$ & $0.00131$ & $1.00$ & $[0.993, 1.00]$ \\
        Joint LR & $0.238$ & $1.81\times 10^{-4}$ & $-0.00332$ & $1.63\times 10^{-4}$ & $1.00$ & $[0.993, 1.00]$ \\
        Aggregate CP & $0.216$ & $9.71\times 10^{-5}$ & $-0.0254$ & $4.40\times 10^{-6}$ & $0.978$ & $[0.961, 0.989]$ \\
        \midrule
        \multicolumn{7}{l}{Active caps, $n=10000$} \\
        Cell CP & $0.157$ & $4.89\times 10^{-5}$ & $0.00$ & $0.00$ & $0.996$ & $[0.986, 1.00]$ \\
        Joint split & $0.165$ & $4.29\times 10^{-4}$ & $0.00809$ & $4.28\times 10^{-4}$ & $1.00$ & $[0.993, 1.00]$ \\
        Joint LR & $0.146$ & $5.77\times 10^{-5}$ & $-0.0106$ & $1.46\times 10^{-5}$ & $1.00$ & $[0.993, 1.00]$ \\
        Aggregate CP & $0.166$ & $3.08\times 10^{-5}$ & $0.00954$ & $3.97\times 10^{-5}$ & $0.984$ & $[0.969, 0.993]$ \\
        \bottomrule
    \end{tabular}
\end{table}

\subsubsection{Contraction and contrasts}
\label{app:results-contraction}

We report endpoint accuracy at a fixed scientific class and confidence intervals for treatment contrasts.
Table~\ref{tab:accuracy} keeps the oracle-calibrated population interval fixed while varying sample size.
Table~\ref{tab:contrast} combines separately protected arm intervals.
Direct projection over the same product of armwise input regions gives the same contrast endpoints; this experiment does not evaluate a different joint input-region construction.

\begin{table}[tbp]
    \centering
    \caption{Cell CP contraction for the oracle-calibrated class, arm zero, sharp width $0.0853$, $500$ replications per row.
    Plug-in error is the Hausdorff distance with its MCSE.
    There were no confidence-program failures in any setting ($0/500$ per setting; 95\% MC interval $[0.000,0.007]$).
    Plug-in failures numbered $7/500$ at $n=50$ (rate $0.014$; 95\% MC interval $[0.006,0.029]$) and $1/500$ at $n=100$ (rate $0.002$; 95\% MC interval $[0.000,0.011]$), with none thereafter.
    Their physical-range fallbacks remain in the summaries.}
    \label{tab:accuracy}
    \small
    \begin{tabular}{@{}rrrrrr@{}}
        \toprule
        $n$ & Plug-in error & MCSE & Outer width & Enlargement & Set coverage \\
        \midrule
        $50$ & $0.0673$ & $0.00309$ & $0.513$ & $0.428$ & $1.00$ \\
        $100$ & $0.0403$ & $0.00155$ & $0.396$ & $0.310$ & $0.998$ \\
        $250$ & $0.0272$ & $8.50\times 10^{-4}$ & $0.283$ & $0.198$ & $0.998$ \\
        $500$ & $0.0185$ & $5.75\times 10^{-4}$ & $0.225$ & $0.140$ & $1.00$ \\
        $1000$ & $0.0121$ & $3.68\times 10^{-4}$ & $0.184$ & $0.0989$ & $1.00$ \\
        $2500$ & $0.00800$ & $2.50\times 10^{-4}$ & $0.148$ & $0.0625$ & $1.00$ \\
        $10000$ & $0.00395$ & $1.23\times 10^{-4}$ & $0.116$ & $0.0312$ & $1.00$ \\
        \bottomrule
    \end{tabular}
\end{table}

\begin{table}[tbp]
    \centering
    \caption{Contrast intervals from independent $n=1000$ arm samples, $500$ paired replications, error $0.025$ per arm.
    Rescue is the arm-averaged rescue probability.}
    \label{tab:contrast}
    \small
    \begin{tabular}{@{}rlrrrr@{}}
        \toprule
        Rescue & Procedure & Sharp width & Mean width & MCSE & Set coverage \\
        \midrule
        $0.500$ & Aggregate CP & $0.283$ & $0.436$ & $1.31\times 10^{-4}$ & $1.00$ \\
        $0.500$ & Endpoint bootstrap & $0.283$ & $0.414$ & $2.03\times 10^{-4}$ & $0.998$ \\
        $0.500$ & Cell CP & $0.283$ & $0.481$ & $1.34\times 10^{-4}$ & $1.00$ \\
        $0.750$ & Aggregate CP & $0.346$ & $0.490$ & $1.04\times 10^{-4}$ & $1.00$ \\
        $0.750$ & Endpoint bootstrap & $0.346$ & $0.472$ & $1.74\times 10^{-4}$ & $0.998$ \\
        $0.750$ & Cell CP & $0.346$ & $0.531$ & $1.15\times 10^{-4}$ & $1.00$ \\
        \bottomrule
    \end{tabular}
\end{table}

\subsubsection{Estimated propensities}
\label{app:empirical-ps}

We report the effect of baseline propensity estimation on the input region and its confidence projection.
Tables~\ref{tab:ps-linear}--\ref{tab:ps-nonlinear} give widths, set coverage, input coverage, and observed-mean errors; Figure~\ref{fig:propensity-coverage} also includes $N=250$.
Set coverage can occur without simultaneous input coverage, consistent with the one-way inclusion in Corollary~\ref{cor:coverage}.
The forest-weight rows are diagnostics for unaugmented weighting and are not a numerical evaluation of augmented moment inference.

\begin{table}[tbp]
    \centering
    \caption{Linear assignment, $250$ paired datasets per sample size, confidence projection with different input constructions.
    At $N=50$ the correctly specified logistic fit failed in $1/250$ datasets (rate $0.004$; 95\% MC interval $[0.000,0.022]$), giving physical-range fallbacks for both logistic procedures.
    No other failures occurred; every other method--sample-size combination had $0/250$ failures (95\% MC interval $[0.000,0.015]$).
    Input coverage includes physical primitive-region fallbacks.}
    \label{tab:ps-linear}
    \small
    \setlength{\tabcolsep}{4pt}
    \begin{tabular}{@{}rlrrrrr@{}}
        \toprule
        $N$ & Procedure & Width & MCSE & Set cov. & Input cov. & Mean $\mu$ error \\
        \midrule
        $50$ & Oracle EB & $1.00$ & $0.00$ & $1.00$ & $1.00$ & $-0.00143$ \\
        $50$ & Oracle Wald & $0.824$ & $0.00475$ & $0.992$ & $0.140$ & $-0.00143$ \\
        $50$ & Logit naive & $0.868$ & $0.00622$ & $0.984$ & $0.144$ & $-0.00282$ \\
        $50$ & Logit sandwich & $0.715$ & $0.00637$ & $0.960$ & $0.128$ & $-0.00282$ \\
        $50$ & Main-effects sandwich & $0.728$ & $0.00553$ & $0.980$ & $0.136$ & $0.00182$ \\
        $50$ & Forest naive & $0.859$ & $0.00533$ & $0.992$ & $0.152$ & $0.0255$ \\
        $50$ & Forest EB & $1.00$ & $0.00$ & $1.00$ & $1.00$ & $0.0255$ \\
        $1000$ & Oracle EB & $0.479$ & $4.97\times 10^{-4}$ & $1.00$ & $1.00$ & $-0.00167$ \\
        $1000$ & Oracle Wald & $0.317$ & $4.33\times 10^{-4}$ & $1.00$ & $0.888$ & $-0.00167$ \\
        $1000$ & Logit naive & $0.319$ & $4.16\times 10^{-4}$ & $1.00$ & $0.892$ & $-0.00161$ \\
        $1000$ & Logit sandwich & $0.279$ & $2.52\times 10^{-4}$ & $1.00$ & $0.884$ & $-0.00161$ \\
        $1000$ & Main-effects sandwich & $0.281$ & $2.50\times 10^{-4}$ & $1.00$ & $0.880$ & $-0.00139$ \\
        $1000$ & Forest naive & $0.338$ & $0.00120$ & $0.996$ & $0.900$ & $0.0210$ \\
        $1000$ & Forest EB & $1.00$ & $0.00$ & $1.00$ & $1.00$ & $0.0210$ \\
        $10000$ & Oracle EB & $0.226$ & $9.42\times 10^{-5}$ & $1.00$ & $1.00$ & $4.70\times 10^{-4}$ \\
        $10000$ & Oracle Wald & $0.199$ & $8.92\times 10^{-5}$ & $0.996$ & $0.960$ & $4.70\times 10^{-4}$ \\
        $10000$ & Logit naive & $0.199$ & $6.77\times 10^{-5}$ & $0.996$ & $0.964$ & $8.46\times 10^{-6}$ \\
        $10000$ & Logit sandwich & $0.187$ & $6.03\times 10^{-5}$ & $0.992$ & $0.956$ & $8.46\times 10^{-6}$ \\
        $10000$ & Main-effects sandwich & $0.187$ & $6.03\times 10^{-5}$ & $0.992$ & $0.956$ & $1.92\times 10^{-5}$ \\
        $10000$ & Forest naive & $0.200$ & $7.36\times 10^{-5}$ & $0.996$ & $0.984$ & $0.00370$ \\
        $10000$ & Forest EB & $0.389$ & $7.84\times 10^{-5}$ & $1.00$ & $1.00$ & $0.00370$ \\
        \bottomrule
    \end{tabular}
\end{table}

\begin{table}[tbp]
    \centering
    \caption{Nonlinear assignment, $250$ paired datasets per sample size.
    The main-effects fit omits the assignment interactions.
    At $N=50$ the correctly specified logistic fit failed in $8/250$ datasets (rate $0.032$; 95\% MC interval $[0.014,0.062]$), giving physical-range fallbacks for both logistic procedures.
    No other failures occurred; every other method--sample-size combination had $0/250$ failures (95\% MC interval $[0.000,0.015]$).
    Main-effects sandwich at $N=10000$ has target coverage $1.000$ (95\% MC interval $[0.985,1.000]$) despite set coverage $0.008$ (95\% MC interval $[0.001,0.029]$).}
    \label{tab:ps-nonlinear}
    \small
    \setlength{\tabcolsep}{4pt}
    \begin{tabular}{@{}rlrrrrr@{}}
        \toprule
        $N$ & Procedure & Width & MCSE & Set cov. & Input cov. & Mean $\mu$ error \\
        \midrule
        $50$ & Oracle EB & $1.00$ & $0.00$ & $1.00$ & $1.00$ & $-0.0102$ \\
        $50$ & Oracle Wald & $0.859$ & $0.00694$ & $0.968$ & $0.0640$ & $-0.0102$ \\
        $50$ & Logit naive & $0.894$ & $0.00774$ & $0.968$ & $0.0880$ & $-0.0150$ \\
        $50$ & Logit sandwich & $0.748$ & $0.00858$ & $0.940$ & $0.0880$ & $-0.0150$ \\
        $50$ & Main-effects sandwich & $0.739$ & $0.00638$ & $0.976$ & $0.0560$ & $-0.0425$ \\
        $50$ & Forest naive & $0.869$ & $0.00621$ & $0.988$ & $0.0680$ & $-0.0285$ \\
        $50$ & Forest EB & $1.00$ & $0.00$ & $1.00$ & $1.00$ & $-0.0285$ \\
        $1000$ & Oracle EB & $0.587$ & $8.05\times 10^{-4}$ & $1.00$ & $1.00$ & $-0.00245$ \\
        $1000$ & Oracle Wald & $0.347$ & $6.74\times 10^{-4}$ & $1.00$ & $0.884$ & $-0.00245$ \\
        $1000$ & Logit naive & $0.351$ & $7.11\times 10^{-4}$ & $1.00$ & $0.912$ & $-1.86\times 10^{-5}$ \\
        $1000$ & Logit sandwich & $0.295$ & $3.99\times 10^{-4}$ & $1.00$ & $0.892$ & $-1.86\times 10^{-5}$ \\
        $1000$ & Main-effects sandwich & $0.283$ & $2.69\times 10^{-4}$ & $0.936$ & $0.668$ & $-0.0437$ \\
        $1000$ & Forest naive & $0.361$ & $0.00201$ & $1.00$ & $0.904$ & $0.00831$ \\
        $1000$ & Forest EB & $1.00$ & $0.00$ & $1.00$ & $1.00$ & $0.00831$ \\
        $10000$ & Oracle EB & $0.245$ & $1.33\times 10^{-4}$ & $1.00$ & $1.00$ & $9.17\times 10^{-4}$ \\
        $10000$ & Oracle Wald & $0.208$ & $1.23\times 10^{-4}$ & $1.00$ & $0.976$ & $9.17\times 10^{-4}$ \\
        $10000$ & Logit naive & $0.208$ & $9.41\times 10^{-5}$ & $1.00$ & $0.968$ & $5.24\times 10^{-4}$ \\
        $10000$ & Logit sandwich & $0.191$ & $7.66\times 10^{-5}$ & $0.996$ & $0.952$ & $5.24\times 10^{-4}$ \\
        $10000$ & Main-effects sandwich & $0.188$ & $6.87\times 10^{-5}$ & $0.00800$ & $0.00$ & $-0.0423$ \\
        $10000$ & Forest naive & $0.211$ & $1.11\times 10^{-4}$ & $1.00$ & $0.960$ & $0.00604$ \\
        $10000$ & Forest EB & $0.402$ & $1.27\times 10^{-4}$ & $1.00$ & $1.00$ & $0.00604$ \\
        \bottomrule
    \end{tabular}
\end{table}

\begin{figure}[tbp]
    \centering
    \includegraphics[width=0.49\linewidth]{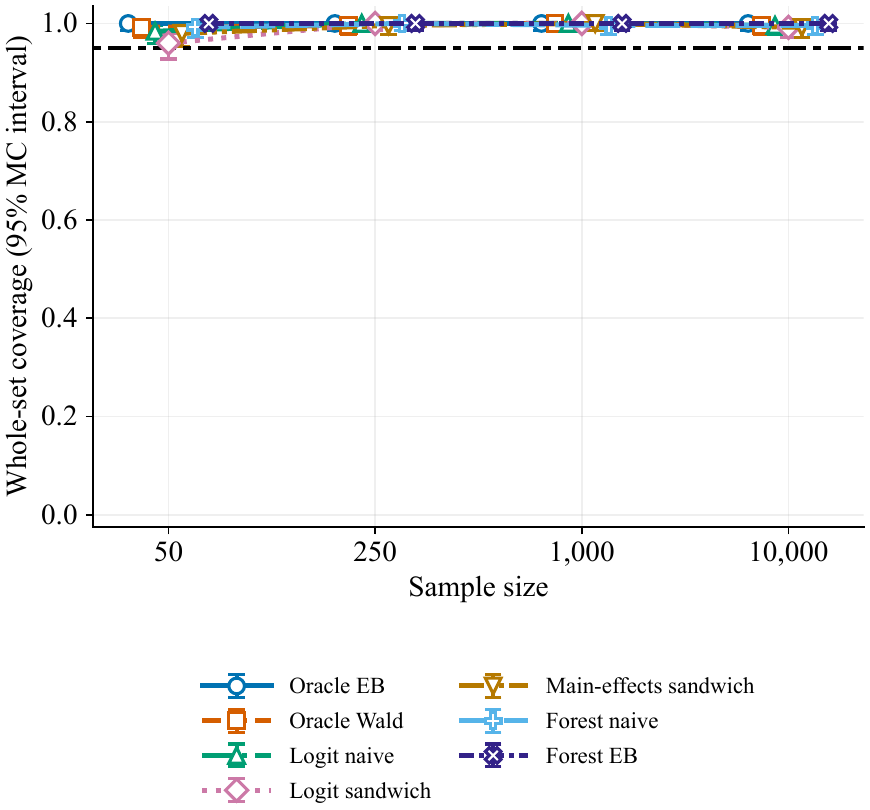}\hfill
    \includegraphics[width=0.49\linewidth]{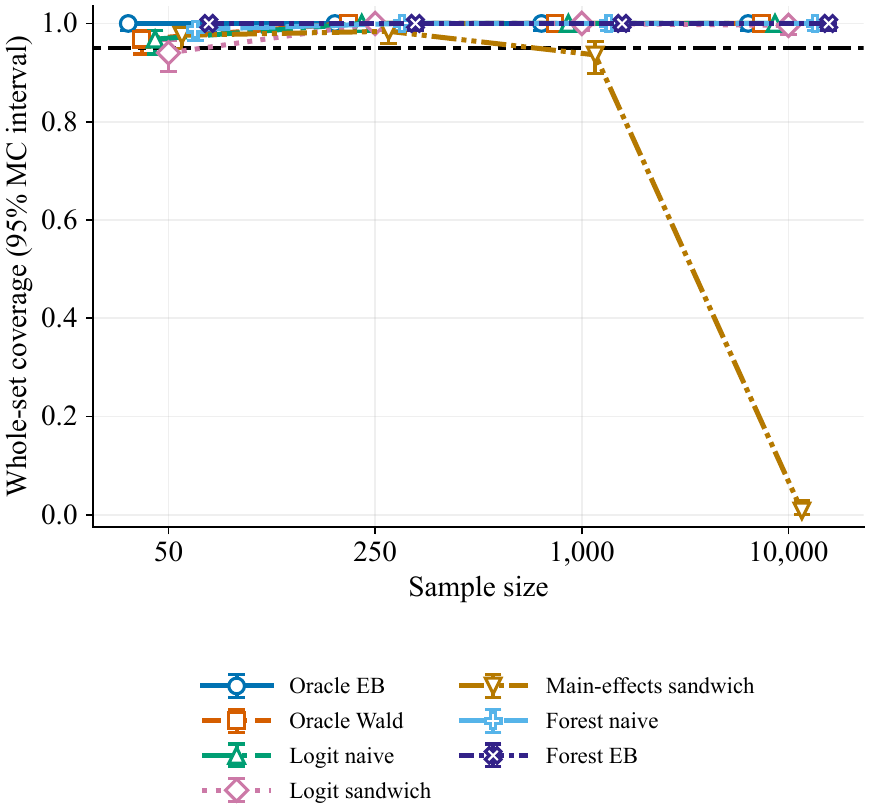}
    \caption{Set coverage under linear (left) and nonlinear (right) assignment for the input constructions of Appendix~\ref{app:estimated-ps}.
    Main-effects sandwich is misspecified on the right.
    Bars are $95\%$ binomial Monte Carlo intervals; the dashed line marks $0.95$.}
    \label{fig:propensity-coverage}
\end{figure}

\subsubsection{Continuous histories and partitions}
\label{app:results-partitions}

We report continuous-history configurations and partition-specific precision.
Table~\ref{tab:continuous} summarizes additional generating laws, while Tables~\ref{tab:partition-oracle} and~\ref{tab:partition} separate each partition's sharp width from its sampling enlargement.
Only the first five baseline coordinates enter the continuous generating equations, and the analysis uses prespecified low-dimensional cells.
The $d_W=50$ configuration evaluates this reduced analysis in the presence of additional recorded covariates.
Reference-batch standard errors (SEs) describe simulation uncertainty in the indicated reference quantities and are separate from trial-replication Monte Carlo errors.
For the four configurations in Table~\ref{tab:continuous}, the lower- and upper-endpoint reference standard errors are $(0.000111,0.000099)$, $(0.000111,0.000096)$, $(0.000103,0.000078)$, and $(0.000088,0.000078)$, respectively.
Cell CP covers the reference sharp interval in all $250$ replications of each configuration (95\% MC interval $[0.985,1.000]$), and the inward and outward endpoint shifts by $2.58$ reference standard errors leave these coverage frequencies unchanged.
For the heterogeneous-cap law in Table~\ref{tab:partition}, the shifts change coverage at $n=10000$, $M=1$ from $0.996$ to $0.992$ outward and $1.000$ inward, with corresponding 95\% MC intervals $[0.971,0.999]$ and $[0.985,1.000]$.
Coverage at the other displayed heterogeneous-cap settings is unchanged by either shift.

\begin{table}[tbp]
    \centering
    \caption{Continuous-history configurations with $250$ replications, $n=1000$, $4K$ visit/risk/exceedance cells, and the signed $\Gamma=0.2$ budget.
    Parentheses give reference-batch standard errors for $\theta_0$; endpoint reference errors and the endpoint-shift diagnostic are reported in the accompanying text.}
    \label{tab:continuous}
    \small
    \begin{tabular}{@{}rrrrrr@{}}
        \toprule
        $K$ & $d_W$ & $\rho$ & Sharp width & Cell CP width & $\theta_0$ (reference SE) \\
        \midrule
        $5$ & $5$ & $0.00$ & $0.142$ & $0.248$ & $0.462$ ($1.13\times 10^{-4}$) \\
        $5$ & $5$ & $0.650$ & $0.141$ & $0.248$ & $0.462$ ($9.34\times 10^{-5}$) \\
        $5$ & $50$ & $0.350$ & $0.141$ & $0.248$ & $0.462$ ($9.23\times 10^{-5}$) \\
        $8$ & $5$ & $0.350$ & $0.142$ & $0.252$ & $0.462$ ($8.57\times 10^{-5}$) \\
        \bottomrule
    \end{tabular}
\end{table}

\begin{table}[tbp]
    \centering
    \caption{Partition-specific population references.
    True norm is the generating-law budget norm $\{\sum_{j:p_j>0}t_j^2/p_j\}^{1/2}$.
    Endpoint MCSEs are reference-batch standard errors.}
    \label{tab:partition-oracle}
    \small
    \begin{tabular}{@{}lrrrrr@{}}
        \toprule
        Outcome law & $M$ & Sharp width & True norm & Lower MCSE & Upper MCSE \\
        \midrule
        Original & $1$ & $0.141$ & $0.146$ & $1.14\times 10^{-4}$ & $1.03\times 10^{-4}$ \\
        Original & $5$ & $0.141$ & $0.146$ & $1.14\times 10^{-4}$ & $1.03\times 10^{-4}$ \\
        Original & $10$ & $0.141$ & $0.147$ & $1.14\times 10^{-4}$ & $1.03\times 10^{-4}$ \\
        Original & $20$ & $0.141$ & $0.147$ & $1.14\times 10^{-4}$ & $1.03\times 10^{-4}$ \\
        Original & $40$ & $0.141$ & $0.147$ & $1.14\times 10^{-4}$ & $1.03\times 10^{-4}$ \\
        Original & $80$ & $0.141$ & $0.147$ & $1.14\times 10^{-4}$ & $1.03\times 10^{-4}$ \\
        Heterogeneous caps & $1$ & $0.141$ & $0.0574$ & $1.90\times 10^{-4}$ & $2.12\times 10^{-4}$ \\
        Heterogeneous caps & $5$ & $0.141$ & $0.0579$ & $1.90\times 10^{-4}$ & $2.06\times 10^{-4}$ \\
        Heterogeneous caps & $10$ & $0.0949$ & $0.0849$ & $1.90\times 10^{-4}$ & $1.76\times 10^{-4}$ \\
        Heterogeneous caps & $20$ & $0.0949$ & $0.0850$ & $1.90\times 10^{-4}$ & $1.76\times 10^{-4}$ \\
        Heterogeneous caps & $40$ & $0.0949$ & $0.0850$ & $1.90\times 10^{-4}$ & $1.76\times 10^{-4}$ \\
        Heterogeneous caps & $80$ & $0.0949$ & $0.0850$ & $1.90\times 10^{-4}$ & $1.76\times 10^{-4}$ \\
        \bottomrule
    \end{tabular}
\end{table}

\begin{table}[tbp]
    \centering
    \caption{Cell CP partition comparison under the heterogeneous-cap law, $250$ paired replications per sample size.
    Coverage and plug-in error use each partition's reference; empty fraction is the mean fraction of cells with zero sample count.}
    \label{tab:partition}
    \small
    \setlength{\tabcolsep}{4pt}
    \begin{tabular}{@{}rrrrrrrr@{}}
        \toprule
        $n$ & $M$ & Sharp width & Outer width & Enlargement & Plug-in error & Empty frac. & Coverage \\
        \midrule
        $250$ & $1$ & $0.141$ & $0.317$ & $0.176$ & $0.0282$ & $0.00$ & $0.992$ \\
        $250$ & $5$ & $0.141$ & $0.338$ & $0.197$ & $0.0285$ & $0.00$ & $0.992$ \\
        $250$ & $10$ & $0.0949$ & $0.349$ & $0.254$ & $0.0276$ & $0.00$ & $0.996$ \\
        $250$ & $20$ & $0.0949$ & $0.361$ & $0.266$ & $0.0277$ & $0.00860$ & $0.996$ \\
        $250$ & $40$ & $0.0949$ & $0.373$ & $0.278$ & $0.0279$ & $0.0807$ & $0.996$ \\
        $250$ & $80$ & $0.0949$ & $0.384$ & $0.289$ & $0.0287$ & $0.257$ & $0.996$ \\
        $1000$ & $1$ & $0.141$ & $0.229$ & $0.0878$ & $0.0135$ & $0.00$ & $0.992$ \\
        $1000$ & $5$ & $0.141$ & $0.240$ & $0.0991$ & $0.0136$ & $0.00$ & $0.996$ \\
        $1000$ & $10$ & $0.0949$ & $0.245$ & $0.151$ & $0.0133$ & $0.00$ & $0.996$ \\
        $1000$ & $20$ & $0.0949$ & $0.252$ & $0.157$ & $0.0133$ & $0.00$ & $1.00$ \\
        $1000$ & $40$ & $0.0949$ & $0.258$ & $0.163$ & $0.0133$ & $2.00\times 10^{-4}$ & $1.00$ \\
        $1000$ & $80$ & $0.0949$ & $0.264$ & $0.169$ & $0.0134$ & $0.00900$ & $1.00$ \\
        $10000$ & $1$ & $0.141$ & $0.169$ & $0.0277$ & $0.00447$ & $0.00$ & $0.996$ \\
        $10000$ & $5$ & $0.141$ & $0.172$ & $0.0319$ & $0.00451$ & $0.00$ & $1.00$ \\
        $10000$ & $10$ & $0.0949$ & $0.161$ & $0.0664$ & $0.00446$ & $0.00$ & $1.00$ \\
        $10000$ & $20$ & $0.0949$ & $0.174$ & $0.0793$ & $0.00446$ & $0.00$ & $1.00$ \\
        $10000$ & $40$ & $0.0949$ & $0.178$ & $0.0833$ & $0.00446$ & $0.00$ & $1.00$ \\
        $10000$ & $80$ & $0.0949$ & $0.180$ & $0.0852$ & $0.00446$ & $0.00$ & $1.00$ \\
        \bottomrule
    \end{tabular}
\end{table}

\subsubsection{Boundary, calibration, and radius diagnostics}
\label{app:results-diagnostics}

We report the separate rare-cell, external-calibration, and radius-sensitivity experiments.
Table~\ref{tab:rare-upper} supplements the zero-count diagnostic with mean upper limits under its stated tail conventions.
Table~\ref{tab:external-precision} varies calibration precision under the idealized source law; its no-calibration row defines a different scientific class.
Table~\ref{tab:radius-grid} and Figure~\ref{fig:sensitivity} distinguish membership of the generating law from inference for the imposed class.

\begin{table}[tbp]
    \centering
    \caption{Mean upper limit for a single rare cell, $n=1000$, $20000$ replications per row, with the lower target endpoint fixed at zero.
    At $np=0.5$ the zero-count frequency is $0.605$ (95\% MC interval $[0.598,0.612]$).
    Wald and bootstrap each have coverage $0.395$ (95\% MC interval $[0.388,0.402]$), while all four bounded-moment procedures cover in $20000/20000$ replications (95\% MC interval $[0.9998,1.0000]$).
    Tail conventions are specified in Appendix~\ref{app:external-rare}.}
    \label{tab:rare-upper}
    \small
    \begin{tabular}{@{}rrrrrrr@{}}
        \toprule
        $np$ & CP & KL & Bernstein & Hoeffding & Wald & Bootstrap \\
        \midrule
        $0.00$ & $0.00368$ & $0.00368$ & $0.0102$ & $0.0429$ & $0.00$ & $0.00$ \\
        $0.500$ & $0.00460$ & $0.00503$ & $0.0120$ & $0.0434$ & $0.00136$ & $9.95\times 10^{-4}$ \\
        $2.00$ & $0.00709$ & $0.00830$ & $0.0160$ & $0.0449$ & $0.00447$ & $0.00370$ \\
        $10.0$ & $0.0182$ & $0.0210$ & $0.0294$ & $0.0529$ & $0.0161$ & $0.0148$ \\
        \bottomrule
    \end{tabular}
\end{table}

\begin{table}[tbp]
    \centering
    \caption{External calibration with $n=1000$ and $500$ trial replications paired across settings.
    Finite sources combine trial Clopper--Pearson and external Bernstein limits under one error allocation; Exact uses oracle targets, and None omits calibration and defines a different scientific class.}
    \label{tab:external-precision}
    \small
    \begin{tabular}{@{}lrr@{}}
        \toprule
        External calibration & Mean width & Set coverage \\
        \midrule
        None & $0.237$ & $1.00$ \\
        $250$ & $0.239$ & $1.00$ \\
        $2000$ & $0.239$ & $1.00$ \\
        $20000$ & $0.223$ & $1.00$ \\
        Exact & $0.184$ & $1.00$ \\
        \bottomrule
    \end{tabular}
\end{table}

\begin{table}[tbp]
    \centering
    \caption{Cell CP radius sensitivity with $500$ replications per sample size and radii paired within datasets.
    Membership refers to the generating full-data law.
    Set coverage concerns the imposed class even when that class excludes the generating no-rescue mean.}
    \label{tab:radius-grid}
    \small
    \setlength{\tabcolsep}{4pt}
    \begin{tabular}{@{}rrlrrrr@{}}
        \toprule
        $n$ & $\Gamma$ & True-law member & Sharp width & Outer width & Set cov. & Target cov. \\
        \midrule
        $50$ & $0.00$ & No & $0.00$ & $0.384$ & $0.996$ & $0.966$ \\
        $50$ & $0.0500$ & No & $0.0359$ & $0.427$ & $0.996$ & $1.00$ \\
        $50$ & $0.100$ & Yes & $0.0718$ & $0.469$ & $0.996$ & $1.00$ \\
        $50$ & $0.200$ & Yes & $0.144$ & $0.554$ & $0.996$ & $1.00$ \\
        $50$ & $0.350$ & Yes & $0.251$ & $0.681$ & $0.996$ & $1.00$ \\
        $1000$ & $0.00$ & No & $0.00$ & $0.0873$ & $0.998$ & $0.0800$ \\
        $1000$ & $0.0500$ & No & $0.0359$ & $0.125$ & $0.998$ & $0.848$ \\
        $1000$ & $0.100$ & Yes & $0.0718$ & $0.162$ & $0.998$ & $1.00$ \\
        $1000$ & $0.200$ & Yes & $0.144$ & $0.237$ & $0.998$ & $1.00$ \\
        $1000$ & $0.350$ & Yes & $0.251$ & $0.350$ & $0.998$ & $1.00$ \\
        $10000$ & $0.00$ & No & $0.00$ & $0.0275$ & $0.994$ & $0.00$ \\
        $10000$ & $0.0500$ & No & $0.0359$ & $0.0639$ & $0.994$ & $0.00$ \\
        $10000$ & $0.100$ & Yes & $0.0718$ & $0.100$ & $0.996$ & $1.00$ \\
        $10000$ & $0.200$ & Yes & $0.144$ & $0.173$ & $0.996$ & $1.00$ \\
        $10000$ & $0.350$ & Yes & $0.251$ & $0.282$ & $0.998$ & $1.00$ \\
        \bottomrule
    \end{tabular}
\end{table}

\begin{figure}[tbp]
    \centering
    \includegraphics[width=0.49\linewidth]{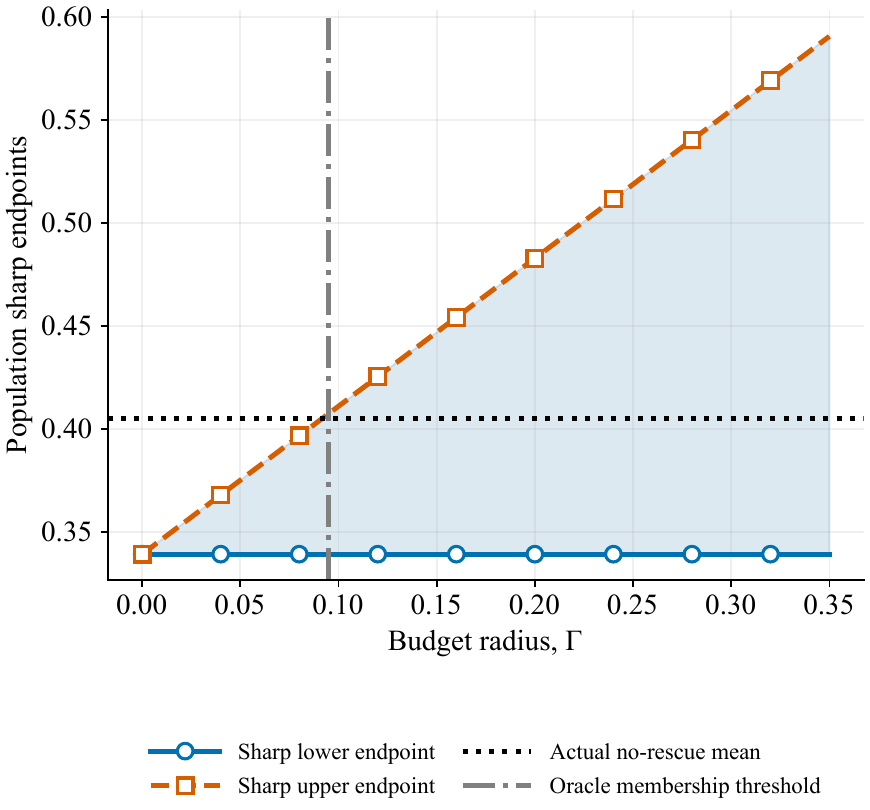}\hfill
    \includegraphics[width=0.49\linewidth]{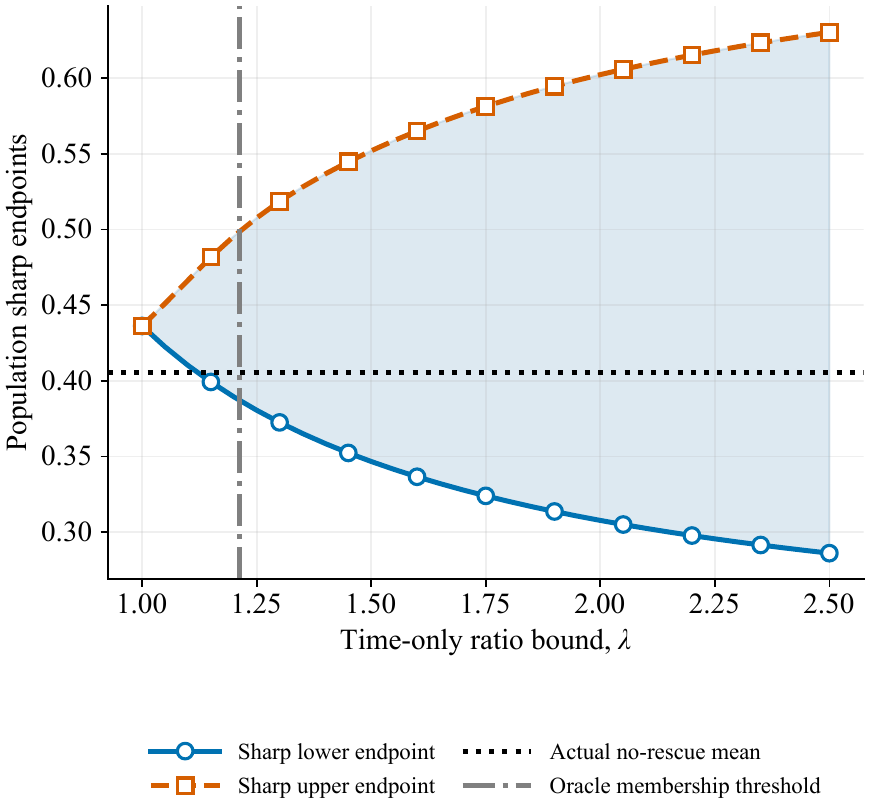}
    \caption{Population sharp endpoints for the signed budget model (left) and the time-only ratio model (right).
    Dotted lines mark the generating no-rescue mean; vertical lines mark the smallest radius and ratio bound compatible with the generating law, $\Gamma=0.095003$ and $\lambda=1.211672$, respectively.}
    \label{fig:sensitivity}
\end{figure}

\subsubsection{Bounded continuous outcomes}
\label{app:results-continuous-outcome}

We compare confidence construction when the observed outcomes are bounded but not binary.
Table~\ref{tab:continuous-outcome} holds the scientific model and population sharp interval fixed while varying sample size and the input-limit construction.
The population sharp width is $0.143673$, so width differences concern sampling and numerical enlargement rather than changes in identification.
Mean widths decrease with sample size for all four procedures.
At every sample size, the endpoint bootstrap gives the shortest mean interval, followed by Cell Wald, Hybrid EB--CP, and Hybrid Hoeffding--CP.
At $n=1000$, the paired width difference between Hybrid Hoeffding--CP and Hybrid EB--CP is $0.03543$ (MCSE $0.00004$), and that between bootstrap and Hybrid EB--CP is $-0.06084$ (MCSE $0.00004$).

Both hybrid procedures cover the sharp interval in all $500$ replications at every sample size.
The table reports exact binomial MC intervals for these frequencies and for the Wald and bootstrap coverages.
These comparisons retain the finite-sample and regular-law qualifications stated in Appendix~\ref{app:inference-comparators}.
The continuous-outcome projections use the fixed grid $h=1/128$ and outward objective certificates, rather than an endpoint-gap stopping criterion.

\begin{table}[tbp]
    \centering
    \caption{Bounded continuous outcomes with beta precision $\phi=10$, arm zero, $M=5$, and signed $\Gamma=0.2$ budget.
    Each sample size uses $500$ paired replications.
    The population sharp width is the binary reference value $0.143673$ in every row.
    Hybrid procedures use CP limits on probability features and the indicated bounded-variable limits on outcome features.
    Every method--sample-size combination has $0/500$ physical-range fallbacks (95\% MC interval $[0.000,0.007]$).
    Each LP-based combination also has $0/500$ empty confidence programs, unavailable outward certificates, and numerical stopping flags, with the same MC interval for each frequency.}
    \label{tab:continuous-outcome}
    \small
    \begin{tabular}{@{}rlrrrr@{}}
        \toprule
        $n$ & Procedure & Mean width & MCSE & Set coverage & $95\%$ MC interval \\
        \midrule
        $250$ & Hybrid EB--CP & $0.374$ & $2.14\times 10^{-4}$ & $1.00$ & $[0.993, 1.00]$ \\
        $250$ & Hybrid Hoeffding--CP & $0.380$ & $1.72\times 10^{-4}$ & $1.00$ & $[0.993, 1.00]$ \\
        $250$ & Cell Wald & $0.224$ & $1.97\times 10^{-4}$ & $1.00$ & $[0.993, 1.00]$ \\
        $250$ & Endpoint bootstrap & $0.187$ & $2.06\times 10^{-4}$ & $0.952$ & $[0.929, 0.969]$ \\
        $1000$ & Hybrid EB--CP & $0.226$ & $9.70\times 10^{-5}$ & $1.00$ & $[0.993, 1.00]$ \\
        $1000$ & Hybrid Hoeffding--CP & $0.262$ & $9.57\times 10^{-5}$ & $1.00$ & $[0.993, 1.00]$ \\
        $1000$ & Cell Wald & $0.184$ & $9.57\times 10^{-5}$ & $0.998$ & $[0.989, 1.00]$ \\
        $1000$ & Endpoint bootstrap & $0.165$ & $1.05\times 10^{-4}$ & $0.950$ & $[0.927, 0.967]$ \\
        $10000$ & Hybrid EB--CP & $0.163$ & $3.01\times 10^{-5}$ & $1.00$ & $[0.993, 1.00]$ \\
        $10000$ & Hybrid Hoeffding--CP & $0.181$ & $3.07\times 10^{-5}$ & $1.00$ & $[0.993, 1.00]$ \\
        $10000$ & Cell Wald & $0.156$ & $3.02\times 10^{-5}$ & $1.00$ & $[0.993, 1.00]$ \\
        $10000$ & Endpoint bootstrap & $0.151$ & $3.28\times 10^{-5}$ & $0.962$ & $[0.941, 0.977]$ \\
        \bottomrule
    \end{tabular}
\end{table}

\clearpage
\bibliographystyle{apalike}
\bibliography{bibliography}

\end{document}